\documentclass[11pt, a4paper]{article}

\usepackage{authblk}
\usepackage[dvipsnames]{xcolor}
\usepackage[dvips,margin=1in,bottom=1in]{geometry}

\usepackage{amsmath}
\usepackage{amsthm}
\usepackage{amssymb}
\usepackage{amsfonts}
\usepackage{mathtools}
\allowdisplaybreaks[1]

\usepackage{enumerate}
\usepackage{enumitem}
\usepackage{float}
\usepackage[title]{appendix}

\usepackage{adjustbox}
\usepackage[nodisplayskipstretch]{setspace}

\usepackage[most]{tcolorbox}

\usepackage{indentfirst}
\usepackage{comment}

\usepackage{booktabs}
\usepackage{array}
\usepackage{colortbl}
\usepackage{subfigure}
\usepackage{makecell}

\usepackage[english]{babel}
\usepackage[T1]{fontenc}
\AtBeginDocument{%
  \DeclareFontShape{T1}{cmr}{m}{scit}{<->ssub*cmr/m/sc}{}%
}

\usepackage[algoruled]{algorithm2e}
\SetAlgorithmName{Protocol}{Protocols} 

\usepackage{tikz}
\usetikzlibrary{quantikz2}

\usepackage{xurl} 
\usepackage[pagebackref,colorlinks,linkcolor=NavyBlue,anchorcolor=blue,citecolor=NavyBlue]{hyperref}

\hypersetup{breaklinks=true}

\usepackage[capitalize,noabbrev]{cleveref}
\crefname{ineq}{inequality}{inequalities}
\creflabelformat{ineq}{#2{\upshape(#1)}#3}
\crefname{fact}{fact}{facts}
\crefname{equation}{equation}{equations}
\crefname{algorithm}{protocol}{protocols} 
\crefname{remark}{remark}{remarks}
\crefname{conjecture}{conjecture}{conjectures}
\crefname{problem}{problem}{problems}

\usepackage{thmtools}
\declaretheorem[style=plain,numberwithin=section]{theorem}
\declaretheorem[style=plain,numberlike=theorem]{lemma,corollary}
\declaretheorem[style=remark,numberlike=theorem]{remark}
\declaretheorem[style=definition,numberlike=theorem]{definition}
\declaretheorem[style=plain,numberlike=theorem]{proposition}

\declaretheorem[style=definition,numberlike=theorem]{problem,conjecture}
\numberwithin{equation}{section}

\makeatletter
\def\@buildmath#1{%
  \expandafter\def\csname bb#1\endcsname{\ensuremath{\mathbb{#1}}}%
  \expandafter\def\csname bf#1\endcsname{\ensuremath{\mathbf{#1}}}%
  \expandafter\def\csname sf#1\endcsname{\ensuremath{\mathsf{#1}}}%
  \expandafter\def\csname cal#1\endcsname{\ensuremath{\mathcal{#1}}}%
  \expandafter\def\csname rm#1\endcsname{\ensuremath{\mathrm{#1}}}%
  \expandafter\def\csname tt#1\endcsname{\ensuremath{\mathtt{#1}}}%
}
\def\@buildmathletters#1{%
  \ifx#1\relax\else
    \@buildmath{#1}%
    \expandafter\@buildmathletters
  \fi
} 
\@buildmathletters ABCDEFGHIJKLMNOPQRSTUVWXYZabcdefghijklmnopqrstuvwxyz\relax
\makeatother

\newcommand{\Ptime}{\textnormal{\textsf{P}}\xspace}

\newcommand{\BQP}{\textnormal{\textsf{BQP}}\xspace}

\newcommand{\NP}{\textnormal{\textsf{NP}}\xspace}

\newcommand{\QIP}{\textnormal{\textsf{QIP}}\xspace}
\newcommand{\QIPtwo}{\textnormal{\textsf{QIP(2)}}\xspace}
\newcommand{\coQIPtwo}{\textnormal{\textsf{co}\text{-}\textsf{QIP(2)}}\xspace}
\newcommand{\qqQAM}{\textnormal{\textrm{qq}\text{-}\textsf{QAM}}\xspace}
\newcommand{\coqqQAM}{\textnormal{\textsf{co}\text{-}\textrm{qq}\text{-}\textsf{QAM}}\xspace}
\newcommand{\PP}{\textnormal{\textsf{PP}}\xspace}
\newcommand{\GapP}{\textnormal{\textsf{GapP}}\xspace}

\newcommand{\AM}{\textnormal{\textsf{AM}}\xspace}
\newcommand{\coAM}{\textnormal{\textsf{coAM}}\xspace}

\newcommand{\countP}{\textnormal{\textsf{\#P}}\xspace}

\newcommand{\SZK}{\textnormal{\textsf{SZK}}\xspace}
\newcommand{\QSZK}{\textnormal{\textsf{QSZK}}\xspace}
\newcommand{\coQSZK}{\textnormal{\textsf{co}\text{-}\textsf{QSZK}}\xspace}
\newcommand{\NIQSZK}{\textnormal{\textsf{NIQSZK}}\xspace}

\newcommand{\QSZKHV}{\texorpdfstring{\textnormal{\textsf{QSZK}\textsubscript{HV}}}\xspace}

\newcommand{\protocol}[2]{{#1}\!\rightleftharpoons\!{#2}}

\newcommand{\PSPACE}{\textnormal{\textsf{PSPACE}}\xspace}

\newcommand{\BQPSPACE}{\textnormal{\textsf{BQPSPACE}}\xspace}

\newcommand{\BQL}{\textnormal{\textsf{BQL}}\xspace}

\newcommand{\NC}{\textnormal{\textsf{NC}}\xspace}

\newcommand{\FEst}{\textnormal{\textsc{F\textsuperscript{2}Est}}\xspace}
\newcommand{\GapFEst}{\textnormal{\textsc{GapF\textsuperscript{2}Est}}\xspace}
\newcommand{\GapFEstlog}{\textnormal{\textsc{GapF\textsuperscript{2}Est}\textsubscript{log}}\xspace}

\newcommand{\GapQSDlog}{\texorpdfstring{\textnormal{\textsc{GapQSD}\textsubscript{log}}}\xspace}

\newcommand{\QSD}{\textnormal{\textsc{QSD}}\xspace}
\newcommand{\GapQSD}{\textnormal{\textsc{GapQSD}}\xspace}
\newcommand{\QSC}{\textnormal{\textsc{QSC}}\xspace}

\newcommand{\QSCMM}{\textnormal{\textsc{QSCMM}}\xspace}

\newcommand{\SD}{\textnormal{\textsc{SD}}\xspace}

\newcommand{\countSH}{\textnormal{\textsc{\#SH}}}
\newcommand{\SAT}{\textnormal{\textsc{SAT}}\xspace}

\DeclarePairedDelimiter\rbra{\lparen}{\rparen}
\DeclarePairedDelimiter\sbra{\lbrack}{\rbrack}
\DeclarePairedDelimiter\cbra{\{}{\}}
\DeclarePairedDelimiter\abs{\lvert}{\rvert}
\DeclarePairedDelimiter\norm{\lVert}{\rVert}
\DeclarePairedDelimiter\ceil{\lceil}{\rceil}

\let\ket\relax\DeclarePairedDelimiter\ket{\lvert}{\rangle}
\let\bra\relax\DeclarePairedDelimiter\bra{\langle}{\rvert}

\newcommand{\ketbra}[2]{\ensuremath{\ket{#1}\!\bra{#2}}}
\renewcommand{\bra}[1]{\langle #1 \rvert}
\renewcommand{\ket}[1]{\lvert #1 \rangle}
\renewcommand{\braket}[3]{\langle #1 | #2 | #3 \rangle}
\newcommand{\innerprod}[2]{\langle #1 | #2 \rangle}

\newcommand{\Tr}{\mathrm{Tr}}
\newcommand{\rank}{\mathrm{rank}}

\newcommand{\td}{\mathrm{T}}

\newcommand{\F}{\mathrm{F}}

\newcommand{\sign}{\mathrm{sgn}}

\newcommand{\acc}{\mathrm{acc}}
\newcommand{\Eval}{\mathrm{Eval}}

\newcommand{\Toffoli}{\textnormal{\textsc{Toffoli}}\xspace}

\newcommand{\SWAP}{\textnormal{\textsc{SWAP}}\xspace}

\newcommand{\Expectation}{%
    \mathop{%
        \mathchoice
            {\vcenter{\hbox{\large$\mathbb{E}$}}}
            {\mathbb{E}}
            {\mathbb{E}}
            {\mathbb{E}}%
    }\displaylimits
}
\renewcommand{\Pr}[1]{\operatorname{Pr}\!\left[#1\right]}

\newcommand{\binset}{\{0,1\}}

\newcommand{\SV}{\mathrm{(SV)}}
\newcommand{\UHH}{U_{\mathrm{HH}}}
\newcommand{\AHH}{A_{\mathrm{HH}}}
\newcommand{\XUhl}{X_{\mathrm{Uhl}}}

\newcommand{\poly} {\operatorname{poly}}
\DeclareMathOperator\polylog{polylog}
\DeclareMathOperator\diag{diag}

\begin{document}
\setlength{\abovedisplayskip}{6pt}
\setlength{\belowdisplayskip}{6pt}

\title{A slightly improved upper bound for quantum statistical zero-knowledge}
\author[1]{Fran\c{c}ois Le Gall\thanks{Email: legall@math.nagoya-u.ac.jp}}
\author[2,1]{Yupan Liu\thanks{Email: yupan.liu@epfl.ch}} 
\author[3]{Qisheng Wang\thanks{Email: QishengWang1994@gmail.com}}
\affil[1]{Graduate School of Mathematics, Nagoya University}
\affil[2]{School of Computer and Communication Sciences, \'Ecole Polytechnique F\'ed\'erale de Lausanne}
\affil[3]{School of Computer Science, Shanghai Jiao Tong University}

\date{}
\maketitle
\pagenumbering{roman}
\thispagestyle{empty}

\begin{abstract}
The complexity class Quantum Statistical Zero-Knowledge ($\mathsf{QSZK}$), introduced by~\hyperlink{cite.Watrous02}{Watrous~(FOCS 2002)} and later refined in~\hyperlink{cite.Wat09}{Watrous~(SICOMP, 2009)}, has the best known upper bound $\mathsf{QIP(2)} \cap \text{co-}\mathsf{QIP(2)}$, which was simplified following the inclusion $\mathsf{QIP(2)} \subseteq \mathsf{PSPACE}$ established in~\hyperlink{cite.JUW09}{Jain, Upadhyay, and Watrous~(FOCS 2009)}. Here, $\mathsf{QIP(2)}$ denotes the class of promise problems that admit two-message quantum interactive proof systems in which the honest prover is typically \textit{computationally unbounded}, and $\text{co-}\mathsf{QIP(2)}$ denotes the complement of $\mathsf{QIP(2)}$.

We slightly improve this upper bound to $\mathsf{QIP(2)} \cap \text{co-}\mathsf{QIP(2)}$ with a \textit{quantum linear-space} honest prover. Specifically, the honest prover uses space linear in the size of the transcript of the original $\mathsf{QSZK}$ proof system. A similar improvement also applies to the upper bound for the non-interactive variant $\mathsf{NIQSZK}$. Our main techniques are \textit{algorithmic} versions of the Holevo--Helstrom measurement and the Uhlmann transform, both implementable in quantum \textit{linear} space, implying polynomial-time complexity in the state dimension, using the recent \textit{space-efficient} quantum singular value transformation of~\hyperlink{cite.LGLW23}{Le Gall, Liu, and Wang~(CC, 2026)}. We also establish that our linear-space bounds for these two tasks are \textit{near-optimal} under the counting version of the Satisfiability Hypothesis ($\#\textsc{SH}$).
\end{abstract}
\newpage
\tableofcontents
\thispagestyle{empty}
\newpage
\pagenumbering{arabic}


\section{Introduction}

Quantum Statistical Zero-Knowledge (\QSZK) is the complexity class of promise problems that admit (single-prover) quantum interactive proof systems with the statistical zero-knowledge property. Intuitively, this property requires that \textit{any} verifier interacting with the honest prover (implicitly on \emph{yes} instances) gains no information from the interaction beyond the validity of the statement.
A weaker variant with \textit{honest} verifiers, denoted by \QSZKHV{},\footnote{For instance, in Graph Non-isomorphism~\cite{GMW91}, where the problem is to decide whether two given graphs $G_0$ and $G_1$ are non-isomorphic, an honest verifier queries only the graphs $G_0$ and $G_1$, whereas an arbitrary verifier may present some graph $G'$ in an attempt to extract additional information.} was first investigated in~\cite{Watrous02}. The resulting class shares most of the basic properties with its classical counterpart \SZK{}~\cite{SV97,GSV98}, including that such proof systems can be parallelized to \textit{two messages}. A few years later, it was shown in~\cite{Wat09} that removing the honest-verifier restriction does not reduce the computational power, establishing the equivalence $\QSZK = \QSZKHV$.\footnote{When the verifier is allowed non-unitary operations, the honest-verifier notion becomes subtle: there is a variant of \QSZKHV{}, introduced in~\cite{CK08}, that contains \PSPACE{}.}

Parallel to the classical \textsc{Statistical Difference Problem} (\SD{}) in~\cite{SV97}, a complete characterization of \QSZK{} was established in~\cite{Watrous02},\footnote{The \QSZK{} containment of $\QSD[\alpha(n),\beta(n)]$ in~\cite{Watrous02} holds only when $\alpha^2(n)-\beta(n) \geq 1/O(\log{n})$, the so-called \textit{polarizing regime}. Slight improvements for the \SZK{} containment of \SD{} beyond this regime were obtained in~\cite{BDRV19} and were later partially extended to the \QSZK{} containment of \QSD{} in~\cite{Liu23}, but the general case remains open.} namely, the \textsc{Quantum State Distinguishability Problem} ($\QSD[\alpha,\beta]$). This promise problem asks whether two quantum states $\rho_0$ and $\rho_1$, whose purifications are prepared by polynomial-size quantum circuits $Q_0$ and $Q_1$, respectively, satisfy $\td(\rho_0,\rho_1) \geq \alpha$ (for \emph{yes} instances) or $\td(\rho_0,\rho_1) \leq \beta$ (for \emph{no} instances), where the trace distance is defined as $\td(\rho_0,\rho_1) \coloneqq \frac{1}{2}\Tr\abs*{\rho_0-\rho_1}$. For convenience, we refer to $\QSD[\alpha(n),\beta(n)]$ with $\alpha(n)-\beta(n) \geq 1/\poly(n)$ as \GapQSD{}. 

Using this complete problem, the best known upper bound for \QSZK{} was shown in~\cite{Watrous02}: 
\[ \QSZK \subseteq \QIPtwo \cap \coQIPtwo \cap \PSPACE. \]
Here, \QIPtwo{} denotes the class of problems admitting two-message quantum interactive proof systems. Notably, the \PSPACE{} containment essentially follows from an $\NC(\poly)$ algorithm for \GapQSD{}. The development of more sophisticated $\NC(\poly)$ algorithms for characterizing quantum interactive proof systems subsequently led to the celebrated result $\QIP=\PSPACE$~\cite{JJUW11}. In particular, an intermediate step proving $\QIPtwo \subseteq \PSPACE$ in~\cite{JUW09} immediately simplified the state-of-the-art upper bound for \QSZK{} to 
\[ \QSZK \subseteq \QIPtwo \cap \coQIPtwo.\]

By contrast, the best known upper bound for the classical counterpart \SZK{} is $\AM\cap\coAM$, as proven in~\cite{Fortnow87,AH91}, where \AM{} denotes the class of problems admitting two-message classical interactive proof systems in which the first message (from the verifier) consists \emph{solely of (public) random coins}. This comparison between the classical and quantum scenarios naturally raises the following intriguing question: 
\begin{problem}
    \label{prob:QSZK-upper-bounds}
    Could the current upper bound for \QSZK{} be improved, even slightly? 
\end{problem}

\subsection{Main results}
\label{subsec:main-results}

In this work, we make progress on \Cref{prob:QSZK-upper-bounds} by restricting the computational power of the honest prover in the proof systems underlying the $\QIPtwo{}\cap\coQIPtwo{}$ upper bound~\cite{Watrous02}, from being computationally \textit{unbounded} to quantum \textit{linear} space (and therefore quantum single-exponential time), as stated in \Cref{thm:GapQSD-in-QIP(2)-informal,thm:GapFEst-in-QIP(2)-informal}.

\begin{theorem}[Informal version of \Cref{thm:GapQSD-in-QIP(2)-bounded-prover}]
    \label{thm:GapQSD-in-QIP(2)-informal}
    \GapQSD{} is in \QIPtwo{} with a quantum linear-space honest prover.
\end{theorem}

The promise problem underlying the \coQIPtwo{} proof system in~\cite{Watrous02} is the \textsc{Quantum State Closeness Problem} (\QSC{}), which is the complement of \QSD{}. This problem is closely related to \FEst{} (to be specified later) via the Fuchs--van de Graaf inequality~\cite{FvdG99}. 

\begin{theorem}[Informal version of \Cref{thm:GapFEst-in-QIP(2)-bounded-prover}]
    \label{thm:GapFEst-in-QIP(2)-informal}
    \GapFEst{} is in \QIPtwo{} with a quantum linear-space honest prover.
\end{theorem}

Here, the promise problem \textsc{Quantum Squared Fidelity Estimation} ($\FEst[\alpha,\beta]$) asks whether $\F^2(\rho_0,\rho_1) \geq \alpha$ for \emph{yes} instances or $\F^2(\rho_0,\rho_1) \leq \beta$ for \emph{no} instances, where the squared fidelity is defined as $\F^2(\rho_0,\rho_1) \coloneqq \rbra*{\Tr\abs*{\sqrt{\rho_0}\sqrt{\rho_1}}}^2$. As with \GapQSD{}, we refer to $\FEst[\alpha(n),\beta(n)]$ with $\alpha(n)-\beta(n)\geq 1/\poly(n)$ as $\GapFEst$. 

\vspace{1em}
We further establish that the linear-space bounds in \Cref{thm:GapQSD-in-QIP(2)-informal,thm:GapFEst-in-QIP(2)-informal} are \emph{near-optimal} under the Counting Satisfiability Hypothesis (\countSH{}), as introduced in~\cite{SH26}. The hypothesis \countSH{} rules out deterministic algorithms that count the satisfying assignments of a size-$L$ \SAT{} formula in time $2^{O(L^{1-\eta})}$ for any fixed $\eta\in(0,1)$. This hypothesis is equivalent to Weak \textsc{\#ETH}, as shown in~\cite[Theorem 2.4]{SH26}, which rules out $2^{O(m^{1-\eta})}\poly(L)$-time algorithms for counting satisfying assignments of $m$-variable $3$-\SAT{} formulas. We state our conditional lower bound (\Cref{thm:near-optimal-space-informal}) and provide a proof sketch in \Cref{subsubsec:near-optimal-space-informal}.

\paragraph{Computational efficiency of the honest prover compared to the general case.}
Approximately implementing the honest prover's strategies in \textit{general} quantum interactive proof systems has been studied in~\cite[Section II.C]{MY23}, which requires quantum \textit{polynomial} space. In contrast, our results (\Cref{thm:GapQSD-in-QIP(2)-informal,thm:GapFEst-in-QIP(2)-informal}) achieve a \textit{polynomial} improvement in space complexity for implementing the honest prover's strategies in \textit{specific} two-message quantum interactive proof systems for \GapQSD{} and \GapFEst{}. Moreover, the corresponding time complexity is \textit{exponentially} improved with respect to the state dimension.\footnote{For a detailed algorithmic comparison, see the discussion in the first paragraph of \Cref{subsec:related-works}.} 

This distinction appears fundamental and challenging to close: even combining the SDP-based approach of~\cite{MY23} with the space-efficient QSVT from~\cite{LGLW23} still requires at least quantum \textit{quadratic} space to approximately implement the
honest prover's strategy in the general case. Further discussion is deferred to \Cref{subsec:open-problems}. 

\paragraph{Implications on \QSZK{} and \NIQSZK{}.}
The main result of this work follows directly from combining \Cref{thm:GapQSD-in-QIP(2)-informal,thm:GapFEst-in-QIP(2)-informal}: 
\begin{corollary}
    \label{corr:QSZK-upper-bound}
    \QSZK{} is in $\QIPtwo{} \cap \coQIPtwo{}$ with a quantum linear-space (and thus single-exponential-time) honest prover. More precisely, the space complexity of the honest prover in the resulting $\QIPtwo{} \cap \coQIPtwo{}$ proof system is linear in the size of the transcript in the original \QSZK{} proof system.
\end{corollary}

In addition to \QSZK{}, a non-interactive variant called \NIQSZK{} was studied in~\cite{Kobayashi03}. In this model, the prover and verifier share prior entanglement (EPR pairs), and only the prover sends a message. As noted in~\cite{KLGN19}, a direct upper bound for \NIQSZK{} is \qqQAM{}, a subclass of \QIPtwo{} in which the verifier's message consists of half of the shared EPR pairs (``quantum public coins''). 
A natural complete problem for \NIQSZK{} is the \textsc{Quantum State Closeness to Maximally Mixed Problem} (\QSCMM{})~\cite{Kobayashi03,BASTS10,CCKV08}, obtained by fixing the state $\rho_0$ in \QSC{} to be the maximally mixed state. 

Noting that $\QSCMM[1/3,2/3]$ is \NIQSZK{}-hard,\footnote{More precisely, if $n$ denotes the number of qubits that the state-preparation circuits act on and $r(n)$ is the number of qubits in the resulting states, then $\QSCMM[1/r,1-1/r]$ is \NIQSZK{}-hard~\cite[Section 8.1]{CCKV08}.} \Cref{thm:GapFEst-in-QIP(2)-informal} also yields the following: 
\begin{corollary}
    \label{corr:NIQSZK-upper-bound}
    \NIQSZK{} is in \qqQAM{} with a quantum linear-space (and thus single-exponential-time) honest prover. More precisely, the space complexity of the honest prover in the resulting \qqQAM{}  proof system is linear in the size of the transcript in the original \NIQSZK{} proof system.
\end{corollary}

To justify that our notion of linearity in \Cref{corr:QSZK-upper-bound,corr:NIQSZK-upper-bound} captures the \emph{nature} of quantum interactive proof systems, we view the inclusions $\QSZK \subseteq \QIPtwo \cap \coQIPtwo$ and $\NIQSZK \subseteq \qqQAM$ primarily as transformations from one class of quantum interactive proof systems (e.g., proof systems with the statistical zero-knowledge property) to another class with distinct structural features (e.g., being parallelizable to two messages, which is stronger than the usual three-message parallelization in~\cite{KW00,KKMV09}). From this perspective, the notion of linearity should be defined in terms of the underlying proof system, which motivates our choice to measure it by the \emph{transcript} of the original proof system, whose size has the same order as the verifier’s message length in the resulting two-message proof systems. 

\subsection{Revisiting the upper bound \texorpdfstring{$\QIPtwo \cap \coQIPtwo$}{QIP(2)∩coQIP(2)}} 
\label{subsec:revisiting-QSZK-upper-bound}

Before explaining the proofs of \Cref{thm:GapQSD-in-QIP(2)-informal,thm:GapFEst-in-QIP(2)-informal}, we first revisit the $\QIPtwo \cap \coQIPtwo$ upper bound established in~\cite{Watrous02}. 

\paragraph{$\GapQSD \in \QIPtwo$.}
The \QIPtwo{} part follows directly from the \QIPtwo{} proof system of \GapQSD{}, as shown in~\cite[Section 4.2]{Watrous02}. This proof system can be seen as a computational version of quantum hypothesis testing (see \Cref{prob:state-discrimination}). In particular, the verifier $\calV$ proceeds as follows:
\begin{enumerate}[label={\upshape(\roman*)}]
    \item $\calV$ sends a quantum state $\rho$, promised to be either $\rho_0$ or $\rho_1$. 
    \item $\calV$ receives a guess $b\in\binset$,\footnote{Although the prover may send an arbitrary quantum state $\sigma$, the verifier can obtain the classical bit $b\in\binset$ by measuring the first qubit of $\sigma$ in the computational basis, ignoring all remaining qubits, and denoting the outcome by $b$.\label{footnote:single-bit-response}} and accepts if $\rho_b$ matches the state $\rho$ exactly. 
\end{enumerate}

Notably, this proof system has classical counterparts, such as the zero-knowledge protocol for Graph Non-isomorphism~\cite{GMW91}. The prover aims to maximize the acceptance probability but can only perform \textit{a two-outcome measurement} on the received state. By the Holevo--Helstrom bound~\cite{Holevo73TraceDist,Helstrom69}, the optimal success probability is $\frac{1}{2} + \frac{1}{2} \td(\rho_0,\rho_1)$, which directly yields an upper bound on the acceptance probability for \emph{no} instances. The optimal measurement $\cbra*{\Pi_0,\Pi_1}$, known as the \textit{Holevo--Helstrom measurement}, has been used to achieve the acceptance probability lower bound for \emph{yes} instances.

\paragraph{$\GapFEst \in \QIPtwo$.}
The \coQIPtwo{} part boils down to the \QIPtwo{} proof system of \GapFEst{}, as presented in~\cite[Section 4.3]{Watrous02}. This proof system can be interpreted as a computational version of the Uhlmann fidelity test (see \Cref{prob:fidelity-test}) and does not have a direct classical counterpart. A natural starting point is testing the closeness between a quantum state $\rho$ and a pure state $\ket{\phi}$, as in~\cite[Exercise 9.2.2]{Wilde13}. The test measures $\rho$ using a two-outcome measurement $\cbra*{\ketbra{\phi}{\phi},I-\ketbra{\phi}{\phi}}$. The test succeeds if the first outcome is obtained, and the success probability $\Tr(\ketbra{\phi}{\phi}\rho)$ coincides exactly with the squared (Uhlmann) fidelity $\F^2(\ketbra{\phi}{\phi},\rho)$. 
In the general case, the verifier $\calV$ proceeds as follows:
\begin{enumerate}[label={\upshape(\roman*)}]
    \item $\calV$ prepares a purification $\ket{\psi_0}$ of $\rho_0$ using the given circuit $Q_0$ and sends the non-output qubits. 
    \item $\calV$ receives these qubits back, which are expected to be transformed by the prover. The modified ``purification'' of $\rho_0$, including the output and received qubits, is denoted by $\rho_{\psi_0}$.  
    \item $\calV$ measures $\rho_{\psi_0}$ using $\cbra*{\ketbra{\psi_1}{\psi_1},I-\ketbra{\psi_1}{\psi_1}}$ and accepts if the first outcome occurs. 
\end{enumerate}

As in the \QIPtwo{} part, the prover aims to maximize the acceptance probability but is restricted to applying \textit{a dimension-preserving quantum channel} $\Phi(\cdot)$ to the received qubits. By a corollary of Uhlmann's theorem~\cite{Uhlmann76} (\Cref{corr:stronger-Uhlmann}), proven in~\cite[Section 4.3]{Watrous02}, the maximum acceptance probability is $\F^2(\rho_0,\rho_1)$, which implies an upper bound on the acceptance probability for \emph{no} instances. The optimal channel $\Phi_\star(\cdot) = U_\star (\cdot) U_\star^{\dagger}$, known as the \textit{Uhlmann transform}, is determined by the chosen purifications of $\rho_0$ and $\rho_1$, and has been used to obtain the acceptance probability lower bound for \emph{yes} instances. 

\subsection{Proof techniques}

We now provide approximate implementations of the honest prover strategies described in \Cref{subsec:revisiting-QSZK-upper-bound}, thereby establishing \Cref{thm:GapQSD-in-QIP(2)-informal,thm:GapFEst-in-QIP(2)-informal}. A central ingredient in our constructions is a \textit{space-efficient} polynomial approximation of the sign function~\cite{LGLW23}. 

\paragraph{The importance of space-efficient polynomial approximations.} The quantum singular value transformation (QSVT) framework~\cite{GSLW19} reduces the design of quantum algorithms to finding good polynomial approximations $P^\mathrm{f}_d$ of a target function $f(x)$ in an appropriate form (``pre-processing''), such as rotation-angle representations~\cite{GSLW19} or coefficients in Chebyshev-type truncations~\cite{MY23,LGLW23}. Moreover, the efficiency of the resulting quantum algorithms is largely determined by the degree $d$.\footnote{The (classical) pre-processing in the time-efficient QSVT~\cite{GSLW19} uses $\poly(d)$ time, so the corresponding space complexity is trivially bounded above by $\poly(d)$.} Importantly, $d$ must be \textit{exponential} in $n$ for QSVT-based approaches to estimating the trace distance~\cite{WGL+22,WZ23} or the fidelity~\cite{WZC+23,GP22,MY23,UNWT25} between quantum states whose  purifications are on $n$ qubits, even to within \textit{constant} precision. This requirement arises from the \textit{square-root}-rank dependence in quantum query complexity lower bounds~\cite{CFMdW10,BKT20,CWZ25}. 

Therefore, to establish \Cref{thm:GapQSD-in-QIP(2)-informal,thm:GapFEst-in-QIP(2)-informal}, we rely on space-efficient polynomial approximations $P_{d'}$ from~\cite{LGLW23}, which can be computed \textit{simultaneously} in $\poly(d)$ time and $O(\log{d})$ space, yielding $2^{O(n)}$ time and $O(n)$ space. Here, the original degree $d$ comes from the time-efficiently computable polynomials~\cite{GSLW19}, and the new degree $d'=O(d)$ is kept explicit to distinguish the space-efficient version.\footnote{To make a polynomial approximation space-efficiently computable, as in~\cite[Section 3.1]{LGLW23}, this increase in degree from $d$ to $d'$ maintains the polynomial approximation error at $O(\epsilon)$, compared with the original approximation error $\epsilon$ associated with $P_d$.} 
 
\subsubsection{Algorithmic Holevo--Helstrom measurement}
\label{subsubsec:algo-HH-meas-informal}

As discussed in \Cref{subsec:revisiting-QSZK-upper-bound}, the honest prover's strategy underlying $\GapQSD\in\QIPtwo$ is the Holevo--Helstrom measurement $\cbra*{\Pi_0,\Pi_1}$, where $\Pi_1 \coloneqq I-\Pi_0$. The decomposition of the trace distance in~\cite[Equation (8)]{WZ23} yields an explicit form of $\Pi_0$ (see \Cref{prop:HH-meas}): 
\[ \td(\rho_0,\rho_1) = \Tr(\Pi_0\rho_0) - \Tr(\Pi_0\rho_1), \quad \text{where } \Pi_0 \coloneqq \frac{I}{2}+\frac{1}{2} \sign^{\SV}\rbra[\Big]{\frac{\rho_0-\rho_1}{2}}.\]

Our first technical contribution is an explicit implementation of $\tilde{\Pi}_0$, which approximately realizes the honest prover's strategy in~\Cref{thm:GapQSD-in-QIP(2)-informal} and ensures that, for \emph{yes} instances, the maximum acceptance probability remains at least $\frac{1}{2} + \frac{1}{2}\td(\rho_0,\rho_1)-2^{-n}$: 

\begin{theorem}[Informal version of \Cref{thm:algo-HH-meas}]
    \label{thm:algo-HH-meas-informal}    
    For quantum states $\rho_0$ and $\rho_1$ specified in \GapQSD{}, whose purifications can be prepared by $n$-qubit polynomial-size quantum circuits $Q_0$ and $Q_1$, the Holevo--Helstrom measurement $\{\Pi_0,\Pi_1\}$ can be approximately implemented in quantum single-exponential time and linear space with additive error $2^{-n}$.  
\end{theorem}

Our approach is inspired by~\cite[Section III.A]{WZ23} (see also~\cite[Section 4.2]{LGLW23}). We start with the one-bit precision phase estimation~\cite{Kitaev95}, commonly referred to as the Hadamard test~\cite{AJL09}, which has previously been used in space-bounded quantum computation~\cite{TS13,FL18}. This procedure enables an explicit implementation of a two-outcome measurement $\cbra*{\Pi,I-\Pi}$, where $\Pi = (I+U)/2$, such that the acceptance probability is $\Tr(\Pi\rho)$, provided that the unitary $U$ can be \emph{approximately} implemented via a block-encoding.\footnote{Following~\cite[Lemma 9]{GP22}, given a block-encoding of a linear operator, the Hadamard test naturally extends to implement $\Pi = (I+A)/2$ with acceptance probability $\operatorname{Re}(\Tr(\Pi\rho))$.} 

To achieve this, we adopt the \textit{space-efficient} quantum singular value transformation~\cite{LGLW23}, specifically employing a polynomial approximation $P_{d'}^\sign$ of the sign function. 
Our explicit implementation of $\tilde{\Pi}_0$ is then accomplished as follows:
\begin{enumerate}[label={\upshape(\arabic*)}]
    \item Using the linear-combinations-of-unitaries technique in~\cite{BCC+15,GSLW19} (see also the space complexity analysis in~\cite[Lemma 3.22]{LGLW23}), one can implement an exact block-encoding of $(\rho_0-\rho_1)/2$, namely $\bra{\bar{0}} U_{(\rho_0-\rho_1)/2} \ket{\bar{0}} = (\rho_0-\rho_1)/2$, in quantum $O(n)$ space. 
    \item Using the space-efficient QSVT associated with the sign function~\cite[Corollary 3.25]{LGLW23}, a block-encoding of $\sign^{\SV}\rbra[\big]{\frac{\rho_0-\rho_1}{2}}$ can be approximately implemented in quantum $O(n)$ space. 
\end{enumerate}

The proof of \Cref{thm:algo-HH-meas-informal} is then completed by analyzing the errors introduced by the polynomial approximation $P_{d'}^\sign$ and the associated space-efficient QSVT implementation, which together accumulate to the desired bound of $2^{-n}$, as detailed in \Cref{subsec:algo-HHmeas-proof}. 

\subsubsection{Algorithmic Uhlmann transform}
\label{subsubsec:algo-Uhlmann-informal}

As discussed in \Cref{subsec:revisiting-QSZK-upper-bound}, the honest prover's strategy for $\GapFEst\in\QIPtwo$ is given by the Uhlmann transform $\Phi_\star (\cdot) = U_\star (\cdot) U_\star^{\dagger}$. Let $\ket{\psi_0}$ and $\ket{\psi_1}$ be purifications of the quantum states $\rho_0$ and $\rho_1$ on register $\sfA$, defined on two registers $(\sfA,\sfR)$, where $\sfR$ serves as the reference register. An explicit form of $U_\star$ is provided implicitly in~\cite[Lemma 6]{Jozsa94} (see \Cref{lemma:Uhlmann-transform}), yielding an alternative expression of the squared Uhlmann fidelity:
\[ \F^2(\rho_0,\rho_1) = \abs[\big]{\bra{\psi_0} \rbra[\big]{I^{\sfA} \!\otimes\! U_{\star}^{\sfR}} \ket{\psi_1}}^2, \quad\text{where } U_{\star} \coloneqq \sign^\SV \rbra*{\Tr_{\sfA} \rbra*{ \ketbra{\psi_0}{\psi_1} }}. \]

Our second technical contribution is an explicit implementation of $\Phi_\star(\cdot)$, which approximately realizes the honest prover's strategy in \Cref{thm:GapFEst-in-QIP(2)-informal}. This implementation guarantees that, for \textit{yes} instances, the maximum acceptance probability remains at least $\F^2(\rho_0,\rho_1)-2^{-n}$: 

\begin{theorem}[Informal version of \Cref{thm:algorithmic-Uhlmann-transform}]
    \label{thm:algo-Uhlmann-transform-informal}    
    For quantum states $\rho_0$ and $\rho_1$ specified in \GapFEst{}, whose purifications can be prepared by $n$-qubit polynomial-size quantum circuits $Q_0$ and $Q_1$, the Uhlmann transform $\Phi_\star(\cdot)$ can be approximately implemented in quantum single-exponential time and linear space with additive error $2^{-n}$. 
\end{theorem}

Our approach is inspired by~\cite[Section 5.1]{UNWT25}. Analogous to \Cref{subsubsec:algo-HH-meas-informal}, we aim to use the space-efficient QSVT associated with the sign function, as established in~\cite[Section 3]{LGLW23}, corresponding to the space-efficient polynomial approximation $P_{d'}^{\sign}$. A technical challenge is to obtain an \textit{exact} block-encoding of 
\[\XUhl \coloneqq \Tr_{\sfA} \rbra*{ \ketbra{\psi_0}{\psi_1} }.\] 
A straightforward approach for realizing the partial trace is to contract $\abs{\sfA}=\log\dim(\calH_\sfA)$ EPR pairs, yielding only an exact encoding of $X_\mathrm{Uhl}/\!\dim(\calH_\sfA)$, as shown in~\cite[Section II.D]{MY23}. Handling this normalization factor $\dim(\calH_\sfA)$ requires additional effort and leads to an implementation that is both less efficient and conceptually more involved. Notably, an exact block-encoding $W$ of $\XUhl$ was recently proposed in~\cite[Section 5.1]{UNWT25}. Leveraging this key ingredient, our explicit implementation of $\Phi_\star(\cdot)$ proceeds as follows:
\begin{enumerate}[label={\upshape(\arabic*)}]
    \item Following~\cite[Section 5.1]{UNWT25} (see \Cref{lemma:Uhlmann-transform-block-encoding}), one can implement an exact block-encoding $W$ of $\Tr_{\sfA} \rbra*{ \ketbra{\psi_0}{\psi_1} }$, namely $\bra{\bar{0}} W \ket{\bar{0}} = \Tr_{\sfA} \rbra*{ \ketbra{\psi_0}{\psi_1} }$, using quantum $O(n)$ space. 
    \item Using the space-efficient QSVT associated with the sign function~\cite[Corollary 3.25]{LGLW23}, a block-encoding of $\sign^{\SV}\rbra*{\Tr_{\sfA} \rbra*{ \ketbra{\psi_0}{\psi_1} }}$ can be approximately implemented in quantum $O(n)$ space. 
\end{enumerate}

Similar to \Cref{subsubsec:algo-HH-meas-informal}, the proof of \Cref{thm:algo-Uhlmann-transform-informal} is completed by analyzing the errors introduced by the polynomial approximation $P_{d'}^\sign$ and the associated space-efficient QSVT implementation. These errors combine to the desired bound of $2^{-n}$, as elaborated in \Cref{subsec:algo-Uhlmann-proof}.

\subsubsection{Near-optimality of the linear-space bounds under \texorpdfstring{\countSH{}}{\#SH}}
\label{subsubsec:near-optimal-space-informal}

We also show that the space bounds in \Cref{thm:algo-HH-meas-informal,thm:algo-Uhlmann-transform-informal} are \emph{near-optimal} under \countSH{}:

\begin{theorem}[Informal version of \Cref{thm:near-optimal-space}]
    \label{thm:near-optimal-space-informal}
    Under \countSH{}, for every fixed $\eta\in(0,1)$, neither the Holevo--Helstrom measurement nor the Uhlmann transform admits an $O(n^{1-\eta})$-space quantum circuit implementation with the same additive error $2^{-n}$ as in \Cref{thm:algo-HH-meas-informal,thm:algo-Uhlmann-transform-informal}. Here, $n$ denotes the number of qubits in the state-preparation circuits, and the space bound includes both the classical workspace used to compute the quantum circuit description and the qubits used in the quantum circuit implementation.
\end{theorem}

Our proof starts by constructing, for each of the two algorithmic tasks, state-preparation circuits $Q_0$ and $Q_1$ acting on $n=\Theta(L+1)$ qubits from a size-$L$ \SAT{} formula $\phi$ and a threshold $t$. The circuits are constructed so that the sign of $\#\SAT(\phi)-t-1/2$ can be determined with bounded error by applying the implemented Holevo--Helstrom measurement to $\ket0$, or by applying the Uhlmann transform to one half of an EPR pair and then performing a Bell measurement. Thus, either implementation yields a quantum test for whether $\#\SAT(\phi)>t$.

Now suppose that either task admits an $O(n^{1-\eta})$-space circuit implementation. By configuration counting and standard state-vector simulation (e.g.,~\cite[Section 1.3, Page 8]{AC17}), each decision-oracle query can be answered deterministically in time $2^{O(n^{1-\eta})}$. The \emph{counting-to-decision reduction} of~\cite[Theorem 4.14]{Fortnow97} then recovers $\#\SAT(\phi)$ by binary search, yielding a deterministic $2^{O(L^{1-\eta})}$-time algorithm for $\#\SAT$, contradicting \countSH{}. 

\subsection{Discussion and open problems}
\label{subsec:open-problems}

\paragraph{Improving upper bounds for \QSZK{}.}
The main open problem is to further improve the upper bounds for \QSZK{} and \NIQSZK{} beyond our results (\Cref{corr:QSZK-upper-bound,corr:NIQSZK-upper-bound}). Since the best known upper bound for the classical counterpart \SZK{} is $\AM\cap\coAM$, as established in~\cite{Fortnow87,AH91}, this inclusion naturally motivates the following question: 

\begin{enumerate}[label={\upshape(\alph*)}]
    \item Could the quantum upper bound for \QSZK{} be improved to a subclass of \QIPtwo{} defined in terms of ``public coins'' quantum interactive proof systems~\cite{MW05,KLGN19}, such as $\qqQAM \cap \coqqQAM$?
\end{enumerate}

A more intriguing question concerns the \textit{classical} upper bound for \QSZK{}, whose best known bound is \PSPACE{} and is believed ``almost certainly can be improved'' in~\cite[Section 7]{Watrous02}:

\begin{enumerate}[label={\upshape(\alph*)}]
    \setcounter{enumi}{1}
    \item \label{open-problem:QSZK-classical-upper-bound} Could the classical upper bounds for \QSZK{} and \NIQSZK{} be improved to any subclass of \PSPACE{}? 
\end{enumerate}

As noted at the beginning of this section, the classical upper bound \PSPACE{} for complexity classes ranging from \QSZK{} to \QIP{}~\cite{Watrous02,JUW09,JJUW11} is obtained via $\NC(\poly)$ algorithms for the corresponding problems. Consequently, making progress on Question~\ref{open-problem:QSZK-classical-upper-bound} likely requires techniques that go beyond this paradigm.

\paragraph{Improving the computational efficiency of the honest prover.}
As noted in \Cref{subsec:main-results}, a na\"ive approach consists of combining the method underlying~\cite[Theorem II.4]{MY23} with the space-efficient QSVT from~\cite{LGLW23}. Let $\omega(\calV)$ denote the maximum acceptance probability of the quantum interactive proof system $\protocol{\calP}{\calV}$. Informally, this combination yields a quantum algorithm that computes $\omega(\calV)$ to within constant precision and simultaneously produces an SDP solution specifying the associated quantum states. This algorithm requires $O(n)$ iterations on a block-encoding that initially acts on $O(n)$ qubits. The honest prover's strategy is then approximately implemented via the algorithmic Uhlmann transform constructed from these states. 
Since the number of required ancillary qubits in this algorithm eventually grows to $O(n^2)$, the resulting algorithm still requires at least \textit{quadratic} quantum space.\footnote{See also the discussion in~\cite[Section 1.6]{LGLW23}.}

Noting that the notion of honest-prover efficiency in our results (\Cref{thm:GapQSD-in-QIP(2)-informal,thm:GapFEst-in-QIP(2)-informal}) appears specifically tailored for \textit{two-message} quantum interactive proof systems, the distinction between our results and the general case suggests the following question:  
\begin{enumerate}[label={\upshape(\alph*)}]
    \setcounter{enumi}{2}
    \item \label{open-problem:prover-efficiency} Could the honest prover's strategy in any two-message quantum interactive proof system be approximately implemented in quantum \emph{linear} space, meaning that the space complexity of the algorithmic implementation scales \emph{linearly} with the number of qubits on which the verifier's message-preparing circuit acts?  
\end{enumerate}

A natural starting point for Question \ref{open-problem:prover-efficiency} is to revisit the \qqQAM{} containment of the \textsc{Close Image to Totally Mixed Problem} (\textsc{CITM})~\cite{KLGN19}.\footnote{It is worth noting that the \qqQAM{} containment of $\textsc{CITM}[a,b]$, as stated in~\cite[Lemma 4.1]{KLGN19}, holds only for the constant parameter regime $(1-a)^2>1-b^2$. This is because the underlying proof system in~\cite[Figure 2]{KLGN19} is essentially designed for the closeness testing problem associated with $\max_\sigma\F^2(\Phi(\sigma),(I/2)^{\otimes r})$.} Here, as a \qqQAM{}-hard problem, \textsc{CITM} is a generalization of \QSCMM{}, defined in terms of $\min_{\sigma}\td(\Phi(\sigma),(I/2)^{\otimes r})$, where the quantum channel $\Phi(\cdot)$ can be implemented by a polynomial-size mixed-state quantum circuit. 

Noting that the proof system in~\cite[Figure 2]{KLGN19} is structurally similar to the \QIPtwo{} containment of \GapFEst{}, one might expect that \Cref{thm:GapFEst-in-QIP(2)-informal,thm:algo-Uhlmann-transform-informal} extend naturally to this more general setting. 
However, the honest prover now also needs to construct a nearly optimal $\tilde{\sigma}_\star$ satisfying 
\[\td(\Phi(\tilde{\sigma}_\star),(I/2)^{\otimes r}) \approx_{\epsilon} \min_{\sigma}\td(\Phi(\sigma),(I/2)^{\otimes r}).\]
It remains unclear how to achieve such a construction in quantum linear space, and this difficulty constitutes a technical barrier to resolving Question~\ref{open-problem:prover-efficiency}. Notably, an affirmative answer to that question would yield a tighter characterization of \QIPtwo{}. 

\subsection{Related works}
\label{subsec:related-works}

An approximate implementation of the Uhlmann transform was previously studied in~\cite[Section II.D]{MY23} under the name ``Algorithmic Uhlmann's Theorem'', in the context of unitary-synthesis complexity classes (e.g., $\mathsf{unitaryPSPACE}$; see also~\cite{BEM+23}). The central distinction between the prior construction in~\cite[Theorem II.5]{MY23}\footnote{For the formal statement, please refer to Theorem 7.4 in the arXiv version of~\cite{MY23}.} and \Cref{thm:algorithmic-Uhlmann-transform} (whose informal version is \Cref{thm:algo-Uhlmann-transform-informal}) is that our construction achieves an \textit{exponentially} improved time complexity when measured in the state dimension $N\coloneq 2^n$. This improvement arises because our construction requires only quantum \textit{linear} space, whereas theirs requires quantum \textit{polynomial} space. As a consequence, our resulting time complexity is $2^{O(n)}=\poly(N)$, while theirs is $2^{p(n)} = N^{q(n)}$ for some functions $p(n)$ and $q(n)\coloneq p(n)/n$ that are both polynomial in $n$. 

Beyond the algorithmic perspective, it is worth noting that a \textit{stability} result of the Uhlmann transform, referred to as ``robust rigidity'', has been recently investigated in~\cite{BMY25}. 

\paragraph{Other notions of the honest-prover efficiency.}
A natural, though folklore, notion of honest-prover efficiency is that of \textit{in-class} interactive proofs, formalized in~\cite[Definition 1]{GKL19}. This notion means that for any promise problem in a complexity class \textsc{C}, there exists a proof system $\protocol{P}{V}$ such that the verifier decides the problem and the honest prover's strategy can be (approximately) implemented in \textsc{C}. This notion applies to complexity classes such as $\Ptime^\countP$ and $\PSPACE$~\cite{LFKN92,Shamir92} via the sum-check protocol, as well as to an intermediate class \textsf{PreciseQCMA}~\cite{GKL19}.\footnote{The equivalence of \textsf{PreciseQCMA} and $\NP^\PP$ is established in~\cite{MN17,GSS+22}.} The same notion naturally extends to other settings, including in-class space-bounded classical interactive proofs for \Ptime{}~\cite{GKR15}, in-class quantum interactive proofs for \BQPSPACE{}~\cite{MY23}, and in-class streaming proofs for \BQL{}~\cite{GRZ23}. 

A more quantitative, practically  motivated notion is that of \textit{doubly-efficient} interactive proofs (see the survey~\cite{Goldreich18}), in which a polynomial-time (honest) prover ideally delegates the computation to an almost-linear-time verifier via interactions, with~\cite{GKR15} serving as a canonical example and subsequent improvements in~\cite{RRR16}.  

\paragraph{Recent developments.}
More than ten months after the release of our work, classical and quantum polarization in the natural regime, originally listed as an open problem in~\cite{SV97} nearly thirty years ago,\footnote{The open problem on classical polarization listed in~\cite{SV97} considers only constant thresholds $\alpha$ and $\beta$, while the new result~\cite{HL26} works even with promise gap $1/O(\log{n})$.} was resolved in~\cite{HL26}. In addition, our algorithmic Holevo--Helstrom measurement (\Cref{thm:algo-HH-meas}) has an application in cryptography and was used in~\cite{HOZ26} to prove that if \textsf{EFI} pairs exist, then $\BQP \subsetneq \PSPACE$. 


\section{Preliminaries}
\label{sec:preliminaries}

We assume that the reader has a basic familiarity with quantum computation and quantum information theory. For an introduction, the textbooks by~\cite{NC10,deWolf19} offer accessible starting points. For a more comprehensive overview of quantum complexity theory, see~\cite{Watrous08}; for a survey specifically focused on quantum interactive proof systems, refer to~\cite{VW16}. 

For convenience, we adopt the following notations throughout this work: (i) the symbol $\ket{\bar{0}}$ denotes an $a$-qubit state $\ket{0}^{\otimes a}$ for $a>1$. (ii) the logarithmic function $\log$ is taken to be base-$2$ by default, i.e., $\log(x)\coloneqq\log_2(x)$ for all positive real numbers $x$. (iii) hidden log factors are suppressed using the notation $\widetilde{O}(f)\coloneqq O(f \polylog(f))$. 

\subsection{Schatten norm and a matrix H\"older inequality}
For $1\leq p \leq \infty$, the Schatten $p$-norm of a matrix $A$ is defined by 
\[\norm{A}_{p} \coloneqq \rbra[\big]{\Tr\rbra{\abs{A}^{p}}}^{1/p}, \quad\text{where } \abs{A} \coloneqq \sqrt{A^\dagger A}.\] 
When $p=1$, this norm reduces to the \textit{trace norm} $\norm{A}_1=\Tr|A|$. When $p=\infty$, this norm becomes the \textit{operator norm}, given by $\norm{A} \coloneqq \norm{A}_{\infty} = \sigma_{\max}(A),$ where $\sigma_{\max}(A)$ denotes the largest singular value of $A$.  We also need the following version of the matrix H\"older inequality:
\begin{lemma}[H\"older inequality for Schatten norms, adapted from~{\cite[Equation 1.174]{Watrous18}}]
    \label{lemma:matrix-Holder-inequality}
    For each $p\in[1,\infty]$, let $q\in[1,\infty]$ satisfy $\frac{1}{p}+\frac{1}{q}=1$. For every matrix $A$, it holds that the Schatten $p$-norm and $q$-norm are \emph{dual}. Consequently, for all matrices $B$,
    \[ \abs*{\Tr\rbra[\big]{B^\dagger A}} \leq \norm{A}_p\norm{B}_q. \]
\end{lemma}

\subsection{Closeness measures for quantum states and the corresponding testing problems}
We begin by defining quantum states. A square matrix $\rho$ is called a \textit{quantum state} if $\rho$ is positive semi-definite and has unit trace, that is, $\Tr(\rho)=1$. 

\paragraph{Closeness measures for quantum states.}
We then introduce two measures of closeness between quantum states that are the focus of this work: 
\begin{definition}[Trace distance]
    Let $\rho_0$ and $\rho_1$ be two (possibly mixed) quantum states. The trace distance between $\rho_0$ and $\rho_1$ is defined by
    \[ \td(\rho_0,\rho_1)\coloneqq\frac{1}{2}\Tr|\rho_0-\rho_1|=\frac{1}{2}\norm*{\rho_0-\rho_1}_1. \]
\end{definition}

\begin{definition}[Uhlmann fidelity]
    Let $\rho_0$ and $\rho_1$ be two (possibly mixed) quantum states. The (Uhlmann) fidelity between $\rho_0$ and $\rho_1$ and its square are defined by
    \[ \F(\rho_0,\rho_1)\coloneqq\Tr|\sqrt{\rho_0}\sqrt{\rho_1}|=\norm*{\sqrt{\rho_0}\sqrt{\rho_1}}_1 \quad\text{and}\quad \F^2(\rho_0,\rho_1)\coloneqq \norm*{\sqrt{\rho_0}\sqrt{\rho_1}}_1^2.\]
\end{definition}

The trace distance reaches its minimum value of $0$ when $\rho_0$ equals $\rho_1$, while the (squared) fidelity attains its maximum of $1$. Conversely, the trace distance reaches its maximum value of $1$ when the supports of $\rho_0$ and $\rho_1$ are orthogonal, and the squared fidelity attains its minimum of $0$. Importantly, the trace distance and the (squared) fidelity are related by the well-known Fuchs--van de Graaf inequalities:
\begin{lemma}[Trace distance vs.~fidelity, adapted from~\cite{FvdG99}]
\label{lemma:traceDist-vs-fidelity}
    Let $\rho_0$ and $\rho_1$ be two (possibly mixed) quantum states. Then, it holds that
    \[1-\F(\rho_0,\rho_1) \leq \td(\rho_0,\rho_1) \leq \sqrt{1-\F^2(\rho_0,\rho_1)}.\]
\end{lemma}

Furthermore, the operational interpretations of the trace distance and the (squared) fidelity, namely the Holevo--Helstrom bound~\cite{Holevo73TraceDist,Helstrom69} and Uhlmann's theorem~\cite{Uhlmann76,Jozsa94}, together with the corresponding \textit{optimal} operations that achieve these maxima, play a central role in this work. To keep the technical sections self-contained, we defer the formal statements of these results to \Cref{sec:algo-HH-meas,sec:algo-Uhlmann}, respectively. 

\paragraph{Closeness testing of quantum states via state-preparation circuits.}
Next, we introduce two promise problems defined with respect to the trace distance: 
\begin{definition}[\textsc{Quantum State Distinguishability}, {$\QSD[\alpha,\beta]$}, adapted from~{\cite[Section 3.3]{Watrous02}}]
    \label{def:QSD}
    Let $Q_0$ and $Q_1$ be polynomial-size quantum circuits acting on $n$ qubits, each with $r$ designated output qubits. For $b\in\binset$, let $\rho_b$ denote the quantum state obtained by applying $Q_b$ to the initial state $\ket{0}^{\otimes n}$ and tracing out the non-output qubits. Let $\alpha(n)$ and $\beta(n)$ be efficiently computable functions. The problem is to decide whether: 
    \begin{itemize}
        \item \textbf{Yes:} A pair of quantum circuits $(Q_0,Q_1)$ such that $\td(\rho_0,\rho_1) \geq \alpha(n)$; 
        \item \textbf{No:} A pair of quantum circuits $(Q_0,Q_1)$ such that $\td(\rho_0,\rho_1) \leq \beta(n)$. 
    \end{itemize}
\end{definition}

\begin{definition}[\textsc{Quantum State Closeness}, {$\QSC[\beta,\alpha]$}, adapted from~{\cite[Section 3]{Kobayashi03}}]
    \label{def:QSC}
    Let $Q_0$ and $Q_1$ be quantum circuits defined as in \Cref{def:QSD}, and let $\rho_0$ and $\rho_1$ denote the corresponding quantum states obtained from these circuits. Let $\alpha(n)$ and $\beta(n)$ be efficiently computable functions. The problem is to decide whether: 
    \begin{itemize}
        \item \textbf{Yes:} A pair of quantum circuits $(Q_0,Q_1)$ such that $\td(\rho_0,\rho_1) \leq \beta(n)$; 
        \item \textbf{No:} A pair of quantum circuits $(Q_0,Q_1)$ such that $\td(\rho_0,\rho_1) \geq \alpha(n)$.  
    \end{itemize}
\end{definition}

It is evident that \QSC{} is the complement of \QSD{}. Beyond \Cref{def:QSD,def:QSC}, this work also focuses on two additional closeness testing problems: 
\begin{enumerate}[label={\upshape(\arabic*)}]
    \item An important special case of \Cref{def:QSC} is the \textsc{Quantum State Closeness to Maximally Mixed State} (\QSCMM). This problem arises when $\rho_0$ is fixed to be the $r(n)$-qubit maximally mixed state $(I/2)^{\otimes r}$, and $Q_0$ is the circuit that prepares $r(n)$ EPR pairs acting on $n=2r$ qubits. 
    \item A closeness testing problem defined with respect to the squared fidelity, in particular, the promise problem \textsc{Quantum Squared Fidelity Estimation} ($\FEst[\alpha,\beta]$). This problem asks whether $\F^2(\rho_0,\rho_1) \geq \alpha(n)$ for \emph{yes} instances or $\F^2(\rho_0,\rho_1) \leq \beta(n)$ for \emph{no} instances. 
\end{enumerate}

\subsection{A space-efficient quantum algorithmic toolkit}

\paragraph{Space-efficient QSVT.}
We start by introducing key tools from the space-efficient quantum singular value transformation (QSVT) framework~\cite[Section 3]{LGLW23}. In particular, we recall the notions of block-encodings and singular value transformations of linear operators: 
\sloppy
\begin{definition}[Block encodings, adapted from~\cite{GSLW19}]
    A unitary $U$ is called an $(\alpha,a,\epsilon)$-\emph{block-encoding} of a linear operator $A$ if 
    \[\norm{A - \alpha\rbra*{\bra{0}^{\otimes a}}U\rbra*{\ket{0}^{\otimes a}}} \leq \epsilon.\] 
    Here, $U$ acts on $s+a$ qubits. In particular, a block-encoding  $U$ is called an \emph{exact block-encoding} if the normalization factor satisfies $\alpha=1$ and the error $\epsilon=0$. 
\end{definition}

\begin{definition}[Singular value transformation by even or odd functions, adapted from Definition 9 in~{\cite{GSLW19}}]
    \label{def:matrix-SV-function}
    Let $f\colon \bbR \rightarrow \bbC$ be an even or odd function, and let $A\in \bbC^{\tilde{d}\times d}$ have the singular value decomposition $A = \sum_{i=1}^{\min\{d,\tilde{d}\}} \sigma_i \ket{\tilde{\psi}_i}\bra{\psi_i}$. The \emph{singular value transformation} corresponding to $f$ is defined as: 
    \[f^{\SV}(A) \coloneqq \begin{cases}
        \sum_{i=1}^{\min\{d,\tilde{d}\}} f(\sigma_i) \ket{\tilde{\psi}_i}\bra{\psi_i},&\text{for odd }f,\\
        \sum_{i=1}^d f(\sigma_i) \ket{\psi_i}\bra{\psi_i},&\text{for even }f.
    \end{cases}\]
    Here, $\sigma_i\coloneqq0$ for $i \in \{\min\{d,\tilde{d}\}+1,\cdots,d-1,d\}$. 
    In particular, for any Hermitian matrix $A$, it holds that $f^{\SV}(A)=f(A)$. 
\end{definition}

In this work, we require the space-efficient QSVT associated with the sign function,
\[\sign(x)\coloneqq{\scriptscriptstyle
    \begin{cases}
    1,& x > 0\\
    -1,& x < 0\\
    0,& x=0
    \end{cases}}.\]
To obtain such an algorithmic subroutine, we use a polynomial approximation of the sign function whose coefficients can be computed space-efficiently: 
\begin{lemma}[Space-efficient approximation to the sign function, adapted from~{\cite[Corollary 3.6]{LGLW23}}]
    \label{lemma:space-efficient-sign}
    For any $\delta > 0$ and $\epsilon > 0$, there exists an explicit odd polynomial 
    \[P_{d'}^{\sign}(x)=\hat{c}_0/2+\sum_{k=1}^{d'} \hat{c}_k T_k(x) \in \bbR[x]\] 
    of degree $d'\leq\tilde C_{\sign
    }\cdot\frac{1}{\delta}\log{\frac{1}{\epsilon}}$, where $d'=2d-1$ and $\tilde C_{\sign}$ is a universal constant. 
    Every entry of the coefficient vector $\hat{\bfc}\coloneqq(\hat{c}_0,\cdots,\hat{c}_{d'})$ can be computed in deterministic time $\widetilde{O}\big(d^{2}/\sqrt{\epsilon}\big)$ and space $O(\log(d^3/\epsilon^{3/2}))$. Furthermore, the polynomial $P_{d'}^{\sign}$ satisfies the following conditions: 
    \begin{align*}
    \forall x \in [-1,1] \setminus [-\delta,\delta],~& \left|\sign(x)-P^{\sign}_{d'}(x)\right| \leq C_{\sign} \epsilon, \text{ where } C_{\sign}=5;\\
    \forall x \in [-1,1],~& \left|P_{d'}^{\sign}(x) \right| \leq 1.
    \end{align*}    
    Furthermore, the coefficient vector $\hat{\bfc}$ satisfies the norm bound $\|\hat{\bfc}\|_1 \leq \hat{C}_{\sign} \log{d'}$, where $\hat{C}_{\sign}$ is another universal constant. 
    We assume without loss of generality that $d'\geq 2$ and that $\hat{C}_{\sign}$ and $\tilde{C}_{\sign}$ are at least $1$.
\end{lemma}

With the polynomial approximation in \Cref{lemma:space-efficient-sign}, we can now state the space-efficient QSVT procedure associated with the sign function:
\begin{lemma}[Sign polynomial with space-efficient coefficients applied to block-encodings, adapted from~{\cite[Corollary 3.25]{LGLW23}}]
    \label{lemma:sign-polynomial-implementation}
    Let $A$ be an Hermitian matrix that acts on $s$ qubits, where $s(n) \geq \Omega(\log(n))$. Let $U$ be a $(1,a, \epsilon_1)$-block-encoding of $A$ that acts on $s+a$ qubits. 
    Then, for any $d' \leq 2^{O(s(n))}$ and $\epsilon_2 \geq 2^{-O(s(n))}$, we have an $\rbra[\big]{1, a+\ceil*{\log{d'}}+3, 144\hat{C}_{\sign}d'\log d'\epsilon_1^{1/2} + (36\hat{C}_{\sign}\log d' +37)\epsilon_2}$-block-encoding $V$ of $P_{d'}^{\sign}(A)$, where $P_{d'}^{\sign}$ is a space-efficient bounded polynomial approximation of the sign function from \Cref{lemma:space-efficient-sign}, and $\hat{C}_{\sign}$ is a universal constant. 
    This construction requires $O(d^2 \log{d})$ uses of $U$, $U^{\dagger}$, $\textsc{C}_{\Pi}\textsc{NOT}$,  $\textsc{C}_{\tilde{\Pi}}\textsc{NOT}$, and $O(d^2 \log{d})$ multi-controlled single-qubit gates. The description of $V$ can be computed in deterministic time $\widetilde{O}(d^{9/2}/\epsilon_2)$ and space $O(s(n))$.

    \noindent Furthermore, our construction directly extends to any non-Hermitian (but linear) matrix $A$ by replacing $P_{d'}^{\sign}(A)$ with $P_{\sign,d'}^{\SV}(A)$, defined analogously to \Cref{def:matrix-SV-function}.
\end{lemma}

\paragraph{Other quantum algorithmic subroutines.}
The first two subroutines are required to obtain an exact block-encoding of $(\rho_0-\rho_1)/2$. In particular, \Cref{lemma:purified-density-matrix} traces back to~\cite{LC19} and \Cref{lemma:space-efficient-LCU} is a space-efficient specialization of the LCU method~\cite{BCC+15}, with its space complexity analyzed in~\cite{LGLW23}.

\begin{lemma}[Purified density matrix, adapted from~{\cite[Lemma 25]{GSLW19}}]
\label{lemma:purified-density-matrix}
    Let $\rho$ be a quantum state on an $s$-qubit register $\sfA$, and let $U$ be a unitary acting on $(\sfA,\sfR)$ that prepares a purification of $\rho$, where the reference register $\sfR$ contains $a$ qubits. Specifically, 
    \[U \ket{0}^{\otimes a} \ket{0}^{\otimes s} = \ket{\rho} \quad \text{and} \quad \rho = \Tr_{\sfR}(\ket{\rho}\bra{\rho}).\]
    Then there exists an $O(a+s)$-qubit quantum circuit $\widetilde U$ that is a $(1, a+s, 0)$-block-encoding of $\rho$, using $O(1)$ queries to $U$ and $O(a+s)$ one- and two-qubit quantum gates. 
\end{lemma}

For a nonzero real vector $\bfy=(y_0,\ldots,y_{m-1})$, we say that
$P_{\abs{\bfy}}$ is an $\epsilon$-state preparation operator for $\bfy$ if $P_{\abs{\bfy}}\ket{\bar{0}} \coloneqq \sum_{i=0}^{m-1}\sqrt{\hat{y}'_i}\ket{i}$ for some $\hat{\bfy}'$ satisfying $ \|\abs{\bfy}/\|\bfy\|_1-\hat{\bfy}'\|_1 \leq \epsilon$. 

\begin{lemma}[Linear combinations of bitstring indexed encodings, adapted from~{\cite[Lemma 29]{GSLW19}} and~{\cite[Lemma 3.22]{LGLW23}}] 
    \label{lemma:space-efficient-LCU}
    Given a matrix $A=\sum_{i=0}^{m-1} y_i A_i$ such that each linear operator $A_i$~$(0 \leq i < m)$ acts on $s$ qubits with the corresponding $(1,a,\epsilon_1)$-bitstring indexed encoding $U_i$ acting on $s+a$ qubits associated with the same projections $\tilde{\Pi}$ and $\Pi$. Also each $y_i$~$(0 \leq i < m)$ can be expressed in $O(s(n))$ bits with an evaluation oracle $\Eval$ that returns $\hat{y}_i$ with precision $\varepsilon\coloneqq O(\epsilon_2^2/m)$. Then utilizing an $\epsilon_2$-state preparation operator $P_{\abs{\bfy}}$ for the nonnegative vector $(\abs*{y_0},\ldots,\abs*{y_{m-1}})$ acting on $O(\log m)$ qubits, the diagonal unitary $D_{\bfy}=\sum_{i=0}^{m-1}\sigma_i \ket{i}\bra{i}+\rbra*{ I-\sum_{i=0}^{m-1}\ket{i}\bra{i}}$, 
    where $\sigma_i\in\{\pm 1\}$ satisfies $y_i=\sigma_i\abs*{y_i}$,
    and a $(s+a+\lceil\log{m}\rceil)$-qubit unitary $W=\sum_{i=0}^{m-1} \ket{i}\bra{i} \otimes U_i + \big(I-\sum_{i=0}^{m-1} \ket{i}\bra{i}\big)\otimes I$, we can implement a $(\|\bfy\|_1,a+\lceil\log{m}\rceil,\epsilon_1 \|\bfy\|_1 + \epsilon_2  \|\bfy\|_1)$-bitstring indexed encoding of $A$ acting on $s + a + \lceil\log m\rceil$ qubits with a single use of $W$, $P_{\abs{\bfy}}$, $P_{\abs{\bfy}}^{\dagger}$, and $D_\bfy$. 
    In addition, the (classical) pre-processing can be implemented in deterministic time $\tilde{O}(m^2 \log(m/\epsilon_2))$ and space $O(\log(m/\epsilon_2^2))$, using $m^2$ oracle calls to $\Eval$ with precision $\varepsilon$. 
\end{lemma}

The final subroutine is a specific version of one-bit precision phase estimation~\cite{Kitaev95}, commonly known as the Hadamard test~\cite{AJL09}:
\begin{lemma} [Hadamard test for block-encodings, adapted from~{\cite[Lemma 9]{GP22}}]
\label{lemma:Hadamard-test}
    Let $U$ be a $(1,a,0)$-block encoding of an $s(n)$-qubit linear operator $A$. There exists an explicit quantum circuit acting on $O(a+s)$ qubits that takes an $s(n)$-qubit quantum state $\rho$ as input and outputs $0$ with probability $\frac{1+\operatorname{Re}(\Tr(A\rho))}{2}$. 
\end{lemma}

\subsection{Error reduction for \QIPtwo{} via parallel repetition}
\label{sbusec:QIP(2)-error-reduction}

We briefly recap error reduction for two-message quantum interactive proof systems based on parallel repetition, following~\cite[Section 3.2]{JUW09}: 
\begin{lemma}[Error reduction for \QIPtwo{}, adapted from~{\cite[Section 3.2]{JUW09}}]
    \label{lemma:QIPtwo-error-reduction}
    Let $\protocol{\calP}{\calV}$ be a two-message quantum interactive proof system  with completeness $c(n)$ and soundness $s(n)$, satisfying $c(n)-s(n) \geq 1/q(n)$ for some function $q(n)$ that is polynomial in $n$. For any efficiently computable function $l(n)$ that is polynomial in $n$, one can construct a two-message quantum interactive proof system $\protocol{\calP'}{\calV'}$ with completeness $c'(n) \geq 1-2^{-l(n)}$ and soundness $s'(n) \leq 2^{-l(n)}$ that follows the repetition procedure described in \Cref{algo:QIP(2)-error-reduction}. 
\end{lemma}

The new proof system $\protocol{\calP'}{\calV'}$ performs $t_0$ parallel batches of repetitions of $\protocol{\calP}{\calV}$, where each batch consists of $t_1$ independent executions, as specified in \Cref{algo:QIP(2)-error-reduction}. Acceptance in $\protocol{\calP'}{\calV'}$ is determined by taking the logical AND of the $t_0$ batch outcomes, where each batch outcome is obtained by applying a (shifted) majority vote to the outcomes of the $t_1$ executions in that batch. 

\begin{algorithm}[!htp]
	\caption{Error reduction for two-message quantum interactive proof systems.}
	\label{algo:QIP(2)-error-reduction}
    \SetEndCharOfAlgoLine{.}
    \SetKwInOut{Parameter}{Parameters}
    \Parameter{$t_0 \coloneqq 2lq$, $t_1 \coloneqq 8lq^2 t_0$.}
        1. The verifier $\calV'$ executes the proof system $\protocol{\calP}{\calV}$ independently and in parallel for every pair $(i,j)$ with $i \in [t_0]$ and $j \in [t_1]$\;
        2. For each execution $(i,j)$, the verifier $\calV'$ measures the designated output qubit and records the measurement outcome as $y_{i,j}\in\binset$\;
        3. The verifier $\calV'$ accepts if $\wedge_{i=1}^{t_0} z_i =1$, and rejects otherwise. For each batch $i \in [t_0]$, 
        \[z_i \coloneqq \begin{cases}
            1, & \text{ if } \sum_{j=1}^{t_1} y_{i,j} \geq t_1 \cdot \frac{c+s}{2}.\\
            0, & \text{ otherwise.}
        \end{cases}\]
\end{algorithm}


\section{Algorithmic Holevo--Helstrom measurement and its implication}
\label{sec:algo-HH-meas}

In this section, we introduce an \textit{algorithmic} version of the Holevo--Helstrom measurement that nearly achieves the optimal probability for discriminating between quantum states $\rho_0$ and $\rho_1$. 
We start by defining the \textsc{Computational Quantum Hypothesis Testing Problem}, which assumes access to the descriptions of the corresponding state-preparation circuits: 
\begin{problem}[Computational Quantum Hypothesis Testing Problem]
    \label{prob:state-discrimination}
    Let $Q_0$ and $Q_1$ be two polynomial-size quantum circuits acting on $n$ qubits and having $r$ designated output qubits. Let $\rho_b$ denote the quantum state obtained by performing $Q_b$ on the initial state $\ket{0}^{\otimes n}$ and tracing out the non-output qubits for $b\in \binset$. Now, consider the following computational task: 
    \begin{itemize}
        \item \textbf{Input:} A quantum state $\rho$, either $\rho_0$ or $\rho_1$, is chosen uniformly at random. 
        \item \textbf{Output:} A bit $b$ indicates that $\rho=\rho_b$.  
    \end{itemize}
\end{problem}

The goal is to maximize the probability that the test in \Cref{prob:state-discrimination} succeeds, which can be achieved by performing an appropriate measurement on the given state $\rho$. 

\paragraph{Information-theoretic background.}
For the \textsc{Quantum Hypothesis Testing Problem} analogous to \Cref{prob:state-discrimination}, where $\rho_0$ and $\rho_1$ are not necessarily efficiently preparable, the maximum success probability to discriminate between quantum states $\rho_0$ and $\rho_1$ is given by the celebrated Holevo--Helstrom bound: 

\begin{theorem}[Holevo--Helstrom bound,~\cite{Holevo73TraceDist,Helstrom69}]
    \label{thm:HH-bound}
    Given a quantum state $\rho$, either $\rho_0$ or $\rho_1$, that is chosen uniformly at random, the maximum success probability to discriminate between quantum states $\rho_0$ and $\rho_1$ is given by $\frac{1}{2}+\frac{1}{2}\td(\rho_0,\rho_1)$.
\end{theorem}

Noting that the trace distance can be written as\footnote{Notably, \Cref{eq:trace-distance-decomposition} is fundamental in quantum algorithms for estimating the trace distance~\cite{WZ23}. An extension of this identity also underlies quantum algorithms for estimating the (powered) quantum $\ell_\alpha$ distance~\cite{LW25Lalpha}, which generalizes the trace distance via the (powered) Schatten ($\alpha$-)norm.}
\begin{equation}
    \label{eq:trace-distance-decomposition}
    \td(\rho_0,\rho_1) \coloneqq \frac{1}{2}\Tr|\rho_0-\rho_1| = \frac{1}{2} \rbra*{ \Tr\rbra[\Big]{\rho_0~\sign^\SV\rbra[\Big]{\frac{\rho_0-\rho_1}{2}}} - \Tr\rbra[\Big]{\rho_1~\sign^\SV\rbra[\Big]{\frac{\rho_0-\rho_1}{2}}} },
\end{equation}
one can directly obtain an explicit form of the optimal two-outcome measurement $\cbra*{\Pi_0,\Pi_1}$ that achieves \Cref{thm:HH-bound} and satisfies $\td(\rho_0,\rho_1) = \Tr(\Pi_0\rho_0)-\Tr(\Pi_0\rho_1)$:
\begin{proposition}[Explicit form of the Holevo--Helstrom measurement]
    \label{prop:HH-meas}
    \sloppy
    An optimal two-outcome measurement $\{\Pi_0,\Pi_1\}$ that maximizes the discrimination probability in quantum hypothesis testing and achieves the Holevo--Helstrom bound (\Cref{thm:HH-bound}) is given by
    \[ \Pi_0 = \frac{I}{2}+\frac{1}{2} \sign^{\SV}\Big(\frac{\rho_0-\rho_1}{2}\Big) \quad \text{and} \quad \Pi_1 = \frac{I}{2}-\frac{1}{2} \sign^{\SV}\Big(\frac{\rho_0-\rho_1}{2}\Big).  \]
\end{proposition}

\paragraph{Algorithmic implementation.}
Using the space-efficient quantum singular value transformation in~\cite[Section 3]{LGLW23}, the Holevo--Helstrom measurement specified in \Cref{prop:HH-meas} can be approximately implemented in quantum \textit{single-exponential} time and \textit{linear} space. We refer to this explicit implementation as the \textit{algorithmic Holevo--Helstrom measurement}:

\begin{theorem}[Algorithmic Holevo--Helstrom measurement]
    \label{thm:algo-HH-meas}
    Let $\rho_0$ and $\rho_1$ be quantum states prepared by $n$-qubit quantum circuits $Q_0$ and $Q_1$, respectively, as defined in \Cref{prob:state-discrimination}.
    An approximate version of the Holevo--Helstrom measurement $\Pi_0$ specified in \Cref{prop:HH-meas}, denoted as $\tilde{\Pi}_0$, can be implemented so that 
    \begin{equation}
        \label{eq:algo-HH-bound}
        \td(\rho_0,\rho_1) - 2^{-n} \leq \Tr(\tilde{\Pi}_0 \rho_0)-\Tr(\tilde{\Pi}_0 \rho_1) \leq \td(\rho_0,\rho_1).
    \end{equation}
    The quantum circuit implementation of $\tilde{\Pi}_0$, acting on $O(n)$ qubits, requires $2^{O(n)}$ queries to the quantum circuits $Q_0$ and $Q_1$, as well as $2^{O(n)}$ one- and two-qubit quantum gates. Moreover, the circuit description can be computed in deterministic time $2^{O(n)}$ and space $O(n)$.
\end{theorem}

Additionally, we demonstrate an implication of our algorithmic Holevo--Helstrom measurement in \Cref{{thm:algo-HH-meas}}. By inspecting the (honest-verifier) quantum statistical zero-knowledge protocol (``distance test'') for $\QSD[\alpha,\beta]$ with constants $\alpha^2>\beta$ in~\cite[Section 4.2]{Watrous02}, we obtain the \QIPtwo{} part in \Cref{corr:QSZK-upper-bound}, since \GapQSD{} is \QSZK{}-hard:
\begin{theorem}[\GapQSD{} is in \QIPtwo{} with a quantum linear-space honest prover]
    \label{thm:GapQSD-in-QIP(2)-bounded-prover}
    There exists a two-message quantum interactive proof system for $\QSD[\alpha(n),\beta(n)]$ with completeness $c(n)=\big(1+\alpha(n)-2^{-n}\big)/2$ and soundness $s(n)=(1+\beta(n))/2$. Moreover, a near-optimal prover strategy for this proof system can be implemented in quantum single-exponential time and linear space. 
    Consequently, for any $\alpha(n)$ and $\beta(n)$ satisfying $\alpha(n) - \beta(n) \geq 1/\poly(n)$, 
    \[\QSD[\alpha(n),\beta(n)] \text{ is in } \QIPtwo{} \text{ with a  quantum } O(n') \text{ space honest prover},\] where $n'$ is the total input length of the quantum circuits that prepare the corresponding tuple of quantum states.\footnote{This tuple of quantum states arises from a standard parallel repetition of the two-message quantum interactive proof system for $\GapQSD[\alpha(n),\beta(n)]$ with $c(n)-s(n) \geq 1/\poly(n)$. See \Cref{sbusec:QIP(2)-error-reduction} for details.}
\end{theorem}

\vspace{1em}
In the rest of this section, we provide the proof of \Cref{thm:algo-HH-meas} and the proof of \Cref{thm:GapQSD-in-QIP(2)-bounded-prover} in \Cref{subsec:algo-HHmeas-proof} and \Cref{subsec:GapQSD-in-QIP(2)-proof}, respectively.

\subsection{Algorithmic Holevo--Helstrom measurement: Proof of \texorpdfstring{\Cref{thm:algo-HH-meas}}{Theorem 3.4}}
\label{subsec:algo-HHmeas-proof}

To implement our algorithmic Holevo--Helstrom measurement, we adopt the one-bit precision phase estimation (often referred to as the Hadamard test, \Cref{lemma:Hadamard-test}), which reduces the task to implementing the corresponding unitary. The starting point is the space-efficient polynomial approximation $P^{\sign}_{d'}$ of the sign function (\Cref{lemma:space-efficient-sign}), which yields the two-outcome measurement $\{\hat{\Pi}_0,\hat{\Pi}_1\}$ defined by:
\begin{equation*}
    \label{eq:approx-HH-meas}
    \hat{\Pi}_0 = \frac{I}{2}+\frac{1}{2} P^{\sign}_{d'}\Big(\frac{\rho_0-\rho_1}{2}\Big) \quad \text{and} \quad \hat{\Pi}_1 = \frac{I}{2}-\frac{1}{2} P^{\sign}_{d'}\Big(\frac{\rho_0-\rho_1}{2}\Big). 
\end{equation*}

By applying the space-efficient QSVT associated with the polynomial $P^{\sign}_{d'}$ to the block-encoding of $(\rho_0-\rho_1)/2$ (\Cref{lemma:sign-polynomial-implementation}),
we obtain the unitary $\UHH$, which is a block-encoding of $\AHH \colonapprox P^{\sign}_{d'}\big(\frac{\rho_0-\rho_1}{2}\big)$. 
We therefore implement the two-outcome measurement $\{\tilde{\Pi}_0,\tilde{\Pi}_1\}$ where $\tilde{\Pi}_0= (I+\AHH)/2$, and the difference between $\{\hat{\Pi}_0,\hat{\Pi}_1\}$ and $\{\tilde{\Pi}_0,\tilde{\Pi}_1\}$ is caused by the implementation error of our space-efficient QSVT. We then proceed to the proof. 

\begin{proof}[Proof of \Cref{thm:algo-HH-meas}]
    Our algorithmic Holevo--Helstrom measurement is inspired by the \BQP{} containment of the low-rank variant of \GapQSD{}~\cite[Section III.A]{WZ23} and the \BQL{} containment of \GapQSDlog{}~\cite[Section 4.2]{LGLW23}, as presented in \Cref{fig:algo-HH-measurement}. 

\begin{figure}[!htp]
    \centering
    \begin{quantikz}[wire types={q,b,b}, classical gap=0.06cm, row sep=0.75em]
        \lstick{$\ket{0}^{\sfF}$} & \gate{H} & \ctrl{1} & \gate{H} & \meter{} & \setwiretype{c} \rstick{$b$}\\
        \lstick{$\ket{\bar 0}^{\sfE}$} &  & \gate[2]{\UHH \colonapprox U_{P^{\sign}_{d'}\big(\frac{\rho_0-\rho_1}{2}\big)}}  &  &  \\
        \lstick{$\rho^\sfA$} &  &  &  &  \\
    \end{quantikz}
    \caption{Algorithmic Holevo--Helstrom measurement.}
    \label{fig:algo-HH-measurement}
\end{figure}
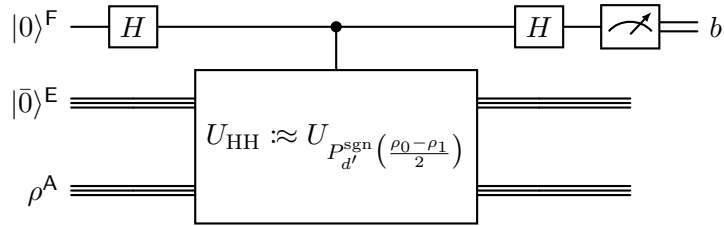

    Note that the input state $\rho$ in register $\sfA$ to the circuit in \Cref{fig:algo-HH-measurement} is an $r(n)$-qubit quantum state, either $\rho_0$ or $\rho_1$. This state is obtained by preparing the corresponding $n$-qubit purification on registers $(\sfA,\sfR)$ using the polynomial-size quantum circuit $Q_0$ or $Q_1$, and then tracing out the non-output qubits in register $\sfR$, as described in \Cref{prob:state-discrimination}. 
    Since the Hadamard test (\Cref{lemma:Hadamard-test}) reduces the task to implementing an appropriate unitary acting on register $\sfA$ and the ancillary register $\sfE$, specifically $\UHH$, we construct it as follows: 
    
\begin{enumerate}[label={\upshape(\arabic*)}]
    \item Applying \Cref{lemma:purified-density-matrix}, we can construct $O(n)$-qubit quantum circuits $U_{\rho_0}$ and $U_{\rho_1}$ that encode $\rho_0$ and $\rho_1$ as $(1,n,0)$-block-encodings, using $O(1)$ queries to $Q_0$ and $Q_1$, as well as $O(n)$ one- and two-qubit quantum gates. 
    \item Applying \Cref{lemma:space-efficient-LCU}, we can construct a $(1,n+1,0)$-block-encoding $U_{\frac{\rho_0-\rho_1}{2}}$ of $\frac{\rho_0 - \rho_1}{2}$, using $O(1)$ queries to $Q_0$ and $Q_1$, as well as $O(n)$ one- and two-qubit quantum gates. 
    \item Let $P_{d'}^{\sign} \in \mathbb{R}[x]$ be the degree-${d'}$ polynomial obtained from some degree-$d$ averaged Chebyshev truncation, with $d'=2d-1$, as specified in \Cref{lemma:space-efficient-sign}. 
    We choose parameters $\varepsilon \coloneqq 2^{-n}$, $\delta \coloneqq \frac{\varepsilon}{2^{r+2}}$, $\epsilon \coloneqq \frac{\varepsilon}{64n\max\{1,36\hat{C}_{\sign},2C_{\sign}+37,\tilde C_{\sign}\}}$. Consequently, the degree $d' \coloneqq \tilde C_{\sign} \cdot \frac{1}{\delta} \log\frac{1}{\epsilon}$ satisfies $d'=2^{O(n)}$, where $\tilde C_{\sign}$ comes from \Cref{lemma:space-efficient-sign}. This degree is fixed before choosing the numerical precision used to compute the coefficients.\footnote{Indeed, since $r\leq n$, we have $\log(1/\delta)=O(n)$, and the explicit choice of $\epsilon$ gives $\log(1/\epsilon)=O(n)$. Hence $\log d'=O(n)$. Since the constants $\hat C_{\sign},C_{\sign},\tilde C_{\sign}$ are fixed, increasing the numerical constant $64$ if necessary gives, for all sufficiently large $n$, $36\hat C_{\sign}\epsilon\log d'\leq \varepsilon/4$ and $(2C_{\sign}+37)\epsilon\leq \varepsilon/4$.}
    Applying the space-efficient QSVT associated with the sign function (\Cref{lemma:sign-polynomial-implementation} with $\epsilon_1 \coloneqq 0$ and $\epsilon_2 \coloneqq \epsilon$), we obtain the unitary $\UHH$. 
\end{enumerate}

    \paragraph{Error analysis.} We first prove that $\cbra[\big]{\tilde{\Pi}_0,\tilde{\Pi}_1}$ forms a valid POVM. Once this is established, the upper bound in \Cref{eq:algo-HH-bound} follows directly from \Cref{thm:HH-bound}. By \Cref{lemma:Hadamard-test}, we have $\Pr{b=0} = \Tr(\tilde{\Pi}_0 \rho) \geq 0$ for every quantum state $\rho$. On the other hand, for any state $\rho$, let $\ket{\psi}$ denote its purification. Then, 
    \begin{align*}
        \Pr{b=0} = \Tr(\tilde{\Pi}_0 \rho) &= \frac{1}{2} + \frac{1}{2} \Tr\rbra*{\bra{\bar{0}}^{\sfE} \UHH \ket{\bar{0}}^\sfE \rho}\\
        &= \frac{1}{2} + \frac{1}{2} \Tr\rbra*{\bra{\psi}^{\sfA\sfR} \bra{\bar{0}}^{\sfE} \rbra[\big]{\UHH \!\otimes\! I^\sfR} \ket{\psi}^{\sfA\sfR} \ket{\bar{0}}^{\sfE}} \leq 1.
    \end{align*}
    Consequently, we conclude that $0 \leq \Tr(\tilde{\Pi}_0 \rho) \leq 1$ for all $\rho$, which confirms that $\cbra[\big]{\tilde{\Pi}_0,\tilde{\Pi}_1}$ is indeed a POVM.
    Next, it suffices to prove the following weaker bound: 
    \[|\td(\rho_0,\rho_1) - \big(\Tr(\tilde{\Pi}_0 \rho_0)-\Tr(\tilde{\Pi}_0 \rho_1)\big)| \leq 2^{-n}.\]    
    To this end, we first bound the error caused by space-efficient polynomial approximation in \Cref{lemma:space-efficient-sign}. 
    Consider the spectral decomposition $\frac{\rho_0-\rho_1}{2} = \sum_j \lambda_j \ket{\psi_j}\bra{\psi_j}$, where $\{\ket{\psi_j}\}$ is an orthonormal basis. We can define index sets $\Lambda_{-} \coloneqq \{j \colon \lambda_j < -\delta\}$, $\Lambda_{0} \coloneqq \{j \colon -\delta \leq \lambda_j \leq \delta\}$, and $\Lambda_{+} \coloneqq \{j \colon \lambda_j > \delta\}$. Next, we have derived that:        
    \begin{subequations}
    \label{eq:algo-HH-polyError}
    \begin{align}
        & \left|\td(\rho_0,\rho_1) - \big(\Tr(\hat{\Pi}_0 \rho_0)-\Tr(\hat{\Pi}_0 \rho_1)\big) \right| \\
        =~& \left| \Tr\left(\sign\Big(\frac{\rho_0-\rho_1}{2}\Big) \frac{\rho_0-\rho_1}{2}\right) - \Tr\left( P^{\sign}_{d'}\Big(\frac{\rho_0-\rho_1}{2}\Big) \frac{\rho_0-\rho_1}{2}\right) \right| \\
        \leq~& \sum_{j \in \Lambda_-} \left| \lambda_j \sign(\lambda_j) - \lambda_j P^{\sign}_{d'}(\lambda_j) \right| + \sum_{j \in \Lambda_0} \left| \lambda_j \sign(\lambda_j) - \lambda_j P^{\sign}_{d'}(\lambda_j) \right|\\
        &\quad + \sum_{j \in \Lambda_+} \left| \lambda_j \sign(\lambda_j) - \lambda_j P^{\sign}_{d'}(\lambda_j) \right|\\
        \leq~& \sum_{j \in \Lambda_-} |\lambda_j| \cdot |-1 - P^{\sign}_{d'}(\lambda_j) | + \sum_{j \in \Lambda_0} \left| \lambda_j \sign(\lambda_j) - \lambda_j P^{\sign}_{d'}(\lambda_j) \right|\\
        & \quad + \sum_{j \in \Lambda_+} |\lambda_j| \cdot |1 - P^{\sign}_{d'}(\lambda_j)|\\
        \leq~& \sum_{j \in \Lambda_-} |\lambda_j| C_{\sign} \epsilon + \sum_{j \in \Lambda_0} 2|\lambda_j| + \sum_{j \in \Lambda_+} |\lambda_j| C_{\sign} \epsilon \\
        \leq~& 2 C_{\sign} \epsilon + 2^{r+1}\delta.
    \end{align}
    \end{subequations}
    Here, the first inequality owes to the triangle inequality, the second inequality applies the sign function, the third inequality is guaranteed by \Cref{lemma:space-efficient-sign}, and the last inequality is because $\sum_j |\lambda_j| = \td(\rho_0,\rho_1) \leq 1$ and $\rank\big(\frac{\rho_0-\rho_1}{2}\big)$ is at most $2^r$. 
    
    We then bound the error caused by space-efficient QSVT implementation in \Cref{lemma:sign-polynomial-implementation}: 
    \begin{subequations}
    \label{eq:algo-HH-implementError}
    \begin{align}
        &\left|\big(\Tr(\hat{\Pi}_0 \rho_0) - \Tr(\hat{\Pi}_0 \rho_1)\big) - \big(\Tr(\tilde{\Pi}_0 \rho_0) -\Tr(\tilde{\Pi}_0 \rho_1)\big) \right|\\
        =~& \left| \Tr\left( P^{\sign}_{d'}\Big(\frac{\rho_0-\rho_1}{2}\Big) \frac{\rho_0-\rho_1}{2}\right) - \Tr\Big( \rbra[\big]{\bra{\bar{0}}^{\sfE} \otimes I^{\sfA}} \UHH \rbra[\big]{\ket{\bar{0}}^{\sfE} \otimes I^{\sfA} } \frac{\rho_0-\rho_1}{2} \Big)  \right|\\
        \leq~& \left\| P^{\sign}_{d'}\Big(\frac{\rho_0-\rho_1}{2}\Big) - \rbra[\big]{\bra{\bar{0}}^{\sfE} \otimes I^{\sfA}} \UHH \rbra[\big]{\ket{\bar{0}}^{\sfE} \otimes I^{\sfA}} \right\| \cdot \td(\rho_0,\rho_1)\\
        \leq~& (36\hat{C}_{\sign}\log{d'} + 37) \epsilon \cdot 1.
    \end{align}
    \end{subequations}
    Here, the third line follows from the H\"older inequality for Schatten norms (\Cref{lemma:matrix-Holder-inequality}), and the last line is guaranteed by \Cref{lemma:sign-polynomial-implementation}. 
    
    Combining error bounds in \Cref{eq:algo-HH-polyError,eq:algo-HH-implementError}, we obtain the following under the aforementioned choice of parameters: 
    \begin{align*}
        &\left|\td(\rho_0,\rho_1) - \big(\Tr(\tilde{\Pi}_0 \rho_0)-\Tr(\tilde{\Pi}_0 \rho_1)\big) \right| \\
        \leq~& \left|\td(\rho_0,\rho_1) - \big(\Tr(\hat{\Pi}_0 \rho_0) - \Tr(\hat{\Pi}_0 \rho_1)\big) \right| + \left|\big(\Tr(\hat{\Pi}_0 \rho_0) - \Tr(\hat{\Pi}_0 \rho_1)\big) - \big(\Tr(\tilde{\Pi}_0 \rho_0) -\Tr(\tilde{\Pi}_0 \rho_1)\big) \right|\\
        \leq~& 2 C_{\sign} \epsilon + 2^{r+1}\delta + (36\hat{C}_{\sign}\log{d'} + 37) \epsilon\\
        \leq~& \varepsilon.
    \end{align*}

    \paragraph{Complexity analysis.} We complete the proof by analyzing the computational complexity of our algorithm. According to \Cref{lemma:sign-polynomial-implementation}, our algorithm specified in \Cref{fig:algo-HH-measurement} requires $O(n)$ qubits and $O(d^2\log{d}) \leq \widetilde{O}(2^{2r}/\varepsilon^{2}) \leq 2^{O(n)}$ queries to $Q_0$ and $Q_1$. In addition, the circuit description of our algorithm can be computed in deterministic time $\widetilde{O}(d^{9/2}/\varepsilon) = \widetilde{O}(2^{4.5r}/\varepsilon^{5.5}) \leq 2^{O(n)}$. 
\end{proof}

\subsection{A slightly improved upper bound for \GapQSD{}: Proof of \texorpdfstring{\Cref{thm:GapQSD-in-QIP(2)-bounded-prover}}{Theorem 3.5}}
\label{subsec:GapQSD-in-QIP(2)-proof}

We start by presenting the quantum interactive proof system used in \Cref{thm:GapQSD-in-QIP(2)-bounded-prover}, as shown in \Cref{algo:GapQSD-in-QIP(2)-bounded-prover}. This proof system aligns with \cite[Figure 2]{Watrous02}, with the new component  being the honest prover's behavior. In particular, the honest prover now employs the algorithmic Holevo--Helstrom measurement $\{\tilde{\Pi}_0,\tilde{\Pi}_1\}$ from \Cref{thm:algo-HH-meas}, rather than the optimal measurement $\{\Pi_0,\Pi_1\}$ in \Cref{prop:HH-meas} as per \Cref{thm:HH-bound}.

\begin{algorithm}[!htp]
	\caption{Two-message proof system for \GapQSD{} (quantum linear-space prover).}
	\label{algo:GapQSD-in-QIP(2)-bounded-prover}
    \SetEndCharOfAlgoLine{.}
    1. The verifier $\calV$ first chooses $b\in\binset$ uniformly at random. Subsequently, $\calV$ applies $Q_b$ to $\ket{0}^{\otimes n}$, traces out all non-output qubits, and sends the resulting state $\rho_b$\;
    2. The verifier $\calV$ receives a single-qubit state $\hat{\sigma}$, measures the state in the computational basis, and denotes the outcome by $\hat{b}\in\binset$\;
    \begin{tcolorbox}[colback=gray!10, colframe=gray!40, boxrule=0pt, sharp corners, width=0.9\linewidth]
        The \emph{honest} prover $\calP$ measures the received state $\rho$ using the algorithmic Holevo--Helstrom measurement $\{\tilde{\Pi}_0,\tilde{\Pi}_1\}$ (\Cref{thm:algo-HH-meas}), and sends the outcome $\hat{b}$. In particular, the outcome equals $\hat{b}$ if the measurement indicates $\rho$ is $\rho_{\hat{b}}$.
    \end{tcolorbox}
    3. The verifier $\calV$ accepts if $b=\hat{b}$; otherwise $\calV$ rejects. 
    \BlankLine
\end{algorithm}

Following that, we delve into the analysis of Protocol \ref{algo:GapQSD-in-QIP(2)-bounded-prover}:
\begin{proof}[Proof of \Cref{thm:GapQSD-in-QIP(2)-bounded-prover}]
    Note that $\Pr{\hat{b}=a'|b=a}$ denotes the probability that the prover $\calP$ uses a two-outcome measurement $\{\Pi'_0,\Pi'_1\}$, which is arbitrary in general, to measure the state $\rho_a$, resulting in the measurement outcome $a'$ for $a,a' \in \binset$. We then derive the corresponding acceptance probability of Protocol \ref{algo:GapQSD-in-QIP(2)-bounded-prover}:
    \begin{subequations}
    \label{eq:QIP(2)-protocol-pacc}
    \begin{align}
        \Pr{b=\hat{b}} &= \Pr{b=0}\cdot\Pr{\hat{b}=0|b=0} + \Pr{b=1} \cdot \Pr{\hat{b}=1|b=1}\\
        &= \frac{1}{2}\Pr{\hat{b}=0|b=0} + \frac{1}{2} \rbra*{1 - \Pr{\hat{b}=0|b=1}}\\
        &= \frac{1}{2}+\frac{1}{2} \big( \Tr(\Pi'_0\rho_0)-\Tr(\Pi'_0\rho_1) \big).
    \end{align}
    \end{subequations}

    \vspace{1em}
    For \textit{yes} instances where $\td(\rho_0,\rho_1) \geq \alpha(n)$, considering that the prover $\calP$ is honest, we have 
    \begin{align*}
        \Pr{b=\hat{b}} &= \frac{1}{2}+\frac{1}{2}\big( \Tr(\tilde{\Pi}_0\rho_0)-\Tr(\tilde{\Pi}_0\rho_1) \big) \\
        &\geq \frac{1}{2} + \frac{1}{2}\rbra*{ \td(\rho_0,\rho_1) - 2^{-n} }
        \geq \frac{1}{2}+\frac{1}{2} \big( \alpha(n)-2^{-n} \big).
    \end{align*}
    Here, the first line uses \Cref{eq:QIP(2)-protocol-pacc}, and the second line follows from the lower bound in \Cref{thm:algo-HH-meas}. 
    Therefore, we have the completeness $c(n) = \frac{1}{2}+\frac{1}{2} \big( \alpha(n)- 2^{-n} \big)$ and the (honest) prover strategy described in \Cref{algo:GapQSD-in-QIP(2)-bounded-prover} is indeed implementable in quantum single-exponential time and linear space due to \Cref{thm:algo-HH-meas}. 

    \vspace{1em}
    For \textit{no} instances where $\td(\rho_0,\rho_1) \leq \beta(n)$, it suffices to consider a prover that returns a single qubit.\footnote{See \Cref{footnote:single-bit-response} for further details.} An arbitrary channel implementing the prover's strategy, followed by the verifier's computational-basis measurement of that qubit, induces a two-outcome POVM $\cbra*{\Pi'_0,\Pi'_1}$ on the received state. Therefore, we obtain the following from \Cref{eq:QIP(2)-protocol-pacc}: 
    $$\Pr{b=\hat{b}} = \frac{1}{2}+\frac{1}{2} \big( \Tr(\Pi'_0\rho_0)-\Tr(\Pi'_0\rho_1) \big) \leq \frac{1}{2}+\frac{1}{2} \td(\rho_0,\rho_1) \leq \frac{1}{2}\big( 1+\beta(n) \big) \coloneqq s(n).$$
    Here, the first inequality follows from the Holevo--Helstrom bound (\Cref{thm:HH-bound}). 

    \paragraph{Error reduction for \Cref{algo:GapQSD-in-QIP(2)-bounded-prover}.}
    To reduce the completeness and soundness errors, we apply error reduction for \QIPtwo{} (\Cref{lemma:QIPtwo-error-reduction}) to \Cref{algo:GapQSD-in-QIP(2)-bounded-prover} with $l(n)=n$. The resulting proof system $\protocol{\calP'}{\calV'}$ is obtained by substituting \Cref{algo:GapQSD-in-QIP(2)-bounded-prover} into \Cref{algo:QIP(2)-error-reduction}. 
    
    We now analyze the complexity of the honest prover $\calP'$. The resulting proof system $\protocol{\calP'}{\calV'}$, which consists of $t_0t_1$ independent and parallel executions of \Cref{algo:GapQSD-in-QIP(2)-bounded-prover}, can be viewed as discriminating $t_0t_1$ pairs of quantum states $(\rho_0^{(j)},\rho_1^{(j)})$ for $1 \leq j \leq t_0t_1$. The total input length of the quantum circuits to prepare the states $\rho_b^{(1)},\cdots,\rho_b^{(t_0t_1)}$ for $b\in\binset$ is 
    $n\cdot t_0t_1 = 32n^4 q^4(n) \leq O(n^\tau) \coloneqq n'$ for some positive constant $\tau$. After replacing $n$ with $n'$, the space complexity of the honest prover $\calP'$ remains $O(n')$. 
    Finally, the desired completeness and soundness errors follow from the fact that $2^{-l\rbra[\big]{(n')^{1/\tau}}} \leq 1/3$ whenever $n'\geq \rbra*{\log{3}}^\tau$. 
\end{proof}


\section{Algorithmic Uhlmann transform and its implications}
\label{sec:algo-Uhlmann}

In this section, we introduce an \textit{algorithmic} version of the Uhlmann transform that approximately attains the maximum overlap between purifications of the quantum states $\rho_0$ and $\rho_1$. We begin by defining the \textsc{Computational Uhlmann Fidelity Test Problem}, which assumes access to the descriptions of the corresponding state-preparation circuits: 
\begin{problem}[Computational Uhlmann Fidelity Test Problem]
    \label{prob:fidelity-test}
    Let $Q_0$ and $Q_1$ be two known polynomial-size quantum circuits acting on $n$ qubits in registers $(\sfA, \sfR)$, each with $r$ designated output qubits in register $\sfA$. 
    Let $\ket{\psi_b}$ be the pure state produced by applying $Q_b$ to the initial state $\ket{0}^{\otimes n}$ for $b\in \binset$, and let $\rho_b$ be the quantum state obtained by tracing out all non-output qubits in register $\sfR$. 
    \begin{itemize}
        \item \textbf{Input:} The qubits of the purification $\ket{\psi_1}$ in the reference register $\sfR$, while all qubits in the output register $\sfA$ are fixed and cannot be modified.         
        \item \textbf{Output:} A bit $z$ obtained from the two-outcome measurement $\cbra{\Pi, I-\Pi}$, where $\Pi \coloneqq \ketbra{\psi_0}{\psi_0}$. 
    \end{itemize}
\end{problem}

The goal is to maximize the probability that the test in \Cref{prob:fidelity-test} succeeds (i.e., obtaining the first outcome), which can be accomplished by applying an appropriate dimension-preserving quantum channel to register $\sfR$.

\paragraph{Information-theoretic background.}
\Cref{prob:fidelity-test} naturally generalizes the (information-theoretic) Uhlmann fidelity test between a quantum state $\rho$ and a pure state $\ketbra{\phi}{\phi}$, as stated in~\cite[Exercise 9.2.2]{Wilde13}. In that setting, the reference register $\sfR$ does not appear, and the two-outcome measurement is $\cbra*{\ketbra{\phi}{\phi}, I-\ketbra{\phi}{\phi}}$. 
The success probability of the test is $\Tr(\ketbra{\phi}{\phi}\rho)$, which coincides exactly with the squared fidelity $\F^2(\ketbra{\phi}{\phi},\rho)$. 

For two quantum states that may be mixed, the probability of obtaining the first outcome in the Uhlmann fidelity test (i.e., the information-theoretic counterpart of \Cref{prob:fidelity-test}) is characterized by the squared fidelity~\cite{Uhlmann76}, as will be seen later in \Cref{corr:stronger-Uhlmann}. A refined formulation with an elementary proof appears in~\cite{Jozsa94} and is stated below: 
\begin{theorem}[Uhlmann's theorem, adapted from~{\cite[Theorem 2]{Jozsa94}}]
    \label{thm:Uhlmann}
    Let $\rho_0$ and $\rho_1$ be quantum states on register $\sfA$. For any fixed purification $\ket{\psi_0}$ of $\rho_0$ on registers $(\sfA,\sfR)$, the maximum squared overlap between $\ket{\psi_0}$ and any purification $\ket{\psi_1}$ of $\rho_1$ on the same registers is given by the squared (Uhlmann) fidelity:\footnote{The last equality in \Cref{eq:Uhlmann-identity} follows from the freedom in purifications (see, e.g.,~\cite[Exercise 2.81]{NC10}).} 
    \begin{equation}
        \label{eq:Uhlmann-identity}
        \F^2(\rho_0,\rho_1) \coloneqq \Tr\rbra*{\abs{\sqrt{\rho_0}\sqrt{\rho_1}}}^2 = \max_{\ket{\psi_1}} \abs*{\innerprod{\psi_0}{\psi_1}}^2 = \max_{U} \abs*{\bra{\psi_0} \rbra[\big]{I^{\sfA}\!\otimes\! U^{\sfR}} \ket{\psi'_1}}^2.
    \end{equation}
    Here, the last identity transfers the freedom in choosing $\ket{\psi_1}$ to the freedom in choosing a unitary $U^\sfR$ while keeping the purification $\ket{\psi'_1}$ fixed.
\end{theorem}

By inspecting the soundness analysis in the proof of~\cite[Theorem 11]{Watrous02}, which uses the monotonicity $\F^2(\rho_0,\rho_1) \leq \F^2(\calE(\rho_0),\calE(\rho_1))$ for every quantum channel $\calE$, see e.g.,~\cite[Theorem 9.6]{NC10}, one obtains a stronger form of \Cref{thm:Uhlmann}:\footnote{Note that although an arbitrary quantum channel $\Phi(\cdot)$ may act on register $\sfR$, the reduced density matrix on register $\sfA$ remains $\rho_1$. Let $\sigma \coloneqq \rbra[\big]{I^{\sfA}\!\otimes\!\Phi^{\sfR}}(\ketbra{\psi_1}{\psi_1})$ be the resulting state on registers $(\sfA,\sfR)$ such that $\Tr_{\sfR}(\sigma)=\rho_1$. Then $\bra{\psi_0} \sigma \ket{\psi_0} = \F^2(\ketbra{\psi_0}{\psi_0},\sigma) \leq \F^2\rbra*{\Tr_{\sfR}\rbra*{\ketbra{\psi_0}{\psi_0}},\Tr_{\sfR}\rbra*{\sigma}} = \F^2(\rho_0,\rho_1)$. The same argument applies when the roles of $\rho_0$ and $\rho_1$ are exchanged. \label{footnote:stronger-Uhlmann}} 

\begin{corollary}[A stronger form of Uhlmann's theorem, implicit in~{\cite[Theorem 11]{Watrous02}}]
    \label{corr:stronger-Uhlmann}
    Let $\rho_0$ and $\rho_1$ be quantum states on register $\sfA$. For any fixed purifications $\ket{\psi_0}$ and $\ket{\psi_1}$ of $\rho_0$ and $\rho_1$, respectively, on registers $(\sfA,\sfR)$, the squared (Uhlmann) fidelity satisfies
    \[ \F^2(\rho_0,\rho_1) = \max_{\Phi} \bra{\psi_0} \rbra[\big]{I^{\sfA}\!\otimes\!\Phi^{\sfR}}(\ketbra{\psi_1}{\psi_1}) \ket{\psi_0}, \]
    where the maximization ranges over all quantum channels $\Phi$ acting on the register $\sfR$ and preserving its dimension. 
\end{corollary}

By examining the proof of \cite[Theorem 2]{Jozsa94} together with \cite[Lemma 6]{Jozsa94}, one can extract an explicit expression for the optimal unitary, later referred to as the \emph{Uhlmann transform}, that achieves the maximum in \Cref{eq:Uhlmann-identity}, as stated in \Cref{lemma:Uhlmann-transform}. A self-contained proof can be found in~\cite[Appendix F]{UNWT25} (derived from Jozsa's lemma) or Lemma 7.6 in the arXiv version of~\cite{MY23}. 

\begin{lemma}[Explicit form of the Uhlmann transform, implicit in~{\cite[Lemma 6]{Jozsa94}}]
    \label{lemma:Uhlmann-transform}
    Let $\ket{\psi_0}$ and $\ket{\psi_1}$ be purifications of quantum states $\rho_0$ and $\rho_1$ on register $\sfA$, defined on registers $(\sfA,\sfR)$, where $\sfR$ is the reference register. A unitary $U_{\star}$ on register $\sfR$ that attains the maximum in Uhlmann's theorem (\Cref{thm:Uhlmann}) is given by\footnote{For singular $X$, $\sign^\SV(X)$ is a partial isometry. Throughout this work, whenever a unitary Uhlmann transform or its induced channel is required, we implicitly extend its action from $(\ker X)^\perp$ to a unitary on the full space. This extension does not affect the overlap in Uhlmann's theorem.}
    \[ U_{\star} = \sign^\SV \rbra*{\Tr_{\sfA} \rbra*{ \ketbra{\psi_0}{\psi_1}^{\sfA\sfR} }}.\]
\end{lemma}

\paragraph{Algorithmic implementation.}
Unlike the Holevo--Helstrom measurement in~\Cref{prop:HH-meas}, for which one can directly obtain an \textit{exact} block-encoding of $X_\mathrm{HH}\coloneqq (\rho_0-\rho_1)/2$, constructing a block-encoding of $X_\mathrm{Uhl}\coloneqq \Tr_{\sfA} \rbra*{ \ketbra{\psi_0}{\psi_1}^{\sfA\sfR} }$ is more involved.
A straightforward approach introduces a normalization factor of $\dim(\calH_\sfA)$ and results only in an exact encoding of $X_\mathrm{Uhl}/\!\dim(\calH_\sfA)$~\cite{MY23}.\footnote{See Section 7 in the arXiv version of~\cite{MY23}.} Instead, an exact block-encoding of $X_\mathrm{Uhl}$ was recently proposed in~\cite[Section 5.1]{UNWT25}. By combining this key ingredient with the space-efficient quantum singular value transformation in~\cite[Section 3]{LGLW23}, one can implement the unitary in \Cref{lemma:Uhlmann-transform} in a natural manner using quantum \textit{single-exponential} time and \textit{linear} space. We refer to this explicit implementation as the \textit{algorithmic Uhlmann transform}: 

\begin{theorem}[Algorithmic Uhlmann transform]
    \label{thm:algorithmic-Uhlmann-transform}
    Let $\rho_0$ and $\rho_1$ be quantum states prepared by $n$-qubit quantum circuits $Q_0$ and $Q_1$, and let $\ket{\psi_0}$ and $\ket{\psi_1}$ denote their purifications before tracing out the non-output qubits, as in \Cref{prob:fidelity-test}. An approximate version of the Uhlmann transform $U_{\star}$ specified in \Cref{lemma:Uhlmann-transform}, denoted by $\tilde{U}_{\star}$, can be implemented so that
    \begin{equation}
        \label{eq:algo-Uhlmann-bound}
        \F^2(\rho_0,\rho_1) - 2^{-n} \leq \abs*{\bra{\psi_0}^{\sfA\sfR} \rbra*{I^\sfA \!\otimes\! \tilde{U}_{\star}^{\sfR}} \ket{\psi_1}^{\sfA\sfR}}^2 \leq \F^2(\rho_0,\rho_1).
    \end{equation}    
    The quantum circuit implementation of $\tilde{U}_{\star}$, acting on $O(n)$ qubits, requires $2^{O(n)}$ queries to the quantum circuits $Q_0$ and $Q_1$, as well as $2^{O(n)}$ one- and two-qubit quantum gates. Moreover, the circuit description can be computed in deterministic time $2^{O(n)}$ and space $O(n)$. 
\end{theorem}

\begin{remark}[\GapFEstlog{} is \BQL{}-complete]
    In analogy with the connection between the algorithmic Holevo--Helstrom measurement (\Cref{thm:algo-HH-meas}) and the \BQL{} containment of \GapQSDlog{} proven in~\cite[Section 4.2]{LGLW23}, one can establish that \GapFEstlog{} is in \BQL{} by adapting the space-efficient QSVT-based approach in \Cref{thm:algorithmic-Uhlmann-transform}. Here, \GapFEstlog{} denotes the space-bounded version of \GapFEst{}, where the state-preparation circuits have input length $O(\log n)$. Moreover, \GapFEstlog{} is \BQL{}-complete,\footnote{The \BQL{} hardness of \GapFEstlog{} holds even for pure states, which follows from combining~\cite[Lemma 4.23]{LGLW23} with the fact that \BQL{} is closed under complement~\cite[Corollary 4.8]{Wat99}.} and we leave a formal proof for future work. 
\end{remark}

By inspecting the (honest-verifier) quantum statistical zero-knowledge protocol (``closeness test'') for $\QSC[\beta,\alpha]$, serving as the complement of \QSD{}, with constants satisfying $\alpha^2>\beta$ in~\cite[Section 4.3]{Watrous02}, we establish a slightly improved upper bound for \GapFEst{}:
\begin{theorem}[\GapFEst{} is in \QIPtwo{} with a quantum linear-space honest prover]
    \label{thm:GapFEst-in-QIP(2)-bounded-prover}
    There exists a two-message quantum interactive proof system for $\FEst[\alpha(n),\beta(n)]$ with completeness $c(n) = \alpha(n)-2^{-n}$ and soundness $s(n) = \beta(n)$. Moreover, a near-optimal prover strategy for this proof system can be implemented in quantum single-exponential time and linear space. 

    \noindent Consequently, for any $\alpha(n)$ and $\beta(n)$ satisfying $\alpha(n)-\beta(n) \geq 1/\poly(n)$, 
    \[\FEst[\alpha(n),\beta(n)] \text{ is in } \QIPtwo{} \text{ with a quantum } O(n') \text{ space honest prover,}\] 
    where $n'$ is the total input length of the quantum circuits that prepare the corresponding tuple of quantum states.
\end{theorem}

\vspace{1em}
In the remainder of this section, we first present the proofs of \Cref{thm:algorithmic-Uhlmann-transform} and \Cref{thm:GapFEst-in-QIP(2)-bounded-prover} in \Cref{subsec:algo-Uhlmann-proof} and \Cref{subsec:GapFEst-in-QIP(2)-proof}, respectively. We then discuss the implications of \Cref{thm:GapFEst-in-QIP(2)-bounded-prover} for promise problems defined with respect to the trace distance in \Cref{subsec:implications-traceDist-from-fidelity}, which are closely related to the complexity classes \QSZK{} and \NIQSZK{}. 

\subsection{Algorithmic Uhlmann transform: Proof of \texorpdfstring{\Cref{thm:algorithmic-Uhlmann-transform}}{Theorem 4.5}}
\label{subsec:algo-Uhlmann-proof}

To implement our algorithmic Uhlmann transform, we begin by stating an exact block-encoding of $\Tr_{\sfA'} \rbra[\big]{ \ketbra{\psi_0}{\psi_1}^{\sfA'\sfR} }$, as specified in \Cref{lemma:Uhlmann-transform-block-encoding}. The proof of \Cref{lemma:Uhlmann-transform-block-encoding} appears at the beginning of \cite[Section 5.1]{UNWT25}, specifically from Equation (37) to Equation (42). 

\begin{lemma}[Exact block-encoding of ${\Tr_{\sfA} \rbra[\big]{ \ketbra{\psi_0}{\psi_1}^{\sfA\sfR} }}$, adapted from~{\cite[Section 5.1]{UNWT25}}]
    \label{lemma:Uhlmann-transform-block-encoding}
    Let $X_\mathrm{Uhl} \coloneqq \Tr_{\sfA} \rbra[\big]{ \ketbra{\psi_0}{\psi_1}^{\sfA\sfR} }$ be a linear operator on register $\sfR$ such that the unitary $\sign^\SV\rbra*{X_\mathrm{Uhl}}$ attains the maximum in Uhlmann's theorem (\Cref{thm:Uhlmann}). 
    Recall that $Q_0$ and $Q_1$ denote the state-preparation circuits of $\rho_0$ and $\rho_1$, as specified in \Cref{prob:fidelity-test}.     
    Then the unitary $W$ on registers $(\sfA',\sfR,\sfE)$, where $\sfA'$ is identical to $\sfA$ and $\sfE$ contains the same number of qubits as $\sfR$, is a $(1,n,0)$-block-encoding of $X_\mathrm{Uhl}$, given by 
    \[
    \bra{\bar{0}}^{\sfA'}\bra{\bar{0}}^{\sfE} W \ket{\bar{0}}^{\sfA'}\ket{\bar{0}}^{\sfE} = X_\mathrm{Uhl}, \text{ where } W \coloneqq \rbra[\big]{Q^{\sfA'\sfE}_1}^{\dagger} \rbra*{I^{\sfA'} \otimes \SWAP^{\sfR,\sfE}} Q_0^{\sfA'\sfE}. \]
    \noindent This block-encoding can be implemented using a single query to each state-preparation circuit $Q_0$ and $Q_1$, together with $O(n)$ one- and two-qubit quantum gates. 
\end{lemma}

Next, using the space-efficient polynomial approximation $P^\sign_{d'}$ of the sign function from \Cref{lemma:space-efficient-sign}, it suffices to implement another transform $\hat{U}_\star$ that is very close to $U_\star$: 
\[ \hat{U}_\star = P^{\SV}_{\sign,d'} \rbra*{ \Tr_{\sfA} \rbra*{ \ketbra{\psi_0}{\psi_1}^{\sfA\sfR} } }.\]
By applying the space-efficient QSVT associated with the polynomial $P^\sign_{d'}$ to the block-encoding of $\Tr_{\sfA} \rbra[\big]{ \ketbra{\psi_0}{\psi_1}^{\sfA\sfR} }$, we obtain a unitary $V_\star$ that is a block-encoding of $\hat{U}_\star$. In particular, $V_\star$ acts as an exact block-encoding of $\tilde{U}_{\star} \coloneqq \bra{\bar{0}} \tilde{V}_\star \ket{\bar{0}}$, which gives an approximate implementation of the Uhlmann transform. The difference between $\tilde{U}_{\star}$ and $\hat{U}_\star$ comes from the implementation error of the space-efficient QSVT.\footnote{It is worth noting that $\tilde{U}_{\star}$ is \textit{not} necessarily a unitary.} We now move on to the proof. 

\begin{proof}[Proof of \Cref{thm:algorithmic-Uhlmann-transform}]
    Our proof strategy is inspired by~\cite[Section 5.1]{UNWT25}, which provides a \BQP{} containment of the low-rank variant of \GapFEst{}. 
    \begin{figure}[!htp]
    \centering
    \begin{quantikz}[wire types={b,b,b,b}, classical gap=0.06cm, row sep=0.75em]
        \lstick{$\ket{\bar{0}}^\sfA~$} & \gate[2]{\quad Q_0 \quad} & & \gate[2]{\quad Q_1^\dagger \quad} & \meter{} & \setwiretype{c} \rstick{$z_\sfA$} \\
        \lstick{$\ket{\bar{0}}^\sfR~$} & & \gate[3]{ V_\star \approx U_{P^{\SV}_{\sign,d'}\rbra*{\Tr_{\sfA} \rbra[\big]{ \ketbra{\psi_0}{\psi_1}^{\sfA\sfR} }}}} & & \meter{} & \setwiretype{c} \rstick{$z_\sfR$} \\
        \lstick{$\ket{\bar{0}}^{\sfA'}$} & & & &\\
        \lstick{$\ket{\bar{0}}^\sfE~$} & & & &\\
    \end{quantikz}
    \caption{Algorithmic Uhlmann transform.}
    \label{fig:algo-Uhlmann-transform}
    \end{figure}
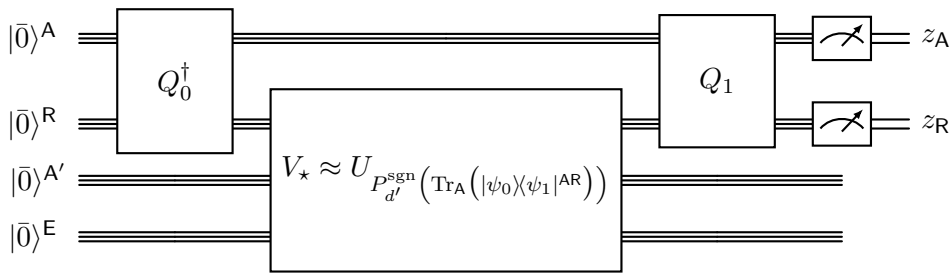
    Recall that $\rho_0$ and $\rho_1$ are $r(n)$-qubit quantum states on the register $\sfA$, each prepared by polynomial-size $n$-qubit quantum circuits $Q_0$ and $Q_1$, respectively, acting on the registers $(\sfA,\sfR)$, as defined in \Cref{prob:fidelity-test}. The overall procedure for estimating $\F^2(\rho_0,\rho_1)$ is presented in \Cref{fig:algo-Uhlmann-transform}, where the register $\sfA'$ contains $r$ qubits and the register $\sfE$ contains all remaining workspace qubits used by the implemented block-encoding, including the $n-r$ purification ancillary qubits and the additional ancillary qubits required by the QSVT implementation.
    In this procedure, acceptance occurs when the joint measurement outcome $(z_\sfA,z_\sfR)$ is the $n$-bit \textit{all-zero} string. 
    The central component in \Cref{fig:algo-Uhlmann-transform} is to implement the unitary $V_\star$, which can be achieved as follows:

    \begin{enumerate}[label={\upshape(\arabic*)}]
    \item Applying \Cref{lemma:Uhlmann-transform-block-encoding}, we obtain a $(1,n,0)$-block-encoding $W$ of $\XUhl = \Tr_{\sfA} \rbra*{ \ketbra{\psi_0}{\psi_1}^{\sfA\sfR} }$, using $O(1)$ queries to $Q_0$ and $Q_1$, together with $O(n)$ one- and two-qubit quantum gates.     
    \item Let $P_{d'}^{\sign} \in \mathbb{R}[x]$ be the degree-${d'}$ polynomial approximation of the sign function, as specified in \Cref{lemma:space-efficient-sign}.\footnote{This polynomial is obtained from some degree-$d$ averaged Chebyshev truncation with $d'=2d-1$.} 
    We choose parameters $\varepsilon \coloneqq 2^{-n}$, $\delta \coloneqq \frac{\varepsilon}{2^{r+3}}$, and $\epsilon \coloneqq \frac{\varepsilon}{64n\max\cbra{1,36\hat C_{\sign}+37,C_{\sign},\tilde C_{\sign}}}$.
    Consequently, the degree $d' \coloneqq \tilde C_{\sign}\cdot \frac{1}{\delta}\log\frac{1}{\epsilon}$ satisfies $d'=2^{O(n)}$, where $\tilde C_{\sign}$ is the constant from \Cref{lemma:space-efficient-sign}. This degree is fixed before choosing the numerical precision used to compute the coefficients. We also define $B_0 \coloneqq \rbra[\big]{36\hat{C}_\sign \log d' +37}^2 + C_\sign^2$ and $B_1 \coloneqq 2\rbra[\big]{36\hat{C}_\sign \log d' +37} + 3C_\sign$.\footnote{Since $r\le n$, we have $\log(1/\delta)=O(n)$ and $\log(1/\epsilon)=O(n)$, hence $\log d'=O(n)$. Therefore $B_0=O(n^2)$ and $B_1=O(n)$, with implied constants depending only on $\hat C_{\sign},C_{\sign},\tilde C_{\sign}$. Thus $B_0\epsilon^2+B_1\epsilon\le \varepsilon/4$ for all sufficiently large $n$, after increasing the absolute constant $64$ if necessary.\label{footnote:Uhlmann-parameters}}
    Applying the space-efficient QSVT associated with this polynomial (\Cref{lemma:sign-polynomial-implementation} with $\epsilon_1 \coloneqq 0$ and $\epsilon_2 \coloneqq \epsilon$), we obtain an implementation of the unitary $V_\star$. 
    \end{enumerate}

    \paragraph{Error analysis.} Noting that the resulting block-encoding $\tilde{V}_\star$ corresponds to a quantum channel acting on register $\sfR$, the upper bound in \Cref{eq:algo-Uhlmann-bound} follows immediately from \Cref{corr:stronger-Uhlmann}. Consequently, it suffices to prove the following weaker bound: 
    \[ \abs*{\F^2(\rho_0,\rho_1) - \abs*{\bra{\psi_0}^{\sfA\sfR} \rbra*{I^\sfA \!\otimes\! \tilde{U}_{\star}^{\sfR}} \ket{\psi_1}^{\sfA\sfR}}^2 } \leq 2^{-n}. \]
    
    To this end, we first bound the error introduced by the space-efficient polynomial approximation (\Cref{lemma:space-efficient-sign}). Consider the singular value decomposition $\XUhl =\sum_j s_j \ketbra{L_j}{R_j}$, where both $\cbra{\ket{L_j}}$ and $\cbra{\ket{R_j}}$ form orthonormal bases. We then define the index sets $\Lambda_0 \coloneqq \cbra{j: 0\leq s_j \leq \delta}$ and $\Lambda_+ \coloneqq \cbra{j: s_j > \delta}$. Using these definitions, we obtain the following bound: 
    \begin{subequations}
    \label{eq:algo-Uhlmann-polyError}
    \begin{align}
        &\abs*{ \F^2(\rho_0,\rho_1) - \abs*{ \bra{\psi_0} \rbra*{I^{\sfA}\!\otimes\!\hat{U}_\star^\sfR} \ket{\psi_1} }^2 }\\
        =~& \abs*{ \abs*{ \bra{\psi_0} \rbra*{I^{\sfA}\!\otimes\! \sign^{\SV}\rbra{\XUhl} } \ket{\psi_1} }^2 - \abs*{ \bra{\psi_0} \rbra*{I^{\sfA}\!\otimes\! P^\SV_{\sign,d'}\rbra{\XUhl} } \ket{\psi_1} }^2 }\\
        \leq~& \rbra*{ 2\abs*{ \bra{\psi_0} \rbra*{I^{\sfA}\!\otimes\! \sign^{\SV}\rbra{\XUhl} } \ket{\psi_1} } + \abs*{ \bra{\psi_0}  \rbra*{I^{\sfA}\!\otimes\! \rbra[\big]{\sign^{\SV}-P^\SV_{\sign,d'}}\rbra*{\XUhl}} \ket{\psi_1} } } \\
        & \quad \cdot \abs*{ \bra{\psi_0}  \rbra*{I^{\sfA}\!\otimes\! \rbra[\big]{\sign^{\SV}-P^\SV_{\sign,d'}}\rbra*{\XUhl}} \ket{\psi_1} }\\
        \leq~& \rbra*{2 + \abs*{ \bra{\psi_0}  \rbra*{I^{\sfA}\!\otimes\! \rbra[\big]{\sign^{\SV}-P^\SV_{\sign,d'}}\rbra*{\XUhl}} \ket{\psi_1} }}\\
        & \quad \cdot \abs*{ \bra{\psi_0} \rbra*{I^{\sfA}\!\otimes\! \rbra[\big]{\sign^{\SV}-P^\SV_{\sign,d'}}\rbra*{\XUhl}} \ket{\psi_1} }
    \end{align}
    \end{subequations}
    Here, the third line follows from the triangle inequality and the difference-of-squares formula, and the last line owes to the fact that $\abs*{ \bra{\psi_0} \rbra*{I^{\sfA}\!\otimes\! \sign^{\SV}\rbra{\XUhl} } \ket{\psi_1} } = \F(\rho_0,\rho_1) \leq 1$.

    Noting that any singular value $s_j$ of $\XUhl$ can be expressed as
    \[ s_j = \abs*{\Tr\rbra*{\ketbra{L_j}{R_j} \XUhl^{\dagger}}} = \abs*{\Tr\rbra*{\ketbra{L_j}{R_j} \Tr_{\sfA}\rbra*{\ketbra{\psi_1}{\psi_0}}}} = \abs*{\bra{\psi_0}\rbra[\big]{I^\sfA\otimes\ketbra{L_j}{R_j}}\ket{\psi_1}},\]
    it then follows from the singular value decomposition of $\XUhl$ that:
    \begin{subequations}
    \label{eq:algo-Uhlmann-polyError-SVbound}
    \begin{align}
        &\abs*{ \bra{\psi_0} \rbra*{I^{\sfA}\!\otimes\! \rbra[\big]{\sign^{\SV}-P^\SV_{\sign,d'}}\rbra*{\XUhl}} \ket{\psi_1} }\\
        \leq~& \sum_{j\in\Lambda_0} \abs*{s_j\sign(s_j)-s_j P^\sign_{d'}(s_j)} + \sum_{j\in\Lambda_+} \abs*{s_j\sign(s_j)-s_j P^\sign_{d'}(s_j)}\\
        \leq~& \sum_{j\in\Lambda_0} s_j\cdot\abs{\sign(s_j)-P^\sign_{d'}(s_j)} + \sum_{j\in\Lambda_+} s_j\cdot\abs{1-P^\sign_{d'}(s_j)}\\
        \leq~& \sum_{j\in\Lambda_0} 2 s_j + \sum_{j\in\Lambda_+} s_j C_{\sign}\epsilon\\
        \leq~& 2^{r+1}\delta + C_{\sign}\epsilon.
    \end{align}
    \end{subequations}
    Here, the second line uses the triangle inequality, the third line applies the sign function, the fourth line is guaranteed by \Cref{lemma:space-efficient-sign}, and the last line uses the facts that $\sum_j s_j = \F(\rho_0,\rho_1) \leq 1$ and that $\rank(\XUhl) \leq \min\cbra{\rank(\rho_0),\rank(\rho_1)} \leq 2^r$. 

    \vspace{1em}
    Next, we bound the error caused by space-efficient QSVT implementation (\Cref{lemma:sign-polynomial-implementation}): 
    \begin{subequations}
    \label{eq:algo-Uhlmann-implementError}
    \begin{align}
        & \abs*{ \abs*{ \bra{\psi_0} \rbra*{I^{\sfA}\!\otimes\!\hat{U}_\star^\sfR} \ket{\psi_1} }^2 - \abs*{ \bra{\psi_0} \rbra*{I^{\sfA}\!\otimes\!\tilde{U}_\star^\sfR} \ket{\psi_1} }^2 }\\
        =~& \abs*{ \abs*{ \bra{\psi_0} \rbra*{I^{\sfA}\!\otimes\! P^\SV_{\sign,d'}\rbra{\XUhl} } \ket{\psi_1} }^2 - \abs*{ \bra{\psi_0}\bra{\bar{0}}^{\sfA'}\bra{\bar{0}}^{\sfE} \rbra*{ I^{\sfA} \!\otimes\! V_{\star}} \ket{\psi_1}\ket{\bar{0}}^{\sfA'} \ket{\bar{0}}^{\sfE} }^2 }\\
        \leq~& \rbra*{ 2 + \abs*{ \bra{\psi_0} \rbra*{I^{\sfA}\!\otimes\! P^\SV_{\sign,d'}\rbra{\XUhl} } \ket{\psi_1} - \bra{\psi_0}\bra{\bar{0}}^{\sfA'}\bra{\bar{0}}^{\sfE} \rbra*{ I^{\sfA} \!\otimes\! V_{\star}} \ket{\psi_1}\ket{\bar{0}}^{\sfA'} \ket{\bar{0}}^{\sfE} } }\\
        & \quad \cdot \abs*{ \bra{\psi_0} \rbra*{I^{\sfA}\!\otimes\! P^\SV_{\sign,d'}\rbra{\XUhl} } \ket{\psi_1} - \bra{\psi_0}\bra{\bar{0}}^{\sfA'}\bra{\bar{0}}^{\sfE} \rbra*{ I^{\sfA} \!\otimes\! V_{\star}} \ket{\psi_1}\ket{\bar{0}}^{\sfA'} \ket{\bar{0}}^{\sfE} }.
    \end{align}
    \end{subequations}
    Here, the third line follows from the triangle inequality, the difference-of-squares formula, and the fact that $\abs*{ \bra{\psi_0}\bra{\bar{0}}^{\sfA'}\bra{\bar{0}}^{\sfE} \rbra*{ I^{\sfA} \!\otimes\! V_{\star}} \ket{\psi_1}\ket{\bar{0}}^{\sfA'} \ket{\bar{0}}^{\sfE} }  \leq 1$, since $I^{\sfA} \!\otimes\! V_{\star}$ is unitary and both $\ket{\psi_0}\ket{\bar{0}}^{\sfA'}\ket{\bar{0}}^{\sfE}$ and $\ket{\psi_1}\ket{\bar{0}}^{\sfA'}\ket{\bar{0}}^{\sfE}$ are pure states. In particular, it now suffices to bound the following: 
    \begin{subequations}
    \label{eq:algo-Uhlmann-implementError-HolderBound}
    \begin{align}
        & \abs*{ \bra{\psi_0} \rbra*{I^{\sfA}\!\otimes\! P^\SV_{\sign,d'}\rbra{\XUhl} } \ket{\psi_1} - \bra{\psi_0}\bra{\bar{0}}^{\sfA'}\bra{\bar{0}}^{\sfE} \rbra*{ I^{\sfA} \!\otimes\! V_{\star}} \ket{\psi_1}\ket{\bar{0}}^{\sfA'} \ket{\bar{0}}^{\sfE} }\\
        =~& \abs*{\Tr\rbra*{ \Tr_{\sfA}\rbra*{\ketbra{\psi_0}{\psi_1}}^{\dagger} P^\SV_{\sign,d'}\rbra{\XUhl} } - \Tr\rbra*{\Tr_{\sfA}\rbra*{\ketbra{\psi_0}{\psi_1}}^{\dagger} \bra{\bar{0}}^{\sfA'}\bra{\bar{0}}^{\sfE} V_{\star}\ket{\bar{0}}^{\sfA'} \ket{\bar{0}}^{\sfE} } } \\
        \leq~& \norm*{ P^\SV_{\sign,d'}\rbra{\XUhl} - \bra{\bar{0}}^{\sfA'}\bra{\bar{0}}^{\sfE} V_{\star}\ket{\bar{0}}^{\sfA'} \ket{\bar{0}}^{\sfE}} \cdot \norm*{\Tr_{\sfA}\rbra*{\ketbra{\psi_0}{\psi_1}}}_1\\
        =~& \norm*{ P^\SV_{\sign,d'}\rbra{\XUhl} - \bra{\bar{0}}^{\sfA'}\bra{\bar{0}}^{\sfE} V_{\star}\ket{\bar{0}}^{\sfA'} \ket{\bar{0}}^{\sfE}} \cdot \F(\rho_0,\rho_1)\\
        \leq~& \rbra*{36 \hat{C}_\sign \log{d'} + 37} \epsilon.
    \end{align}
    \end{subequations}
    Here, the third line follows from the H\"older inequality for Schatten norms (\Cref{lemma:matrix-Holder-inequality}), the fourth line utilizes the identity $\norm*{\Tr_{\sfA}\rbra*{\ketbra{\psi_0}{\psi_1}}}_1 = \max_U \abs*{\bra{\psi_1} I^\sfA\otimes U^{\sfR} \ket{\psi_0}} = \F(\rho_0,\rho_1)$, and the last line is guaranteed by \Cref{lemma:sign-polynomial-implementation}. 

    \vspace{1em}
    Combining the bounds in \Cref{eq:algo-Uhlmann-polyError,eq:algo-Uhlmann-polyError-SVbound,eq:algo-Uhlmann-implementError,eq:algo-Uhlmann-implementError-HolderBound}, we obtain the following error bound under the specified choice of parameters:
    \begin{align*}
        &\abs*{ \F^2(\rho_0,\rho_1) - \abs*{ \bra{\psi_0} \rbra*{I^{\sfA}\!\otimes\!\tilde{U}_\star^\sfR} \ket{\psi_1} }^2 }\\
        \leq~& \abs*{ \F^2(\rho_0,\rho_1) - \abs*{ \bra{\psi_0} \rbra*{I^{\sfA}\!\otimes\!\hat{U}_\star^\sfR} \ket{\psi_1} }^2 } + \abs*{ \abs*{ \bra{\psi_0} \rbra*{I^{\sfA}\!\otimes\!\hat{U}_\star^\sfR} \ket{\psi_1} }^2 - \abs*{ \bra{\psi_0} \rbra*{I^{\sfA}\!\otimes\!\tilde{U}_\star^\sfR} \ket{\psi_1} }^2 }\\
        \leq~& \rbra*{2+2^{r+1}\delta+C_\sign\epsilon} \rbra*{2^{r+1}\delta+C_\sign\epsilon} + \rbra*{2+\rbra[\big]{36\hat{C}_\sign\log{d'} + 37}\epsilon} \rbra[\big]{36\hat{C}_\sign\log{d'} + 37}\epsilon\\
        =~&\rbra*{ \rbra[\big]{36\hat{C}_\sign\log{d'} + 37}^2 + C_\sign^2 } \epsilon^2 + \rbra*{ 2\rbra[\big]{36\hat{C}_\sign\log{d'} + 37} + \rbra*{2+\frac{\varepsilon}{2}} C_\sign } \epsilon + \frac{\varepsilon}{2} + \frac{\varepsilon^2}{16}\\
        \leq~& B_0 \epsilon^2 + B_1 \epsilon + \frac{\varepsilon}{2} + \frac{\varepsilon^2}{16}\\
        \leq~& \frac{\varepsilon}{4}+\frac{\varepsilon}{2}+\frac{\varepsilon}{16}\\
        <~&\varepsilon.
    \end{align*}
    Here, the fourth line owes to $\delta = \varepsilon/2^{r+3}$, the fifth line follows from our choices of $B_0$ and $B_1$ together with $0 < \varepsilon \leq 1$, which implies $2\rbra[\big]{36\hat{C}_\sign\log d'+37}+\rbra*{2+\frac{\varepsilon}{2}}C_\sign\leq B_1$, and the sixth line uses the bound $B_0\epsilon^2+B_1\epsilon\leq \varepsilon/4$ in \Cref{footnote:Uhlmann-parameters}.

    \paragraph{Complexity analysis.} We complete the proof by analyzing the computational complexity of our construction. According to \Cref{lemma:sign-polynomial-implementation}, the procedure specified in \Cref{fig:algo-Uhlmann-transform} requires $O(n)$ qubits and $O(d^2\log{d}) \leq \widetilde{O}(2^{2r}/\varepsilon^{2}) \leq 2^{O(n)}$ queries to $Q_0$ and $Q_1$. In addition, the circuit description can be computed deterministically in time $\widetilde{O}(d^{9/2}/\varepsilon) = \widetilde{O}(2^{4.5r}/\varepsilon^{5.5}) \leq 2^{O(n)}$. 
\end{proof}

\subsection{A slightly improved upper bound for \texorpdfstring{\GapFEst{}}{GapF2Est}: Proof of \texorpdfstring{\Cref{thm:GapFEst-in-QIP(2)-bounded-prover}}{Theorem 4.7}}
\label{subsec:GapFEst-in-QIP(2)-proof}

We begin by presenting the quantum interactive proof system used in \Cref{thm:GapFEst-in-QIP(2)-bounded-prover}, as shown in \Cref{algo:GapFEST-in-QIP(2)-bounded-prover}. This proof system aligns with~\cite[Figure 3]{Watrous02}, with the new component being the honest prover's behavior. Specifically, the honest prover now utilizes the algorithmic Uhlmann transform $\tilde{U}_{\star}$ from \Cref{thm:algorithmic-Uhlmann-transform}, instead of the Uhlmann transform in \Cref{lemma:Uhlmann-transform}. 

\begin{algorithm}[!htp]
	\caption{\parbox{\linewidth}{Two-message proof system for \GapFEst{} (quantum linear-space prover).}}
	\label{algo:GapFEST-in-QIP(2)-bounded-prover}
    \SetEndCharOfAlgoLine{.}
    1. The verifier $\calV$ applies $Q_0$ to $\ket{0}^{\otimes n}$, and sends the \emph{non-output} qubits in register $\sfR$, while keeping output qubits in register $\sfA$\;
    2. The verifier $\calV$ receives the (possibly modified) qubits in register $\sfR$\;
    \begin{tcolorbox}[colback=gray!10, colframe=gray!40, boxrule=0pt, sharp corners, width=0.9\linewidth]
        The \emph{honest} prover $\calP$ applies the algorithmic Uhlmann transform $\tilde{U}_{\star}$ (\Cref{thm:algorithmic-Uhlmann-transform}) to the received qubits, and sends the resulting qubits back. 
    \end{tcolorbox}
    3. The verifier $\calV$ applies $Q_1^\dagger$ to the registers $(\sfA,\sfR)$, where register $\sfA$ contains the output qubits of $Q_0$ and register $\sfR$ contains the received qubits. $\calV$ then measures all qubits in $(\sfA,\sfR)$ in the computational basis and accepts if the measurement outcome is the $n$-bit all-zero string; otherwise, $\calV$ rejects. 
    \BlankLine
\end{algorithm}

Next, we complete the analysis of \Cref{algo:GapFEST-in-QIP(2)-bounded-prover}. 

\begin{proof}[Proof of \Cref{thm:GapFEst-in-QIP(2)-bounded-prover}]
    For \textit{yes} instances, where $\F^2(\rho_0,\rho_1) \geq \alpha(n)$, we note that the scenario in \Cref{algo:GapFEST-in-QIP(2)-bounded-prover} with the honest prover coincides with \Cref{prob:fidelity-test}. Therefore, the maximum acceptance probability $p_\acc$ of \Cref{algo:GapFEST-in-QIP(2)-bounded-prover} equals $\F^2(\rho_0,\rho_1)$, as guaranteed by Uhlmann's theorem (\Cref{thm:Uhlmann}). 
    Since the honest prover applies only an \textit{approximate} implementation of the Uhlmann transform (\Cref{thm:algorithmic-Uhlmann-transform}), it holds that 
    \[p_\acc \geq \F^2(\rho_0,\rho_1)-2^{-n} \geq \alpha(n)-2^{-n}\coloneqq c(n),\] 
    and the (honest) prover strategy described in \Cref{algo:GapFEST-in-QIP(2)-bounded-prover} is indeed implementable in quantum single-exponential time and linear space. 

    For \textit{no} instances, where $\F^2(\rho_0,\rho_1) \leq \beta(n)$, the argument follows immediately from \Cref{corr:stronger-Uhlmann} (see also \Cref{footnote:stronger-Uhlmann}), which gives $s(n) \coloneqq \beta(n)$. 

    \paragraph{Error reduction for \Cref{algo:GapFEST-in-QIP(2)-bounded-prover}.}
    To reduce the completeness and soundness errors, we apply error reduction for \QIPtwo{} (\Cref{lemma:QIPtwo-error-reduction}) to \Cref{algo:GapFEST-in-QIP(2)-bounded-prover} with $l(n)=n$. The resulting proof system $\protocol{\calP'}{\calV'}$ is obtained by substituting \Cref{algo:GapFEST-in-QIP(2)-bounded-prover} into \Cref{algo:QIP(2)-error-reduction}. 

    We now analyze the complexity of the honest prover $\calP'$. Noting that the resulting proof system $\protocol{\calP'}{\calV'}$ can be seen as testing the closeness of $t_0t_1$ pairs of quantum states $(\rho_0^{(j)},\rho_1^{(j)})$ for $1 \leq j \leq t_0t_1$, the total input length of the state-preparation circuits is $n\cdot t_0t_1 = 32n^4 q^4(n) \leq O(n^\tau) \coloneqq n'$ for some positive constant $\tau$, and the space complexity of the honest prover $\calP'$ is \textit{linear} in $n'$. Lastly, the desired completeness and soundness errors owe to the fact that $2^{-l\rbra[\big]{(n')^{1/\tau}}} \leq 1/3$ whenever $n'\geq \rbra*{\log{3}}^\tau$. 
\end{proof}

\subsection{Implications for closeness testing problems based on the trace distance}
\label{subsec:implications-traceDist-from-fidelity}
Combining the Fuchs--van de Graaf inequality (\Cref{lemma:traceDist-vs-fidelity}), which relates the squared fidelity to the trace distance, with the polarization lemma for the trace distance in the natural regime~\cite{HL26}, a direct calculation yields the following corollary:
\begin{corollary}[A slightly improved upper bound for \QSC{}]
    \label{corr:QSC-improved-bound}
    For any efficiently computable functions $\alpha(n)$ and $\beta(n)$ satisfying $\alpha(n)-\beta(n) \geq 1/O(\log{n})$, 
    \[\QSC[\beta(n),\alpha(n)] \text{ is in } \QIPtwo \text{ with a quantum } O(n') \text{ space honest prover}.\] 
    Here, $n'$ denotes the total input length of the state-preparation circuits. 
\end{corollary}

\begin{proof}
    For any pair of quantum states $\rho_0$ and $\rho_1$ whose purifications are efficiently preparable and for which either $\td(\rho_0,\rho_1) \leq \beta(n)$ or $\td(\rho_0,\rho_1) \geq \alpha(n)$, where $\alpha(n)-\beta(n) \geq 1/O(\log{n})$, the polarization lemma for the trace distance in the natural regime~\cite[Theorem 4.4]{HL26} allows us to construct a new pair of quantum states $\rho'_0$ and $\rho'_1$ whose purifications are also efficiently preparable and for which either $\td(\rho'_0,\rho'_1) \leq 1/6$ or $\td(\rho'_0,\rho'_1) \geq 5/6$. 
    
    Applying \Cref{lemma:traceDist-vs-fidelity}, the condition $\td(\rho_0,\rho_1) \leq \beta(n)$ for \textit{yes} instances implies that 
    \[\F^2(\rho'_0,\rho'_1) \geq \rbra*{1- \td(\rho'_0,\rho'_1)}^2 \geq (1-1/6)^2 = 25/36 \geq 2/3 \coloneqq c'.\] 
    Likewise, the condition $\td(\rho_0,\rho_1) \geq \alpha(n)$ for \textit{no} instances implies that 
    \[\F^2(\rho'_0,\rho'_1) \leq 1-\td^2(\rho'_0,\rho'_1) \leq 1 - (5/6)^2 = 11/36 \leq 1/3 \coloneqq s'.\] 

    The proof is therefore complete, because $c'-s' \geq 1/3 > 0$, and this follows by applying \Cref{thm:GapFEst-in-QIP(2)-bounded-prover} to $\FEst[2/3,1/3]$. 
\end{proof}

Since the \coQSZK{}-hard regime of \QSC{}, implicitly specified in~\cite[Section 5]{Watrous02}, is covered by \Cref{corr:QSC-improved-bound}, applying the complement gives the \coQIPtwo{} part in \Cref{corr:QSZK-upper-bound}. In addition, by fixing $\rho_0$ in \QSC{} to be the maximally mixed state and choosing the state-preparation circuit $Q_0$ to create EPR pairs across the registers $\sfA$ and $\sfR$, \Cref{corr:QSC-improved-bound} yields a two-message quantum interactive proof system in which the verifier's message consists exactly of half of EPR pairs, leading to \Cref{corr:NIQSZK-upper-bound}: 
\begin{corollary}[A slightly improved upper bound for \QSCMM{}]
    \label{corr:QSCMM-improved-bound}
    For any efficiently computable functions $\alpha(n)$ and $\beta(n)$ satisfying $\alpha(n)-\beta(n) \geq 1/O(\log n)$,
    \[\QSCMM[\beta(n),\alpha(n)] \text{ is in } \qqQAM \text{ with a quantum } O(n') \text{ space honest prover}.\] 
    Here, $n'$ denotes the total input length of the state-preparation circuits.  
\end{corollary}


\section{Near-optimality of the linear-space bounds under \texorpdfstring{\countSH{}}{\#SH}}
\label{sec:near-optimal-space}

In this section, we establish that the linear-space bounds in \Cref{thm:algo-HH-meas,thm:algorithmic-Uhlmann-transform} are \emph{near-optimal} under the Counting Satisfiability Hypothesis (\countSH{}). Specifically, neither the algorithmic Holevo--Helstrom measurement nor the algorithmic Uhlmann transform admits an $O(n^{1-\eta})$-space implementation for any fixed $\eta>0$, while achieving the same additive error $2^{-n}$. Here, space complexity includes both the classical workspace used to compute the circuit description and the qubits used by the resulting circuit implementation.

To establish such conditional lower bounds, we use \countSH{}, introduced in~\cite[Section~2]{SH26}, as the counting version of the Satisfiability Hypothesis~\cite{SH90}. Let $\phi \coloneqq C_1\wedge\cdots\wedge C_\ell$ be a SAT formula on $m$ variables. Let $\iota_j$ be the number of (possibly repeated) literals occurring in the clause $C_j$, and let $\iota \coloneqq \sum_{j=1}^\ell \iota_j$ be the total number of literal occurrences. Let $L\coloneqq m+\ell+\iota$ denote the formula size. Consequently, the binary encoding length is $O(L\log(L+2))$.\footnote{This logarithmic factor does not affect \Cref{conj:sh-counting}, as explained in~\cite[Footnotes~1~and~2]{SH26}.}

\begin{conjecture}[Counting Satisfiability Hypothesis,~{\cite[Section 2]{SH26}}]
    \label{conj:sh-counting}
    For every constant $\epsilon>0$, there is no deterministic algorithm that counts the satisfying assignments of every SAT formula of size $L$ in time $O(2^{L^{1-\epsilon}})$.
\end{conjecture}

A (uniform) quantum circuit implementation consists of a deterministic classical algorithm that, given polynomial-size descriptions of the $n$-qubit preparation circuits $Q_0$ and $Q_1$, computes an explicit quantum circuit description in time $T_{\rmC}(n)$ and space $S_{\rmC}(n)$. The resulting quantum circuit has runtime $T_{\rmQ}(n)$ (i.e., the number of elementary gates) and space $S_{\rmQ}(n)$ (including ancillary qubits). For convenience, we choose the gateset that includes Hadamard, $T$, and Toffoli gates. Our main result in this section is the following lower bound under \countSH{}: 

\begin{theorem}[Near-optimality of the linear-space bounds in \Cref{thm:algo-HH-meas,thm:algorithmic-Uhlmann-transform}]
    \label{thm:near-optimal-space}
    Under \countSH{}, for every fixed $\eta\in (0,1)$, assuming $S_{\rmC}(n)=\Omega(\log{n})$, neither the Holevo--Helstrom measurement satisfying \Cref{eq:algo-HH-bound} nor the Uhlmann transform satisfying \Cref{eq:algo-Uhlmann-bound} admits a quantum circuit implementation satisfying the time--space bound 
    \[ S_{\rmQ}(n)+\log\rbra*{T_{\rmC}(n)+T_{\rmQ}(n)+2}=O(n^{1-\eta}). \]
    Consequently, neither the Holevo--Helstrom measurement nor the Uhlmann transform admits a quantum circuit implementation satisfying the space bound 
    \[ S_{\rmC}(n)+S_{\rmQ}(n)=O(n^{1-\eta}). \] 
\end{theorem}

\subsection{From counting to quantum threshold tests}
\label{subsec:sh-threshold-tests}

We first reduce the problem of deciding whether $\#\SAT(\phi)>t$, where $\#\SAT(\phi)\coloneqq\sum_{\alpha\in\binset^m}\phi(\alpha)$ denotes the number of satisfying assignments, to each of the two algorithmic tasks: implementing the Holevo--Helstrom measurement (\Cref{thm:algo-HH-meas}) and implementing the Uhlmann transform (\Cref{thm:algorithmic-Uhlmann-transform}).

\begin{lemma}[Threshold tests]
    \label{lemma:SH-threshold-tests}
    For each of the algorithmic tasks underlying \Cref{thm:algo-HH-meas,thm:algorithmic-Uhlmann-transform}, there is a polynomial-time procedure that, given a SAT formula $\phi$ of size $L$ on $m$ variables and an integer $t\in[0,2^m)$, constructs state-preparation circuits $Q_0$ and $Q_1$. Each quantum circuit uses $O(L+1)$ gates and $n\coloneqq\kappa(L+1)\geq m+4$ qubits, where $\kappa$ is an absolute integer. Any quantum circuit implementation of the corresponding task for the resulting pair that satisfies \Cref{eq:algo-HH-bound} or \Cref{eq:algo-Uhlmann-bound}, respectively, yields a quantum test with acceptance probability $p_{\acc}$ satisfying
    \[
        \begin{aligned}
            \#\mathrm{SAT}(\phi)>t&\quad\Longrightarrow\quad p_{\acc}\geq3/4,\\
            \#\mathrm{SAT}(\phi)\leq t&\quad\Longrightarrow\quad p_{\acc}\leq1/4.
        \end{aligned}
    \]
    The test uses this circuit implementation once and requires only $O(1)$ additional gates and qubits.
\end{lemma}

\begin{proof}
    For an $m$-variable \SAT{} formula $\phi$ of size $L$, and for each integer $0\leq t<2^m$, we define a threshold $\Upsilon_t$ such that $\Upsilon_t>0$ if and only if $\#\SAT(\phi)>t$. 
    To this end, let $\alpha\in\binset^m$ be an assignment of $\phi$, with $\operatorname{int}(\alpha)$ denoting the integer representation of $\alpha$. We define a Boolean function $f_t\colon \binset\times\binset^m\times\binset \to \binset$, where $\bbI$ is the indicator function:
    \[
        \forall t\in\cbra{0,1,\cdots,2^m-1}, \quad f_t(j,\alpha,b)\coloneqq
        \begin{cases}
            \phi(\alpha),&j=0\\
            \bbI[2\operatorname{int}(\alpha)+b\geq2t+1],&j=1
        \end{cases}.
    \]
    Taking $j,b\in\binset$ and $\alpha\in\binset^m$ uniformly at random defines the threshold $\Upsilon_t$,
    \begin{equation}
        \label{eq:SH-shifted-count}
        \Upsilon_t \coloneqq 2 \Expectation_{\substack{j,b\in_{\rmU}\binset\\ \alpha\in_{\rmU}\binset^m}} \sbra*{f_t(j,\alpha,b)=1} - 1
        = \frac{2\#\SAT(\phi)-2t-1}{2^{m+1}}.
    \end{equation}
    The numerator in the last expression in \Cref{eq:SH-shifted-count} is odd and has absolute value between $1$ and $2^{m+1}-1$, so $2^{-m-1}\leq\abs{\Upsilon_t}<1$ and $\Upsilon_t$ has the desired sign for deciding whether $\#\SAT(\phi)>t$. Consequently, it suffices to determine the sign of $\Upsilon_t$.

    \vspace{1em}
    We now proceed with the threshold test. Note that evaluating any \SAT{} formula requires $O(L+1)$ bounded-fan-in and bounded-fan-out Boolean gates. Standard reversible simulation (e.g.,~\cite[Lemma 7.2]{KSV02}) gives $\ket{j,\alpha,b}\ket c\ket{\bar{0}} \mapsto \ket{j,\alpha,b}\ket{c\oplus f_t(j,\alpha,b)}\ket{\bar{0}}$ using $O(L+1)$ gates and $O(L+1)$ ancillary qubits, with the resulting circuit constructible in polynomial time. 
    Thus, for $\varepsilon\coloneqq 2^{-n}$, it holds that
    \begin{equation}
        \label{eq:SH-relative-error}
        \frac{\varepsilon}{\abs{\Upsilon_t}}\leq2^{m+1-n}\leq\frac18.
    \end{equation}
    We then complete the explicit constructions for these two algorithmic tasks.
    
    \paragraph{Holevo--Helstrom measurement.}
    We construct state-preparation circuits $Q_0$ and $Q_1$ acting on the registers $(\sfA,\sfR)$, which prepare the purifications $\ket{\psi_0}$ and $\ket{\psi_1}$, respectively:
    \[ 
        \ket{\psi_0}^{\sfA\sfR}\coloneqq 2^{-\frac{m+2}{2}}\!\sum_{\substack{j,b\in\binset\\ \alpha\in\binset^m}} \ket{1-f_t(j,\alpha,b)}^{\sfA} \ket{j,\alpha,b,\bar0}^{\sfR}
        \quad\text{and}\quad
        \ket{\psi_1}^{\sfA\sfR}\coloneqq \rbra[\big]{X^{\sfA}\otimes I^{\sfR}}\ket{\psi_0}^{\sfA\sfR}.
    \]
    Here, $\sfA$ denotes a single-qubit register and $\sfR$ contains the remaining qubits. 
    By \Cref{eq:SH-shifted-count}, $\Expectation_{j,\alpha,b}\sbra{f_t(j,\alpha,b)=z}=\rbra*{1-(-1)^z\Upsilon_t}/2$ for $z\in\binset$, so tracing out the register $\sfR$ gives
    \begin{equation}
        \label{eq:HH-meas-lower-bound-states}
        \rho_0 =\frac{1+\Upsilon_t}{2}\ketbra{0}{0} + \frac{1-\Upsilon_t}{2}\ketbra{1}{1} =\frac{I+\Upsilon_t Z}{2}
        \quad\text{and}\quad
        \rho_1 = X \rho_0 X = \frac{I-\Upsilon_t Z}{2}.
    \end{equation}

    Our threshold test applies the implemented measurement $\{\tilde\Pi_0,\tilde\Pi_1\}$ to $\ket{0}$ and accepts on outcome $0$. Following \Cref{eq:HH-meas-lower-bound-states}, a direct calculation gives
    \begin{subequations}
    \label{eq:HH-meas-lower-bound}
    \begin{align}
        \sign\rbra{\Upsilon_t} \rbra[\big]{\braket{0}{\tilde\Pi_0}{0}-\braket{1}{\tilde\Pi_0}{1}} 
        &= \frac{\Tr\rbra[\big]{\tilde\Pi_0\rho_0} - \Tr\rbra[\big]{\tilde\Pi_0\rho_1}}{\abs{\Upsilon_t}} \\
        &\geq \frac{\td(\rho_0,\rho_1)-\varepsilon}{\abs{\Upsilon_t}} \\
        &= 1-\frac{\varepsilon}{\abs{\Upsilon_t}} 
        \geq \frac{7}{8}.
    \end{align}
    \end{subequations}
    Here, the first line uses the fact that $\rho_0-\rho_1=\Upsilon_t Z$ and $\Upsilon_t = \sign(\Upsilon_t)\abs{\Upsilon_t}$, the second line follows from \Cref{eq:algo-HH-bound}, and the last line uses $\td(\rho_0,\rho_1)=\abs{\Upsilon_t}$ and \Cref{eq:SH-relative-error}. 

    Since $0\preceq \tilde{\Pi}_0 \preceq I$, \Cref{eq:HH-meas-lower-bound} yields $\braket{0}{\tilde\Pi_0}{0} \geq 7/8$ when $\Upsilon_t>0$, while $\braket{0}{\tilde\Pi_0}{0} \leq \braket{1}{\tilde\Pi_0}{1} - 7/8 \leq 1/8$ when $\Upsilon_t<0$. Therefore, applying the measurement $\{\tilde\Pi_0,\tilde\Pi_1\}$ to $\ket{0}$ decides whether $\#\SAT(\phi)>t$ with error at most $1/8$, as desired.

    \paragraph{Uhlmann transform.}
    For convenience, we slightly abuse notation and reuse $Q_z$, $\ket{\psi_z}$, and $\rho_z$ for $z\in\binset$. 
    We construct state-preparation circuits $Q_0$ and $Q_1$, which prepare the following purifications, with $\sfF$ and $\sfR$ being single-qubit registers:
    \begin{equation}
        \label{eq:Uhlmann-transforms-purifications}
        \ket{\psi_0} = \frac{\ket{0}^{\sfF}\ket{\chi_0}\ket{0}^{\sfR}  + \ket{1}^{\sfF}\ket{\chi_0}\ket{1}^{\sfR}}{\sqrt2}
        \quad\text{and}\quad
        \ket{\psi_1} = \frac{\ket{0}^{\sfF}\ket{\chi_1}\ket{0}^{\sfR} + \ket{1}^{\sfF}\ket{\chi_0}\ket{1}^{\sfR}}{\sqrt2}.
    \end{equation}
    Here, the states $\ket{\chi_0}$ and $\ket{\chi_1}$ satisfy $\innerprod{\chi_0}{\chi_1}=\Expectation_{j,\alpha,b}\sbra*{2f_t(j,\alpha,b)-1}=\Upsilon_t$, where 
    \[ 
        \ket{\chi_0} \coloneqq 2^{-\frac{m+2}{2}} \!\sum_{\substack{j,b\in\binset\\ \alpha\in\binset^m}} \ket{j,\alpha,b}\ket{\bar0} 
        \quad\text{and}\quad
        \ket{\chi_1} \coloneqq 2^{-\frac{m+2}{2}} \!\sum_{\substack{j,b\in\binset\\ \alpha\in\binset^m}} \rbra*{2f_t(j,\alpha,b)-1}\ket{j,\alpha,b}\ket{\bar0}. 
    \]

    The circuit $Q_0$ prepares $\ket{\chi_0}$ using $m+2$ Hadamard gates and separately prepares an EPR pair on $(\sfF,\sfR)$. The circuit $Q_1$ additionally computes $f_t(j,\alpha,b)$, applies $XZX=\diag(-1,1)$ to the qubit storing $f_t(j,\alpha,b)$, controlled on $\sfF=0$, and then uncomputes $f_t(j,\alpha,b)$. 
    For the quantum states $\rho_b\coloneqq\Tr_{\sfR}\rbra*{\ketbra{\psi_b}{\psi_b}}$, it holds that
    \begin{equation}
        \label{eq:SH-fidelity}
        \XUhl=\Tr_{\sfA}\rbra*{\ketbra{\psi_0}{\psi_1}} = \frac{1}{2}\diag(\Upsilon_t,1)
        \quad\text{and}\quad
        \F(\rho_0,\rho_1)=\norm{\XUhl}_1 = \frac{1+\abs{\Upsilon_t}}{2}.
    \end{equation}
    Here, the register $\sfA$ contains all qubits except the qubit in $\sfR$, and the fidelity identity follows from \Cref{thm:Uhlmann,lemma:Uhlmann-transform}.

    \vspace{1em}
    In the threshold test, we determine the sign of $\Upsilon_t$ by applying the implemented Uhlmann transform to one half of the EPR pair $\ket{\Psi_\pm}^{\sfR'\sfR} \coloneqq \rbra[\big]{\ket{0}^{\sfR'}\ket{0}^{\sfR}\pm\ket{1}^{\sfR'}\ket{1}^{\sfR}}/\sqrt{2}$ on the registers $(\sfR,\sfR')$, where $\sfR'$ is a fresh single-qubit register. 
    As in \Cref{subsec:algo-Uhlmann-proof}, let $\tilde V_\star$ denote the implemented unitary on $(\sfR,\sfE)$, with all ancillary qubits collected in $\sfE$ and initialized to $\ket{\bar0}$. The operator $\tilde U_\star$ in \Cref{eq:algo-Uhlmann-bound} corresponds to the block-encoding 
    $\tilde U_\star = \rbra[\big]{I^{\sfR}\otimes\bra{\bar0}^{\sfE}} \tilde V_\star \rbra[\big]{I^{\sfR}\otimes\ket{\bar0}^{\sfE}}$. 
    
    To see why the threshold test reveals the sign of $\Upsilon_t$, note that \Cref{lemma:Uhlmann-transform} and \Cref{eq:SH-fidelity} give the ideal choice $U_\star=\diag\rbra*{\sign(\Upsilon_t),1}$. Thus, applying $I^{\sfR'}\otimes U_\star^{\sfR}$ maps $\ket{\Psi_+}$ to $\ket{\Psi_+}$ when $\Upsilon_t>0$, and to $-\ket{\Psi_-}$ when $\Upsilon_t<0$. The test is then completed by using a single Bell measurement to determine the sign exactly. 

    We next consider the setting with the approximate implementation $\tilde U_\star$. Our threshold test starts with $\ket{\Psi_+}^{\sfR'\sfR}\otimes\ket{\bar0}^{\sfE}$, applies $I^{\sfR'}\otimes\tilde V_\star^{\sfR\sfE}$, traces out $\sfE$, and finally measures $(\sfR',\sfR)$ in the Bell basis. The test accepts only on outcome $\Psi_+$. Let $\mu_\pm$ denote the probabilities of obtaining the measurement outcome $\Psi_\pm$, with $p_{\acc}=\mu_+$ and $\mu_++\mu_-\leq 1$. Consequently, it suffices to prove that $\mu_+\geq 3/4$ when $\Upsilon_t>0$, while $\mu_-\geq 3/4$ when $\Upsilon_t<0$. 

    Taking the positive sign when $\Upsilon_t>0$ and the negative sign when $\Upsilon_t<0$, we obtain
    \begin{subequations}
    \label{eq:Uhlmann-transform-lower-bound}
    \begin{align}
        \mu_\pm
        &\geq\abs*{
          \rbra[\big]{\bra{\Psi_\pm}\otimes\bra{\bar0}^{\sfE}}
          \rbra[\big]{I^{\sfR'}\otimes\tilde V_\star}
          \rbra[\big]{\ket{\Psi_+}\otimes\ket{\bar0}^{\sfE}}}^2\\
        &=\frac{1}{4} \abs*{\braket0{\tilde U_\star}0\pm\braket1{\tilde U_\star}1}^2\\
        &\geq1-\frac{(1+\abs{\Upsilon_t})^2
          -\abs*{\Upsilon_t\braket0{\tilde U_\star}0+\braket1{\tilde U_\star}1}^2}{2\abs{\Upsilon_t}}\\
        &=1-\frac{2}{\abs{\Upsilon_t}}\rbra*{
          \F(\rho_0,\rho_1)^2
          -\abs[\big]{\bra{\psi_0}\rbra[\big]{I^{\sfA}\otimes\tilde U_\star^{\sfR}}\ket{\psi_1}}^2}\\
        &\geq 1-\frac{2\varepsilon}{\abs{\Upsilon_t}}
          \geq\frac{3}{4}.
    \end{align}
    \end{subequations}
    Here, the first line keeps only the all-zero contribution of $\sfE$, the second line uses the facts that $\tilde{V}_\star$ is a block-encoding of $\tilde{U}_\star$ and that only diagonal terms survive in the register $\sfR'$, and the third line follows from $\abs[\big]{\braket b{\tilde U_\star}b}\leq1$ for $b\in\binset$ and the expansion
    \begin{align*}
        & (1+|\Upsilon_t|)^2 -\abs*{ \Upsilon_t\braket{0}{\tilde U_\star}{0} + \braket{1}{\tilde U_\star}{1} }^2 - 2|\Upsilon_t| \rbra*{ 1-\frac{1}{4} \abs*{ \braket{0}{\tilde U_\star}{0} \pm \braket{1}{\tilde U_\star}{1} }^2 }\\
        =~& |\Upsilon_t|^2 \rbra[\big]{ 1-\abs{\braket{0}{\tilde U_\star}{0}}^2 } + \rbra[\big]{ 1-\abs{\braket{1}{\tilde U_\star}{1}}^2 } + \frac{|\Upsilon_t|}{2} \abs[\big]{ \braket{0}{\tilde U_\star}{0} \mp \braket{1}{\tilde U_\star}{1} }^2 \geq 0.
    \end{align*}
    The fourth line follows from \Cref{eq:Uhlmann-transforms-purifications,eq:SH-fidelity}, and the last line uses \Cref{eq:algo-Uhlmann-bound,eq:SH-relative-error}.
    Therefore, our threshold test accepts with probability $p_\acc=\mu_+\geq 3/4$ when $\Upsilon_t>0$, while $p_\acc = \mu_+\leq 1-\mu_-\leq 1/4$ when $\Upsilon_t<0$. Since $\Upsilon_t>0$ if and only if $\#\SAT(\phi)>t$, the test decides whether $\#\SAT(\phi)>t$ with error at most $1/4$, as desired. 

    \paragraph{Complexity analysis.}
    Both pairs of state-preparation circuits use $O(1)$ calls to the reversible evaluation of $f_t(j,\alpha,b)$, together with $O(m+1)$ additional gates and qubits. Hence, each circuit uses $O(L+1)$ gates and qubits and can be constructed in polynomial time. Choosing a sufficiently large absolute integer $\kappa$ and adding $\ket{0}$-padding qubits gives $n=\kappa(L+1)\geq m+4$. In addition, each threshold test uses the corresponding implementation once and requires only $O(1)$ additional gates and qubits to prepare the input states and perform the final measurement. 
\end{proof}

\subsection{Deterministic simulation and the conditional lower bound}
\label{subsec:SH-main-proof}

In this subsection, we complete the proof of \Cref{thm:near-optimal-space}. To this end, we first use a deterministic simulation of quantum circuits that implement the quantum threshold tests underlying \Cref{lemma:SH-threshold-tests}, based on standard state-vector simulation as in \Cref{lemma:SH-space-simulation}, to simulate a single decision-oracle query asking whether $\#\SAT(\phi)>t$ for a size-$L$ \SAT{} formula $\phi$ and a given integer threshold $t$. Then, we adopt a \emph{counting-to-decision reduction} (\Cref{lemma:SH-counting-thresholds}) to recover the exact count $\#\SAT(\phi)$. Consequently, an $O(n^{1-\eta})$-space quantum circuit implementation of the quantum threshold tests underlying \Cref{lemma:SH-threshold-tests} implies a deterministic $2^{O(L^{1-\eta})}$-time algorithm for $\#\SAT$, which violates \countSH{}. 
We now present the two aforementioned key ingredients: 

\begin{lemma}[Deterministic simulation of quantum circuits, adapted from~{\cite[Section 1.3, Page 8]{AC17}}]
    \label{lemma:SH-space-simulation}
    Given an explicit description of an $n_{\rmQ}$-qubit quantum circuit consisting of $M$ gates from the gateset $\cbra{H,T,\Toffoli}$ and initialized in a computational-basis state, the probability of any specified computational-basis outcome on designated output qubits can be approximated to additive error $1/8$ in deterministic time $(M+2^{n_{\rmQ}})^{O(1)}$.
\end{lemma}

\begin{lemma}[Counting-to-decision reduction, adapted from~{\cite[Theorem 4.14]{Fortnow97}}]
    \label{lemma:SH-counting-thresholds}
    Let \GapP{} denote the class of differences of \countP{} functions. Then, $\Ptime^{\PP}=\Ptime^{\GapP}$. 
    
    \noindent More explicitly, for every \GapP{} function $g$, the language $\calL_g \coloneqq \cbra{(x,t) \colon g(x)>t}$ is in \PP{}, where $t$ is an integer given in binary, and $g(x)$ can be computed exactly in polynomial time with oracle access to $\calL_g$. In particular, $\#\SAT(\phi) \in \cbra{0,\cdots,2^m}$ can be recovered using at most $m+1$ queries of the form $\#\SAT(\phi)>t$, with $0\leq t<2^m$, and $\poly(m+1)$ additional classical running time. 
\end{lemma}

\begin{proof}[Proof of \Cref{thm:near-optimal-space}]
    We prove the lower bound under \countSH{} by contradiction. 
    Suppose that either algorithmic task admits a quantum circuit implementation satisfying 
    \begin{equation}
        \label{eq:near-opt-space-assumed-bound}
        S_{\rmQ}(n)+\log\rbra*{T_{\rmC}(n)+T_{\rmQ}(n)+2}=O(n^{1-\eta}).
    \end{equation}
    For each query of the form $\#\mathrm{SAT}(\phi)>t$, where $\phi$ is a size-$L$ \SAT{} formula, \Cref{lemma:SH-threshold-tests} constructs state-preparation circuits $Q_0$ and $Q_1$, each acting on $n=\kappa(L+1)$ qubits, in polynomial time. The quantum test in \Cref{lemma:SH-threshold-tests} uses the assumed Holevo--Helstrom measurement or Uhlmann transform implementation for the pair $(Q_0,Q_1)$ once, together with only $O(1)$ additional gates and qubits. Thus, the threshold test circuit acts on $n_{\rmQ}=S_{\rmQ}(n)+O(1)$ qubits and contains $M=T_{\rmQ}(n)+O(1)$ gates.
    Altogether, \Cref{lemma:SH-space-simulation} gives a deterministic simulation in time
    \begin{align*}
        T_{\mathrm{threshold}}
        &\leq\poly(L+1)+T_{\rmC}(n)+(M+2^{n_{\rmQ}})^{O(1)}\\
        &\leq\poly(L+1)+(T_{\rmC}(n)+T_{\rmQ}(n)+2)^{O(1)} +2^{O(S_{\rmQ}(n))}\\
        &\leq\poly(L+1)+2^{O\rbra*{S_{\rmQ}(n)+\log(T_{\rmC}(n)+T_{\rmQ}(n)+2)}}\\
        &\leq\poly(L+1)+2^{O(L^{1-\eta})}.
    \end{align*}
    The last line uses the assumed time--space bound in \Cref{eq:near-opt-space-assumed-bound} and $n=\kappa(L+1)$.
    
    Let $p_{\acc}$ denote the acceptance probability of the threshold test, and let $\hat{p}_{\acc}$ be the numerical approximation computed by the deterministic simulation in \Cref{lemma:SH-space-simulation}, which satisfies the error bound $|\hat{p}_{\acc}-p_{\acc}|\leq 1/8$. Together with \Cref{lemma:SH-threshold-tests}, we obtain $\hat{p}_{\acc}\geq 3/4 - 1/8 = 5/8$ when $\#\SAT(\phi)>t$, while $\hat{p}_{\acc}\leq 1/4 + 1/8 = 3/8$ when $\#\SAT(\phi)\leq t$. 
    Comparing $\hat{p}_{\acc}$ with $1/2$ therefore answers the query exactly. Using the counting-to-decision reduction (\Cref{lemma:SH-counting-thresholds}), at most $m+1\leq L+1$ such queries recover $\#\SAT(\phi)$, with total runtime $T_{\rm count}(L)$ satisfying, for some constants $c_1,c_2>0$ and all sufficiently large $L$, 
    \[ 
        \log T_{\rm count}(L) 
        \leq \log\rbra[\big]{ \poly(L+1)\,2^{O(L^{1-\eta})} }
        \leq c_1L^{1-\eta}+c_2\log(L+2)
        \leq L^{1-\eta/2}.
    \]
    Therefore, $T_{\mathrm{count}}(L)=\smash{O\rbra[\big]{ 2^{L^{1-\eta/2}} }}$, contradicting \countSH{} with $\epsilon=\eta/2$.

    To see the space lower bound, we note that the quantum circuit description is generated explicitly, and thus $T_{\rmQ}(n)\leq T_{\rmC}(n)$. A standard configuration counting argument gives 
    \[ \log(T_{\rmC}(n)+2)=O(S_{\rmC}(n)+\log n)=O(S_{\rmC}(n)) .\] 
    Here, the last equality uses $S_{\rmC}(n)=\Omega(\log n)$. If $S_{\rmC}(n)+S_{\rmQ}(n)=O(n^{1-\eta})$, then
    \[ S_{\rmQ}(n)+\log\rbra*{T_{\rmC}(n)+T_{\rmQ}(n)+2} = O(S_{\rmQ}(n)+S_{\rmC}(n)) = O(n^{1-\eta}), \]
    again contradicting \countSH{} with $\epsilon=\eta/2$.
\end{proof}


\section*{Acknowledgments}
\noindent
A preliminary version of \Cref{sec:algo-HH-meas} (whose informal theorems are summarized in \Cref{thm:GapQSD-in-QIP(2)-informal,thm:algo-HH-meas-informal}) appeared in Section 5 of the second arXiv version of~\cite{LGLW23} and in the second-named author's PhD thesis~\cite[Section 6.3]{Liu25}.
Part of the work of Qisheng Wang was done when the author was with the School of Informatics, University of Edinburgh, United Kingdom. 

The authors thank Harumichi Nishimura and Thomas Vidick for helpful feedback on a preliminary version of the manuscript.
This work was partially supported by MEXT Q-LEAP grant No.~\mbox{JPMXS0120319794}. 
FLG was also supported by JSPS KAKENHI grant No.~\mbox{JP24H00071}, JST ASPIRE grant No.~\mbox{JPMJAP2302}, and JST CREST grant No.~\mbox{JPMJCR24I4}.
YL was also supported in part by funding from the Swiss State Secretariat for Education, Research and Innovation (SERI), and in part by JSPS KAKENHI grant No.~\mbox{JP24H00071}.
QW was also supported by the Engineering and Physical Sciences Research Council under Grant \mbox{EP/X026167/1}. 

In addition, ChatGPT was used only for proofreading the manuscript, including identifying possible calculation errors and suggesting changes, with all corresponding changes verified and made by the authors.
The circuit diagrams were drawn using the Quantikz package~\cite{Kay18}.


\bibliographystyle{alphaurlQ}
\bibliography{QSZK-improved-upper-bounds}

\newcommand{\etalchar}[1]{$^{#1}$}
\DeclareRobustCommand{\dutchPrefix}[2]{#2}\providecommand{\dutchPrefix}[2]{#2}\renewcommand{\dutchPrefix}[2]{#2}\newcommand{\prelimVersion}[2]{Preliminary version in \textit{\MakeUppercase{#1} #2}}
\begin{thebibliography}{CFM{\dutchPrefix{Wolf}{d}}W10}

\bibitem[AC17]{AC17}
Scott Aaronson and Lijie Chen.
\newblock Complexity-theoretic foundations of quantum supremacy experiments.
\newblock In {\em Proceedings of the 32nd Computational Complexity Conference (CCC 2017)}. volume~79 of {\em Leibniz International Proceedings in Informatics}. pages 22:1--22:67. Schloss Dagstuhl -- Leibniz-Zentrum f{\"u}r Informatik. 2017.
\newblock \href {https://doi.org/10.4230/LIPIcs.CCC.2017.22} {\nolinkurl{doi:10.4230/LIPIcs.CCC.2017.22}}. \href {https://arxiv.org/abs/1612.05903} {\nolinkurl{arXiv:1612.05903}}. Appearances:\!

\bibitem[AH91]{AH91}
William Aiello and Johan H{\aa}stad.
\newblock Statistical zero-knowledge languages can be recognized in two rounds.
\newblock {\em Journal of Computer and System Sciences}. 42(3):327--345. 1991.
\newblock \href {https://doi.org/10.1016/0022-0000(91)90006-Q} {\nolinkurl{doi:10.1016/0022-0000(91)90006-Q}}. \prelimVersion{FOCS}{1987}. Appearances:\!

\bibitem[AJL09]{AJL09}
Dorit Aharonov, Vaughan Jones, and Zeph Landau.
\newblock A polynomial quantum algorithm for approximating the {Jones} polynomial.
\newblock {\em Algorithmica}. 55(3):395--421. 2009.
\newblock \href {https://doi.org/10.1007/s00453-008-9168-0} {\nolinkurl{doi:10.1007/s00453-008-9168-0}}. \prelimVersion{STOC}{2006}. \href {https://arxiv.org/abs/quant-ph/0511096} {\nolinkurl{arXiv:quant-ph/0511096}}. Appearances:\!

\bibitem[BCC{\etalchar{+}}15]{BCC+15}
Dominic~W Berry, Andrew~M Childs, Richard Cleve, Robin Kothari, and Rolando~D Somma.
\newblock Simulating {Hamiltonian} dynamics with a truncated {Taylor} series.
\newblock {\em Physical Review Letters}. 114(9):090502. 2015.
\newblock \href {https://doi.org/10.1103/PhysRevLett.114.090502} {\nolinkurl{doi:10.1103/PhysRevLett.114.090502}}. \href {https://arxiv.org/abs/1412.4687} {\nolinkurl{arXiv:1412.4687}}. Appearances:\!

\bibitem[BDRV19]{BDRV19}
Itay Berman, Akshay Degwekar, Ron~D Rothblum, and Prashant~Nalini Vasudevan.
\newblock Statistical difference beyond the polarizing regime.
\newblock In {\em Theory of Cryptography Conference}. pages 311--332. Springer. 2019.
\newblock \href {https://doi.org/10.1007/978-3-030-36033-7\_12} {\nolinkurl{doi:10.1007/978-3-030-36033-7\_12}}. \href {https://eccc.weizmann.ac.il/report/2019/038} {\nolinkurl{ECCC:TR19-038}}. Appearances:\!

\bibitem[BEM{\etalchar{+}}26]{BEM+23}
John Bostanci, Yuval Efron, Tony Metger, Alexander Poremba, Luowen Qian, and Henry Yuen.
\newblock Unitary complexity and the {Uhlmann} transformation problem.
\newblock In {\em Proceedings of the 17th Innovations in Theoretical Computer Science Conference ({ITCS} 2026)}. volume 362 of {\em LIPIcs}. pages 24:1--24:17. Schloss Dagstuhl - Leibniz-Zentrum f{\"{u}}r Informatik. 2026.
\newblock \href {https://doi.org/10.4230/LIPICS.ITCS.2026.24} {\nolinkurl{doi:10.4230/LIPICS.ITCS.2026.24}}. \href {https://arxiv.org/abs/2306.13073} {\nolinkurl{arXiv:2306.13073}}. Appearances:\!

\bibitem[BKT20]{BKT20}
Mark Bun, Robin Kothari, and Justin Thaler.
\newblock The polynomial method strikes back: tight quantum query bounds via dual polynomials.
\newblock {\em Theory of Computing}. 16(10):1--71. 2020.
\newblock \href {https://doi.org/10.4086/toc.2020.v016a010} {\nolinkurl{doi:10.4086/toc.2020.v016a010}}. \prelimVersion{STOC}{2018}. \href {https://arxiv.org/abs/1710.09079} {\nolinkurl{arXiv:1710.09079}}. Appearances:\!

\bibitem[BMY26]{BMY25}
John Bostanci, Tony Metger, and Henry Yuen.
\newblock Local transformations of bipartite entanglement are rigid.
\newblock In {\em Proceedings of the 17th Innovations in Theoretical Computer Science Conference (ITCS 2026)}. volume 362 of {\em LIPIcs}. pages 26:1--26:8. Schloss Dagstuhl - Leibniz-Zentrum f{\"{u}}r Informatik. 2026.
\newblock \href {https://doi.org/10.4230/LIPICS.ITCS.2026.26} {\nolinkurl{doi:10.4230/LIPICS.ITCS.2026.26}}. \href {https://arxiv.org/abs/2509.05257} {\nolinkurl{arXiv:2509.05257}}. Appearances:\!

\bibitem[BST10]{BASTS10}
Avraham {Ben-Aroya}, Oded Schwartz, and Amnon {Ta-Shma}.
\newblock Quantum expanders: Motivation and construction.
\newblock {\em Theory of Computing}. 6(1):47--79. 2010.
\newblock \href {https://doi.org/10.4086/toc.2010.v006a003} {\nolinkurl{doi:10.4086/toc.2010.v006a003}}. \prelimVersion{CCC}{2008}. Appearances:\!

\bibitem[CCKV08]{CCKV08}
Andr{\'e} Chailloux, Dragos~Florin Ciocan, Iordanis Kerenidis, and Salil Vadhan.
\newblock Interactive and noninteractive zero knowledge are equivalent in the help model.
\newblock In {\em Theory of Cryptography Conference}. pages 501--534. Springer. 2008.
\newblock \href {https://doi.org/10.1007/978-3-540-78524-8\_28} {\nolinkurl{doi:10.1007/978-3-540-78524-8\_28}}. \href {https://eprint.iacr.org/2007/467} {\nolinkurl{IACR ePrint:2007/467}}. Appearances:\!

\bibitem[CFM{\dutchPrefix{Wolf}{d}}W10]{CFMdW10}
Sourav Chakraborty, Eldar Fischer, Arie Matsliah, and Ronald {\dutchPrefix{Wolf}{d}}e~Wolf.
\newblock New results on quantum property testing.
\newblock In {\em IARCS Annual Conference on Foundations of Software Technology and Theoretical Computer Science (FSTTCS 2010)}. volume~8 of {\em LIPIcs}. pages 145--156. Schloss Dagstuhl - Leibniz-Zentrum f{\"{u}}r Informatik. 2010.
\newblock \href {https://doi.org/10.4230/LIPICS.FSTTCS.2010.145} {\nolinkurl{doi:10.4230/LIPICS.FSTTCS.2010.145}}. \href {https://arxiv.org/abs/1005.0523} {\nolinkurl{arXiv:1005.0523}}. Appearances:\!

\bibitem[CK08]{CK08}
Andr{\'{e}} Chailloux and Iordanis Kerenidis.
\newblock Increasing the power of the verifier in quantum zero knowledge.
\newblock In {\em Proceedings of the {IARCS} Annual Conference on Foundations of Software Technology and Theoretical Computer Science ({FSTTCS} 2008)}. LIPIcs. pages 95--106. 2008.
\newblock \href {https://doi.org/10.4230/LIPICS.FSTTCS.2008.1744} {\nolinkurl{doi:10.4230/LIPICS.FSTTCS.2008.1744}}. \href {https://arxiv.org/abs/0711.4032} {\nolinkurl{arXiv:0711.4032}}. Appearances:\!

\bibitem[CWZ25]{CWZ25}
Kean Chen, Qisheng Wang, and Zhicheng Zhang.
\newblock A list of complexity bounds for property testing by quantum sample-to-query lifting.
\newblock {\em arXiv preprint}. 2025.
\newblock \href {https://arxiv.org/abs/2512.01971} {\nolinkurl{arXiv:2512.01971}}. Appearances:\!

\bibitem[F{\dutchPrefix{Graaf}{v}}dG99]{FvdG99}
Christopher~A Fuchs and Jeroen {\dutchPrefix{Graaf}{v}}an~de Graaf.
\newblock Cryptographic distinguishability measures for quantum-mechanical states.
\newblock {\em IEEE Transactions on Information Theory}. 45(4):1216--1227. 1999.
\newblock \href {https://doi.org/10.1109/18.761271} {\nolinkurl{doi:10.1109/18.761271}}. \href {https://arxiv.org/abs/quant-ph/9712042} {\nolinkurl{arXiv:quant-ph/9712042}}. Appearances:\!

\bibitem[FL18]{FL18}
Bill Fefferman and Cedric Yen-Yu Lin.
\newblock A complete characterization of unitary quantum space.
\newblock In {\em Proceedings of the 9th Innovations in Theoretical Computer Science Conference}. volume~94. page~4. 2018.
\newblock \href {https://doi.org/10.4230/LIPIcs.ITCS.2018.4} {\nolinkurl{doi:10.4230/LIPIcs.ITCS.2018.4}}. \href {https://arxiv.org/abs/1604.01384} {\nolinkurl{arXiv:1604.01384}}. Appearances:\!

\bibitem[For87]{Fortnow87}
Lance Fortnow.
\newblock The complexity of perfect zero-knowledge.
\newblock In {\em Proceedings of the 19th Annual ACM Symposium on Theory of Computing}. pages 204--209. 1987.
\newblock \href {https://doi.org/10.1145/28395.28418} {\nolinkurl{doi:10.1145/28395.28418}}. Appearances:\!

\bibitem[For97]{Fortnow97}
Lance Fortnow.
\newblock Counting complexity.
\newblock In Lane~A. Hemaspaandra and Alan~L. Selman, editors, {\em Complexity Theory Retrospective II}. pages 81--107. Springer. 1997.
\newblock \href {https://doi.org/10.1007/978-1-4612-1872-2\_4} {\nolinkurl{doi:10.1007/978-1-4612-1872-2\_4}}. Appearances:\!

\bibitem[GKL21]{GKL19}
Ayal Green, Guy Kindler, and Yupan Liu.
\newblock Towards a quantum-inspired proof for $\mathsf{IP} = \mathsf{PSPACE}$.
\newblock {\em Quantum Information \& Computation}. 21(5{\&}6):377--386. 2021.
\newblock \href {https://doi.org/10.26421/QIC21.5-6-2} {\nolinkurl{doi:10.26421/QIC21.5-6-2}}. \href {https://arxiv.org/abs/1912.11611} {\nolinkurl{arXiv:1912.11611}}. Appearances:\!

\bibitem[GKR15]{GKR15}
Shafi Goldwasser, Yael~Tauman Kalai, and Guy~N. Rothblum.
\newblock Delegating computation: interactive proofs for muggles.
\newblock {\em Journal of the ACM}. 62(4):1--64. 2015.
\newblock \href {https://doi.org/10.1145/2699436} {\nolinkurl{doi:10.1145/2699436}}. \prelimVersion{STOC}{2008}. \href {https://eccc.weizmann.ac.il/report/2017/108} {\nolinkurl{ECCC:TR17-108}}. Appearances:\!

\bibitem[GMW91]{GMW91}
Oded Goldreich, Silvio Micali, and Avi Wigderson.
\newblock Proofs that yield nothing but their validity for all languages in $\mathsf{NP}$ have zero-knowledge proof systems.
\newblock {\em Journal of the {ACM}}. 38(3):691--729. 1991.
\newblock \href {https://doi.org/10.1145/116825.116852} {\nolinkurl{doi:10.1145/116825.116852}}. \prelimVersion{FOCS}{1986}. Appearances:\!

\bibitem[Gol18]{Goldreich18}
Oded Goldreich.
\newblock On doubly-efficient interactive proof systems.
\newblock {\em Foundations and Trends{\textregistered} in Theoretical Computer Science}. 13(3):158--246. 2018.
\newblock \href {https://doi.org/10.1561/0400000084} {\nolinkurl{doi:10.1561/0400000084}}. \href {https://eccc.weizmann.ac.il/report/2017/017} {\nolinkurl{ECCC:TR17-017}}. Appearances:\!

\bibitem[GP22]{GP22}
Andr{\'a}s Gily{\'e}n and Alexander Poremba.
\newblock Improved quantum algorithms for fidelity estimation.
\newblock {\em arXiv preprint}. 2022.
\newblock \href {https://arxiv.org/abs/2203.15993} {\nolinkurl{arXiv:2203.15993}}. Appearances:\!

\bibitem[GRZ24]{GRZ23}
Uma Girish, Ran Raz, and Wei Zhan.
\newblock Quantum logspace computations are verifiable.
\newblock In {\em Proceedings of the 2024 Symposium on Simplicity in Algorithms}. pages 144--150. 2024.
\newblock \href {https://doi.org/10.1137/1.9781611977936.14} {\nolinkurl{doi:10.1137/1.9781611977936.14}}. \href {https://arxiv.org/abs/2307.11083} {\nolinkurl{arXiv:2307.11083}}. Appearances:\!

\bibitem[GSLW19]{GSLW19}
Andr{\'a}s Gily{\'e}n, Yuan Su, Guang~Hao Low, and Nathan Wiebe.
\newblock Quantum singular value transformation and beyond: exponential improvements for quantum matrix arithmetics.
\newblock In {\em Proceedings of the 51st Annual ACM SIGACT Symposium on Theory of Computing}. pages 193--204. 2019.
\newblock \href {https://doi.org/10.1145/3313276.3316366} {\nolinkurl{doi:10.1145/3313276.3316366}}. \href {https://arxiv.org/abs/1806.01838} {\nolinkurl{arXiv:1806.01838}}. Appearances:\!

\bibitem[GSS{\etalchar{+}}22]{GSS+22}
Sevag Gharibian, Miklos Santha, Jamie Sikora, Aarthi Sundaram, and Justin Yirka.
\newblock Quantum generalizations of the polynomial hierarchy with applications to $\mathsf{QMA(2)}$.
\newblock {\em computational complexity}. 31(2):1--52. 2022.
\newblock \href {https://doi.org/10.1007/s00037-022-00231-8} {\nolinkurl{doi:10.1007/s00037-022-00231-8}}. \prelimVersion{MFCS}{2018}. \href {https://arxiv.org/abs/1805.11139} {\nolinkurl{arXiv:1805.11139}}. Appearances:\!

\bibitem[GSV98]{GSV98}
Oded Goldreich, Amit Sahai, and Salil Vadhan.
\newblock Honest-verifier statistical zero-knowledge equals general statistical zero-knowledge.
\newblock In {\em Proceedings of the 30th Annual ACM Symposium on Theory of Computing}. pages 399--408. 1998.
\newblock \href {https://doi.org/10.1145/276698.276852} {\nolinkurl{doi:10.1145/276698.276852}}. Appearances:\!

\bibitem[Hel69]{Helstrom69}
Carl~W Helstrom.
\newblock Quantum detection and estimation theory.
\newblock {\em Journal of Statistical Physics}. 1:231--252. 1969.
\newblock \href {https://doi.org/10.1007/BF01007479} {\nolinkurl{doi:10.1007/BF01007479}}. Appearances:\!

\bibitem[HL26]{HL26}
Zihan Hu and Yupan Liu.
\newblock On quantum interactive proofs with a laconic prover.
\newblock {\em arXiv preprint}. 2026.
\newblock \href {https://arxiv.org/abs/2609.39495v2} {\nolinkurl{arXiv:2609.39495v2}}. Appearances:\!

\bibitem[Hol73]{Holevo73TraceDist}
Alexander~S Holevo.
\newblock Statistical decision theory for quantum systems.
\newblock {\em Journal of Multivariate Analysis}. 3(4):337--394. 1973.
\newblock \href {https://doi.org/10.1016/0047-259X(73)90028-6} {\nolinkurl{doi:10.1016/0047-259X(73)90028-6}}. Appearances:\!

\bibitem[HOZ26]{HOZ26}
Entong He, Fengning Ou, and Mingnan Zhao.
\newblock $\mathsf{unitaryQIP} = \mathsf{unitaryPSPACE}$.
\newblock {\em arXiv preprint}. 2026.
\newblock \href {https://arxiv.org/abs/2611.02819} {\nolinkurl{arXiv:2611.02819}}. Appearances:\!

\bibitem[JJUW11]{JJUW11}
Rahul Jain, Zhengfeng Ji, Sarvagya Upadhyay, and John Watrous.
\newblock $\mathsf{QIP}=\mathsf{PSPACE}$.
\newblock {\em Journal of the ACM}. 58(6):1--27. 2011.
\newblock \href {https://doi.org/10.1145/2049697.2049704} {\nolinkurl{doi:10.1145/2049697.2049704}}. \prelimVersion{STOC}{2010}. \href {https://arxiv.org/abs/0907.4737} {\nolinkurl{arXiv:0907.4737}}. Appearances:\!

\bibitem[Joz94]{Jozsa94}
Richard Jozsa.
\newblock Fidelity for mixed quantum states.
\newblock {\em Journal of modern optics}. 41(12):2315--2323. 1994.
\newblock \href {https://doi.org/10.1080/09500349414552171} {\nolinkurl{doi:10.1080/09500349414552171}}. Appearances:\!

\bibitem[JUW09]{JUW09}
Rahul Jain, Sarvagya Upadhyay, and John Watrous.
\newblock Two-message quantum interactive proofs are in $\mathsf{PSPACE}$.
\newblock In {\em Proceedings of the 50th Annual IEEE Symposium on Foundations of Computer Science}. pages 534--543. IEEE. 2009.
\newblock \href {https://doi.org/10.1109/FOCS.2009.30} {\nolinkurl{doi:10.1109/FOCS.2009.30}}. \href {https://arxiv.org/abs/0905.1300} {\nolinkurl{arXiv:0905.1300}}. Appearances:\!

\bibitem[Kay18]{Kay18}
Alastair Kay.
\newblock Tutorial on the quantikz package.
\newblock {\em arXiv preprint}. 2018.
\newblock \href {https://arxiv.org/abs/1809.03842} {\nolinkurl{arXiv:1809.03842}}. Appearances:\!

\bibitem[Kit95]{Kitaev95}
Alexei~Yu Kitaev.
\newblock Quantum measurements and the {Abelian} stabilizer problem.
\newblock {\em arXiv preprint}. 1995.
\newblock \href {https://arxiv.org/abs/quant-ph/9511026} {\nolinkurl{arXiv:quant-ph/9511026}}. Appearances:\!

\bibitem[KKMV09]{KKMV09}
Julia Kempe, Hirotada Kobayashi, Keiji Matsumoto, and Thomas Vidick.
\newblock Using entanglement in quantum multi-prover interactive proofs.
\newblock {\em computational complexity}. 18:273--307. 2009.
\newblock \href {https://doi.org/10.1007/s00037-009-0275-3} {\nolinkurl{doi:10.1007/s00037-009-0275-3}}. \prelimVersion{CCC}{2008}. \href {https://arxiv.org/abs/0711.3715} {\nolinkurl{arXiv:0711.3715}}. Appearances:\!

\bibitem[KLN19]{KLGN19}
Hirotada Kobayashi, François {Le Gall}, and Harumichi Nishimura.
\newblock Generalized quantum {Arthur--Merlin} games.
\newblock {\em {SIAM} Journal on Computing}. 48(3):865--902. 2019.
\newblock \href {https://doi.org/10.1137/17M1160173} {\nolinkurl{doi:10.1137/17M1160173}}. \prelimVersion{CCC}{2015}. \href {https://arxiv.org/abs/1312.4673} {\nolinkurl{arXiv:1312.4673}}. Appearances:\!

\bibitem[Kob03]{Kobayashi03}
Hirotada Kobayashi.
\newblock Non-interactive quantum perfect and statistical zero-knowledge.
\newblock In {\em Proceedings of the 14th International Symposium on Algorithms and Computation}. pages 178--188. Springer. 2003.
\newblock \href {https://doi.org/10.1007/978-3-540-24587-2\_20} {\nolinkurl{doi:10.1007/978-3-540-24587-2\_20}}. \href {https://arxiv.org/abs/quant-ph/0207158} {\nolinkurl{arXiv:quant-ph/0207158}}. Appearances:\!

\bibitem[KSV02]{KSV02}
A.~Yu. Kitaev, A.~H. Shen, and M.~N. Vyalyi.
\newblock {\em Classical and Quantum Computation}. volume~47 of {\em Graduate Studies in Mathematics}.
\newblock American Mathematical Society. 2002.
\newblock \href {https://doi.org/10.1090/gsm/047} {\nolinkurl{doi:10.1090/gsm/047}}. Appearances:\!

\bibitem[KW00]{KW00}
Alexei Kitaev and John Watrous.
\newblock Parallelization, amplification, and exponential time simulation of quantum interactive proof systems.
\newblock In {\em Proceedings of the 32nd Annual ACM Symposium on Theory of Computing}. pages 608--617. 2000.
\newblock \href {https://doi.org/10.1145/335305.335387} {\nolinkurl{doi:10.1145/335305.335387}}. Appearances:\!

\bibitem[LC19]{LC19}
Guang~Hao Low and Isaac~L Chuang.
\newblock Hamiltonian simulation by qubitization.
\newblock {\em Quantum}. 3:163. 2019.
\newblock \href {https://doi.org/10.22331/q-2019-07-12-163} {\nolinkurl{doi:10.22331/q-2019-07-12-163}}. \href {https://arxiv.org/abs/1610.06546} {\nolinkurl{arXiv:1610.06546}}. Appearances:\!

\bibitem[LFKN92]{LFKN92}
Carsten Lund, Lance Fortnow, Howard Karloff, and Noam Nisan.
\newblock Algebraic methods for interactive proof systems.
\newblock {\em Journal of the ACM}. 39(4):859--868. 1992.
\newblock \href {https://doi.org/10.1145/146585.146605} {\nolinkurl{doi:10.1145/146585.146605}}. \prelimVersion{FOCS}{1990}. Appearances:\!

\bibitem[Liu25a]{Liu25}
Yupan Liu.
\newblock {\em Complexity-theoretic perspectives on quantum state testing}.
\newblock PhD thesis. Nagoya University. 2025.
\newblock URL: \url{https://nagoya.repo.nii.ac.jp/records/2012662}. Appearances:\!

\bibitem[Liu25b]{Liu23}
Yupan Liu.
\newblock Quantum state testing beyond the polarizing regime and quantum triangular discrimination.
\newblock {\em computational complexity}. 34(11):1--67. 2025.
\newblock \href {https://doi.org/10.1007/s00037-025-00273-8} {\nolinkurl{doi:10.1007/s00037-025-00273-8}}. \href {https://arxiv.org/abs/2303.01952} {\nolinkurl{arXiv:2303.01952}}. Appearances:\!

\bibitem[LLW26]{LGLW23}
Fran{\c{c}}ois {Le Gall}, Yupan Liu, and Qisheng Wang.
\newblock Space-bounded quantum state testing via space-efficient quantum singular value transformation.
\newblock {\em computational complexity}. 35(2):14:1--14:135. 2026.
\newblock \href {https://doi.org/10.1007/s00037-025-00284-5} {\nolinkurl{doi:10.1007/s00037-025-00284-5}}. \href {https://arxiv.org/abs/2308.05079} {\nolinkurl{arXiv:2308.05079}}. Appearances:\!

\bibitem[LW25]{LW25Lalpha}
Yupan Liu and Qisheng Wang.
\newblock On estimating the quantum $\ell_{\alpha}$ distance.
\newblock In {\em Proceedings of the 33rd Annual European Symposium on Algorithms ({ESA} 2025)}. volume 351 of {\em LIPIcs}. pages 105:1--105:20. Schloss Dagstuhl - Leibniz-Zentrum f{\"{u}}r Informatik. 2025.
\newblock \href {https://doi.org/10.4230/LIPIcs.ESA.2025.105} {\nolinkurl{doi:10.4230/LIPIcs.ESA.2025.105}}. \href {https://arxiv.org/abs/2505.00457} {\nolinkurl{arXiv:2505.00457}}. Appearances:\!

\bibitem[MBD{\etalchar{+}}26]{SH26}
{MIT Hardness Group}, Josh Brunner, Erik~D. Demaine, Fabian Frei, Timothy Gomez, Markus Hecher, Jeffery Li, and Jayson Lynch.
\newblock {ETH} is dead, long live {SH}. 2026.
\newblock \href {https://arxiv.org/abs/2611.02650} {\nolinkurl{arXiv:2611.02650}}. Appearances:\!

\bibitem[MN17]{MN17}
Tomoyuki Morimae and Harumichi Nishimura.
\newblock Merlinization of complexity classes above $\mathsf{BQP}$.
\newblock {\em Quantum Information \& Computation}. 17(11{\&}12):959--972. 2017.
\newblock \href {https://doi.org/10.26421/QIC17.11-12-3} {\nolinkurl{doi:10.26421/QIC17.11-12-3}}. \href {https://arxiv.org/abs/1704.01514} {\nolinkurl{arXiv:1704.01514}}. Appearances:\!

\bibitem[MW05]{MW05}
Chris Marriott and John Watrous.
\newblock Quantum {Arthur--Merlin} games.
\newblock {\em computational complexity}. 14(2):122--152. 2005.
\newblock \href {https://doi.org/10.1007/s00037-005-0194-x} {\nolinkurl{doi:10.1007/s00037-005-0194-x}}. \prelimVersion{CCC}{2004}. \href {https://arxiv.org/abs/cs/0506068} {\nolinkurl{arXiv:cs/0506068}}. Appearances:\!

\bibitem[MY23]{MY23}
Tony Metger and Henry Yuen.
\newblock $\mathsf{stateQIP}=\mathsf{statePSPACE}$.
\newblock In {\em Proceedings of the 64th Annual IEEE Symposium on Foundations of Computer Science}. pages 1349--1356. {IEEE}. 2023.
\newblock \href {https://doi.org/10.1109/FOCS57990.2023.00082} {\nolinkurl{doi:10.1109/FOCS57990.2023.00082}}. \href {https://arxiv.org/abs/2301.07730} {\nolinkurl{arXiv:2301.07730}}. Appearances:\!

\bibitem[NC10]{NC10}
Michael~A Nielsen and Isaac~L Chuang.
\newblock {\em Quantum computation and quantum information}.
\newblock Cambridge University Press. 2010.
\newblock \href {https://doi.org/10.1017/CBO9780511976667} {\nolinkurl{doi:10.1017/CBO9780511976667}}. Appearances:\!

\bibitem[RRR21]{RRR16}
Omer Reingold, Guy~N. Rothblum, and Ron~D. Rothblum.
\newblock Constant-round interactive proofs for delegating computation.
\newblock {\em SIAM Journal on Computing}. 50(3). 2021.
\newblock \href {https://doi.org/10.1137/16M1096773} {\nolinkurl{doi:10.1137/16M1096773}}. \prelimVersion{STOC}{2016}. \href {https://eccc.weizmann.ac.il/report/2016/016} {\nolinkurl{ECCC:TR16-016}}. Appearances:\!

\bibitem[SH90]{SH90}
Richard~E. Stearns and Harry~B. Hunt, III.
\newblock Power indices and easier hard problems.
\newblock {\em Mathematical Systems Theory}. 23(4):209--225. 1990.
\newblock \href {https://doi.org/10.1007/BF02090776} {\nolinkurl{doi:10.1007/BF02090776}}. Appearances:\!

\bibitem[Sha92]{Shamir92}
Adi Shamir.
\newblock $\mathsf{IP}=\mathsf{PSPACE}$.
\newblock {\em Journal of the ACM}. 39(4):869--877. 1992.
\newblock \href {https://doi.org/10.1145/146585.146609} {\nolinkurl{doi:10.1145/146585.146609}}. \prelimVersion{FOCS}{1990}. Appearances:\!

\bibitem[SV03]{SV97}
Amit Sahai and Salil Vadhan.
\newblock A complete problem for statistical zero knowledge.
\newblock {\em Journal of the ACM}. 50(2):196--249. 2003.
\newblock \href {https://doi.org/10.1145/636865.636868} {\nolinkurl{doi:10.1145/636865.636868}}. \prelimVersion{FOCS}{1997}. \href {https://eccc.weizmann.ac.il/report/2000/084} {\nolinkurl{ECCC:TR00-084}}. Appearances:\!

\bibitem[TS13]{TS13}
Amnon Ta-Shma.
\newblock Inverting well conditioned matrices in quantum logspace.
\newblock In {\em Proceedings of the 45th Annual ACM Symposium on Theory of Computing}. pages 881--890. 2013.
\newblock \href {https://doi.org/10.1145/2488608.2488720} {\nolinkurl{doi:10.1145/2488608.2488720}}. Appearances:\!

\bibitem[Uhl76]{Uhlmann76}
Armin Uhlmann.
\newblock The ``transition probability'' in the state space of {$A^*$}-algebra.
\newblock {\em Reports on Mathematical Physics}. 9(2):273--279. 1976.
\newblock \href {https://doi.org/10.1016/0034-4877(76)90060-4} {\nolinkurl{doi:10.1016/0034-4877(76)90060-4}}. Appearances:\!

\bibitem[UNWT25]{UNWT25}
Takeru Utsumi, Yoshifumi Nakata, Qisheng Wang, and Ryuji Takagi.
\newblock Quantum algorithms for {U}hlmann transformation.
\newblock {\em arXiv preprint}. 2025.
\newblock \href {https://arxiv.org/abs/2509.03619} {\nolinkurl{arXiv:2509.03619}}. Appearances:\!

\bibitem[VW16]{VW16}
Thomas Vidick and John Watrous.
\newblock Quantum proofs.
\newblock {\em Foundations and Trends{\textregistered} in Theoretical Computer Science}. 11(1-2):1--215. 2016.
\newblock \href {https://doi.org/10.1561/0400000068} {\nolinkurl{doi:10.1561/0400000068}}. \href {https://arxiv.org/abs/1610.01664} {\nolinkurl{arXiv:1610.01664}}. Appearances:\!

\bibitem[Wat99]{Wat99}
John Watrous.
\newblock Space-bounded quantum complexity.
\newblock {\em Journal of Computer and System Sciences}. 59(2):281--326. 1999.
\newblock \href {https://doi.org/10.1006/jcss.1999.1655} {\nolinkurl{doi:10.1006/jcss.1999.1655}}. \prelimVersion{CCC}{1998}. Appearances:\!

\bibitem[Wat02]{Watrous02}
John Watrous.
\newblock Limits on the power of quantum statistical zero-knowledge.
\newblock In {\em Proceedings of the 43rd Annual IEEE Symposium on Foundations of Computer Science}. pages 459--468. IEEE. 2002.
\newblock \href {https://doi.org/10.1109/SFCS.2002.1181970} {\nolinkurl{doi:10.1109/SFCS.2002.1181970}}. \href {https://arxiv.org/abs/quant-ph/0202111} {\nolinkurl{arXiv:quant-ph/0202111}}. Appearances:\!

\bibitem[Wat09a]{Watrous08}
John Watrous.
\newblock Quantum computational complexity.
\newblock {\em Encyclopedia of Complexity and Systems Science}. pages 7174--7201. 2009.
\newblock \href {https://doi.org/10.1007/978-0-387-30440-3\_428} {\nolinkurl{doi:10.1007/978-0-387-30440-3\_428}}. \href {https://arxiv.org/abs/0804.3401} {\nolinkurl{arXiv:0804.3401}}. Appearances:\!

\bibitem[Wat09b]{Wat09}
John Watrous.
\newblock Zero-knowledge against quantum attacks.
\newblock {\em {SIAM} Journal on Computing}. 39(1):25--58. 2009.
\newblock \href {https://doi.org/10.1137/060670997} {\nolinkurl{doi:10.1137/060670997}}. \prelimVersion{STOC}{2006}. \href {https://arxiv.org/abs/quant-ph/0511020} {\nolinkurl{arXiv:quant-ph/0511020}}. Appearances:\!

\bibitem[Wat18]{Watrous18}
John Watrous.
\newblock {\em The Theory of Quantum Information}.
\newblock Cambridge University Press. 1st edition. 2018.
\newblock \href {https://doi.org/10.1017/9781316848142} {\nolinkurl{doi:10.1017/9781316848142}}. Appearances:\!

\bibitem[WGL{\etalchar{+}}24]{WGL+22}
Qisheng Wang, Ji~Guan, Junyi Liu, Zhicheng Zhang, and Mingsheng Ying.
\newblock New quantum algorithms for computing quantum entropies and distances.
\newblock {\em IEEE Transactions on Information Theory}. 70(8):5653--5680. 2024.
\newblock \href {https://doi.org/10.1109/TIT.2024.3399014} {\nolinkurl{doi:10.1109/TIT.2024.3399014}}. \href {https://arxiv.org/abs/2203.13522} {\nolinkurl{arXiv:2203.13522}}. Appearances:\!

\bibitem[Wil13]{Wilde13}
Mark~M Wilde.
\newblock {\em Quantum Information Theory}.
\newblock Cambridge University Press. 1st edition. 2013.
\newblock \href {https://doi.org/10.1017/9781316809976} {\nolinkurl{doi:10.1017/9781316809976}}. Appearances:\!

\bibitem[{\dutchPrefix{Wolf}{d}}W19]{deWolf19}
Ronald {\dutchPrefix{Wolf}{d}}e~Wolf.
\newblock Quantum computing: Lecture notes.
\newblock {\em arXiv preprint}. 2019.
\newblock \href {https://arxiv.org/abs/1907.09415} {\nolinkurl{arXiv:1907.09415}}. Appearances:\!

\bibitem[WZ24]{WZ23}
Qisheng Wang and Zhicheng Zhang.
\newblock Fast quantum algorithms for trace distance estimation.
\newblock {\em IEEE Transactions on Information Theory}. 70(4):2720--2733. 2024.
\newblock \href {https://doi.org/10.1109/TIT.2023.3321121} {\nolinkurl{doi:10.1109/TIT.2023.3321121}}. \href {https://arxiv.org/abs/2301.06783} {\nolinkurl{arXiv:2301.06783}}. Appearances:\!

\bibitem[WZC{\etalchar{+}}23]{WZC+23}
Qisheng Wang, Zhicheng Zhang, Kean Chen, Ji~Guan, Wang Fang, Junyi Liu, and Mingsheng Ying.
\newblock Quantum algorithm for fidelity estimation.
\newblock {\em IEEE Transactions on Information Theory}. 69(1):273--282. 2023.
\newblock \href {https://doi.org/10.1109/TIT.2022.3203985} {\nolinkurl{doi:10.1109/TIT.2022.3203985}}. \href {https://arxiv.org/abs/2103.09076} {\nolinkurl{arXiv:2103.09076}}. Appearances:\!

\end{thebibliography}

\end{document}